\pdfoutput=1

\documentclass{lmcs}
\usepackage[utf8]{inputenc}
\pdfoutput=1

\usepackage{lastpage}
\lmcsdoi{19}{1}{13}
\lmcsheading{}{\pageref{LastPage}}{}{}%
{Nov.~30,~2022}{Feb.~15,~2023}{}


\usepackage{hyperref}
\theoremstyle{plain} 

\usepackage{amsmath}
\usepackage{amssymb}
\usepackage[T1]{fontenc}
\usepackage[utf8]{inputenc}
\usepackage[mathscr]{eucal}
\usepackage{bm}
\usepackage[all,cmtip,2cell]{xy}
\usepackage{enumerate}
\usepackage{graphicx}
\usepackage{stmaryrd}
\usepackage{rotating}
\usepackage{enumitem}
\usepackage{algorithm}
\usepackage{algorithmic}
\usepackage{hyperref}

\allowdisplaybreaks

\usepackage{makecell}

\newcounter{para}[section]

\usepackage{tikz}
\usetikzlibrary{arrows,decorations.pathmorphing,automata,backgrounds,petri, positioning,calc,cd, automata}
\usetikzlibrary{shapes.geometric,shapes.misc,matrix,arrows.meta,decorations.markings}
\tikzset{x=1em, y=1.5ex, baseline=-0.5ex}
\tikzset{ihbase/.style={inner sep=0,circle,draw,fill=lightgray,minimum size=0.4em,node contents={}}}
\tikzset{ihblack/.style={ihbase,fill=black}}
\tikzset{ihwhite/.style={ihbase,fill=white}}
\tikzset{mat/.style={draw,fill=white,rectangle,node font=\scriptsize}}
\tikzset{ha/.style={mat,rounded rectangle,rounded rectangle left arc=none}}
\tikzset{haop/.style={mat,rounded rectangle,rounded rectangle right arc=none}}
\tikzset{blackha/.style={mat,rounded rectangle,rounded rectangle left arc=none,font=\color{white},fill=black}}
\tikzset{blackhaop/.style={mat,rounded rectangle,rounded rectangle right arc=none,font=\color{white},fill=black}}
\tikzset{anti/.style={inner sep=0,isosceles triangle,fill=black,draw=black, minimum width=0.75em, node contents={}}}
\tikzset{antiop/.style={anti,shape border rotate=180}}
\tikzset{antisq/.style={inner sep=0,rectangle,fill=black, minimum height=1em, minimum width=0.6em, node contents={}}}
\tikzset{count/.style={above,inner ysep=0.15em,font=\scriptsize}}
\tikzset{axiom/.style={above,font=\small}}
\tikzset{dir/.style={-Latex}}
\tikzset{st/.style={decoration={markings,
    mark={at position 0.5 with {\draw (0, 2pt) to (0, -2pt);}}},
    postaction=decorate}}

\newcommand{\objr}{\blacktriangleright}
\newcommand{\objl}{\blacktriangleleft}

\newcommand{\alphabet}{\Sigma}
\newcommand{\Lang}{\mathcal{L}_\alphabet}

\newcommand{\sem}[1]{\left\llbracket{#1}\right\rrbracket}
\newcommand{\semreg}[1]{\left\llbracket{#1}\right\rrbracket_R}
\newcommand{\transreg}[1]{\left\langle{#1}\right\rangle}

\newcommand{\arrowright}{
\tikzset{x=1em, y=2.1ex}
\begin{tikzpicture}
	\begin{pgfonlayer}{nodelayer}
		\node [style=none] (0) at (-1.5, 0) {};
		\node [style=none] (1) at (0, 0) {};
		\node [style=none] (2) at (1.25, 0) {};
	\end{pgfonlayer}
	\begin{pgfonlayer}{edgelayer}
		\draw [->] (0.center) to (1.center);
		\draw (1.center) to (2.center);
	\end{pgfonlayer}
\end{tikzpicture}
}
\tikzset{x=1em, y=1.5ex}
}
\newcommand{\arrowleft}{
\tikzset{x=1em, y=2.1ex}
\begin{tikzpicture}
	\begin{pgfonlayer}{nodelayer}
		\node [style=none] (1) at (-1.25, 0) {};
		\node [style=none] (2) at (0, 0) {};
		\node [style=none] (3) at (1.25, 0) {};
	\end{pgfonlayer}
	\begin{pgfonlayer}{edgelayer}
		\draw (1.center) to (2.center);
		\draw [->] (3.center) to (2.center);
	\end{pgfonlayer}
\end{tikzpicture}
}
\tikzset{x=1em, y=1.5ex}
}

\newcommand{\genericcounitn}[2]{
  \tikz \draw (0, 0) -- node[count] {#2} (1, 0) node[ihbase,#1];
}

\newcommand{\Bcomult}{
\tikzset{x=1em, y=2.1ex}
\InputIfFileExists{lr-copy.tikz}{}{\input{./tikz/lr-copy.tikz}}
\tikzset{x=1em, y=1.5ex}
}

\newcommand{\Bcounit}{
\tikzset{x=1em, y=2.1ex}
\begin{tikzpicture}
	\begin{pgfonlayer}{nodelayer}
		\node [style=black] (37) at (0.75, 0) {};
		\node [style=none] (43) at (0.25, 0) {};
		\node [style=none] (44) at (-0.5, 0) {};
	\end{pgfonlayer}
	\begin{pgfonlayer}{edgelayer}
		\draw (43.center) to (37);
		\draw [->] (44.center) to (43.center);
	\end{pgfonlayer}
\end{tikzpicture}
}
\tikzset{x=1em, y=1.5ex}
}

\newcommand{\Bmult}{
\tikzset{x=1em, y=2.1ex}
\InputIfFileExists{lr-merge.tikz}{}{\input{./tikz/lr-merge.tikz}}
\tikzset{x=1em, y=1.5ex}
}

\newcommand{\Bunit}{
\tikzset{x=1em, y=2.1ex}
\begin{tikzpicture}
	\begin{pgfonlayer}{nodelayer}
		\node [style=black] (37) at (-0.5, 0) {};
		\node [style=none] (43) at (0.25, 0) {};
		\node [style=none] (44) at (0.75, 0) {};
	\end{pgfonlayer}
	\begin{pgfonlayer}{edgelayer}
		\draw [->] (37) to (43.center);
		\draw (44.center) to (43.center);
	\end{pgfonlayer}
\end{tikzpicture}
}
\tikzset{x=1em, y=1.5ex}
}

\newcommand{\Wmult}{
\tikzset{x=1em, y=2.1ex}
\InputIfFileExists{wmult.tikz}{}{\input{./tikz/wmult.tikz}}
\tikzset{x=1em, y=1.5ex}
}

\newcommand{\Wunit}{
\tikzset{x=1em, y=2.1ex}
\begin{tikzpicture}
	\begin{pgfonlayer}{nodelayer}
		\node [style=none] (8) at (0.75, 0) {};
		\node [style=none] (17) at (0.25, 0) {};
		\node [style=white-dot] (23) at (-0.5, 0) {};
	\end{pgfonlayer}
	\begin{pgfonlayer}{edgelayer}
		\draw (17.center) to (8.center);
		\draw [->] (23) to (17.center);
	\end{pgfonlayer}
\end{tikzpicture}
}
\tikzset{x=1em, y=1.5ex}
}

\newcommand{\Wcomult}{
\tikzset{x=1em, y=2.1ex}
\InputIfFileExists{wcomult.tikz}{}{\input{./tikz/wcomult.tikz}}
\tikzset{x=1em, y=1.5ex}
}

\newcommand{\Wcounit}{
\tikzset{x=1em, y=2.1ex}
\begin{tikzpicture}
	\begin{pgfonlayer}{nodelayer}
		\node [style=white-dot] (37) at (0.75, 0) {};
		\node [style=none] (43) at (0.25, 0) {};
		\node [style=none] (44) at (-0.5, 0) {};
	\end{pgfonlayer}
	\begin{pgfonlayer}{edgelayer}
		\draw (43.center) to (37);
		\draw [->] (44.center) to (43.center);
	\end{pgfonlayer}
\end{tikzpicture}
}
\tikzset{x=1em, y=1.5ex}
}
\newcommand{\Wcounitn}[1]{\genericcounitn{ihwhite}}

\newcommand\idzero{
\tikzset{x=1em, y=2.1ex}
\InputIfFileExists{empty-diag.tikz}{}{\input{./tikz/empty-diag.tikz}}
\tikzset{x=1em, y=1.5ex}
} 

\newcommand{\idone}{
\tikzset{x=1em, y=2.1ex}
\begin{tikzpicture}
	\begin{pgfonlayer}{nodelayer}
		\node [style=none] (0) at (-1.5, 0) {};
		\node [style=none] (1) at (0, 0) {};
		\node [style=none] (2) at (1.25, 0) {};
	\end{pgfonlayer}
	\begin{pgfonlayer}{edgelayer}
		\draw [->] (0.center) to (1.center);
		\draw (1.center) to (2.center);
	\end{pgfonlayer}
\end{tikzpicture}
}
\tikzset{x=1em, y=1.5ex}
}

\newcommand{\sym}{
  \tikz {
    \draw (0,  0.5) .. controls (0.5,  0.5) and (0.5, -0.5) .. (1, -0.5);
    \draw (0, -0.5) .. controls (0.5, -0.5) and (0.5,  0.5) .. (1,  0.5);
  }
}

\newcommand\scalar[1]{
  \tikz {
    \node[ha] (ha) {$#1$};
    \draw (ha.west) -- ++(-0.75, 0);
    \draw (ha.east) -- ++(0.75, 0);
  }
}

\tikzset{x=1em, y=1.5ex}
}

\pgfdeclarelayer{edgelayer}
\pgfdeclarelayer{nodelayer}
\pgfsetlayers{edgelayer,nodelayer,main}

\definecolor{light-gray}{gray}{.5}
\tikzstyle{none}=[inner sep=0pt]
\tikzstyle{plain}=[inner sep=0pt]
\tikzstyle{black}=[circle, draw=black, fill=black, inner sep=0pt, minimum size=4pt]
\tikzstyle{black-faded}=[circle, draw=light-gray, fill=light-gray, inner sep=0pt, minimum size=4pt]
\tikzstyle{white-dot}=[circle, draw=black, fill=white, inner sep=0pt, minimum size=4.5pt]
\tikzstyle{white-faded}=[circle, draw=light-gray, fill=white, inner sep=0pt, minimum size=4.5pt]
\tikzstyle{delay}=[fill=black, regular polygon, regular polygon sides=3,rotate=-90, scale=.55]
\tikzstyle{delay-op}=[fill=black, regular polygon, regular polygon sides=3,rotate=90, scale=.55]
\tikzstyle{reg}=[draw, fill=white, rounded rectangle, rounded rectangle left arc=none, minimum height=1.2em, minimum width=1.4em, node font={\scriptsize}]
\tikzstyle{coreg}=[draw, fill=white, rounded rectangle, rounded rectangle right arc=none, minimum height=1.2em, minimum width=1.4em, node font={\scriptsize}]
\tikzstyle{basicb}=[draw, fill=white, rectangle, rounded corners, minimum height=1.2em, minimum width=1.4em]
\tikzstyle{smallb}=[draw, fill=white, rectangle, rounded corners, minimum height=1.2em, minimum width=1.4em, node font={\scriptsize}]
\tikzstyle{rcoreg}=[draw=red, fill=white, rounded rectangle, rounded rectangle right arc=none, minimum height=1.2em, minimum width=1.4em, node font={\scriptsize}]
\tikzstyle{regb}=[draw, fill=black, rounded rectangle, rounded rectangle left arc=none, minimum height=1.2em, minimum width=1.4em, node font={\scriptsize}]
\tikzstyle{regbw}=[draw, left color=black, right color=white, middle color=white, rounded rectangle, rounded rectangle left arc=none, minimum height=1.2em, minimum width=1.4em, node font={\scriptsize}]
\tikzstyle{regwb}=[draw, left color=white, right color=black, middle color=white, rounded rectangle, rounded rectangle left arc=none, minimum height=1.2em, minimum width=1.4em, node font={\scriptsize}]
\tikzstyle{coregb}=[draw, fill=black, rounded rectangle, rounded rectangle right arc=none, minimum height=1.2em, minimum width=1.4em, node font={\scriptsize}]
\tikzstyle{coregbw}=[draw, left color=black, right color=white, middle color=white, rounded rectangle, rounded rectangle right arc=none, minimum height=1.2em, minimum width=1.4em, node font={\scriptsize}]
\tikzstyle{coregwb}=[draw, left color=white, right color=black, middle color=white, rounded rectangle, rounded rectangle right arc=none, minimum height=1.2em, minimum width=1.4em, node font={\scriptsize}]
\tikzstyle{rn}=[circle, draw=red, fill=red, inner sep=0pt, minimum size=4pt]
\tikzstyle{wrn}=[circle, draw=red, fill=white, inner sep=0pt, minimum size=4pt]
\tikzstyle{place}=[circle, draw=black, fill=white, inner sep=0pt, minimum size=9pt]
\tikzstyle{act}=[circle, draw=black, fill=white, inner sep=0pt, minimum size=4.5pt]
\tikzstyle{coact}=[draw, fill=white, rounded rectangle, rounded rectangle right arc=none, minimum height=.7em, minimum width=.9em, node font={\scriptsize}]
\tikzstyle{basic rounded box}=[draw, fill=white, rectangle, rounded corners, minimum height=1.2em, minimum width=1.4em]
\tikzstyle{small rounded box}=[draw, fill=white, rectangle, rounded corners, minimum height=1.2em, minimum width=1.4em, node font={\scriptsize}]
\tikzset{
BWmatrix/.pic={
    \coordinate (center) at (0,0);
    \filldraw[fill=white, draw=black, line width=1pt] (.5,0)
        [rounded corners=14pt] -- (1,0)
        [rounded corners=14pt] -- (1,1)
        [rounded corners=0pt] -- (.5,1)
        [rounded corners=0pt] -- cycle;
    \filldraw[fill=black, draw=black, line width=1pt] (0,0)
        -- (.5,0)
        -- (.5,1)
        -- (0,1)
        -- cycle;
   },
pics/BWmatrix/.default=0.2
}

\tikzstyle{pl}=[circle,thick,draw=black!75,fill=white,minimum size=17pt]
\tikzstyle{port}=[circle, fill,inner sep=1.2pt]
\tikzstyle{transition}=[rectangle,thick,draw=black!75,
  			  fill=black!20,minimum size=7pt]

\tikzstyle{arrow}=[->]

\newcommand{\diagbox}[3]{
\tikzset{x=1em, y=2.1ex}
\begin{tikzpicture}
	\begin{pgfonlayer}{nodelayer}
		\node [style=none] (0) at (-0.75, 0.5) {};
		\node [style=none] (1) at (-0.25, 1) {};
		\node [style=none] (2) at (-0.75, -0.5) {};
		\node [style=none] (3) at (0.75, -0.5) {};
		\node [style=none] (4) at (-0.25, -1) {};
		\node [style=none] (5) at (0.75, 0.5) {};
		\node [style=none] (6) at (2.5, -0) {};
		\node [style=none] (7) at (0.75, -0) {};
		\node [style=none] (8) at (0.25, -1) {};
		\node [style=none] (9) at (0.25, 1) {};
		\node [style=none] (10) at (0, -0) {$#1$};
		\node [style=none] (11) at (2.25, 0.5) {\scriptsize $#3$};
		\node [style=none] (12) at (-2.25, 0.5) {\scriptsize $#2$};
		\node [style=none] (13) at (-2.5, -0) {};
		\node [style=none] (14) at (-0.75, -0) {};
	\end{pgfonlayer}
	\begin{pgfonlayer}{edgelayer}
		\draw [in=180, out=0, looseness=1.25] (7.center) to (6.center);
		\draw [semithick, in=0, out=-90, looseness=1.00] (3.center) to (8.center);
		\draw [semithick, in=-90, out=180, looseness=1.00] (4.center) to (2.center);
		\draw [semithick, in=180, out=90, looseness=1.00] (0.center) to (1.center);
		\draw [semithick, in=90, out=0, looseness=1.00] (9.center) to (5.center);
		\draw [semithick] (1.center) to (9.center);
		\draw [semithick] (5.center) to (3.center);
		\draw [semithick] (8.center) to (4.center);
		\draw [semithick] (2.center) to (0.center);
		\draw [in=180, out=0, looseness=1.25] (13.center) to (14.center);
	\end{pgfonlayer}
\end{tikzpicture}
\tikzset{x=1em, y=1.5ex}
}

\newcommand{\dbox}[1]{
\tikzset{x=1em, y=2.1ex}
\begin{tikzpicture}
	\begin{pgfonlayer}{nodelayer}
		\node [style=none] (6) at (1.5, 0) {};
		\node [style=basicb] (10) at (0, 0) {$#1$};
		\node [style=none] (11) at (-1.25, 0) {};
		\node [style=none] (13) at (-2, 0) {};
		\node [style=none] (14) at (2, 0) {};
	\end{pgfonlayer}
	\begin{pgfonlayer}{edgelayer}
		\draw [->] (13.center) to (11.center);
		\draw (6.center) to (14.center);
		\draw (11.center) to (10.center);
		\draw [->] (10.center) to (6.center);
	\end{pgfonlayer}
\end{tikzpicture}
\tikzset{x=1em, y=1.5ex}
}

\newcommand{\traceaction}[5]{
\tikzset{x=1em, y=2.1ex}
\begin{tikzpicture}
	\begin{pgfonlayer}{nodelayer}
		\node [style=none] (0) at (-0.75, 1.25) {};
		\node [style=none] (1) at (0.75, 1) {};
		\node [style=none] (2) at (-0.25, 1.75) {};
		\node [style=none] (3) at (-0.75, -0.25) {};
		\node [style=none] (4) at (0.75, -0.25) {};
		\node [style=none] (5) at (-0.25, -0.75) {};
		\node [style=none] (6) at (-0.75, 1) {};
		\node [style=none] (7) at (0.75, 1.25) {};
		\node [style=none] (8) at (2.5, 0) {};
		\node [style=none] (9) at (0.75, 0) {};
		\node [style=none] (10) at (0.25, -0.75) {};
		\node [style=none] (11) at (0.25, 1.75) {};
		\node [style=none] (12) at (0, 0.5) {$#1$};
		\node [style=none] (13) at (4, 0.5) {\scriptsize $#3$};
		\node [style=none] (14) at (2.5, 2.5) {};
		\node [style=none] (15) at (-1.5, 2.5) {};
		\node [style=none] (16) at (-2.75, 0.5) {\scriptsize $#2$};
		\node [style=none] (17) at (-1.75, 0) {};
		\node [style=none] (18) at (-0.75, 0) {};
		\node [style=none] (20) at (2.5, 1) {};
		\node [style=none] (21) at (3.5, 2.75) {\scriptsize $#4$};
		\node [style=none] (22) at (-1.5, 1) {};
		\node [style=reg] (23) at (1.75, 1) {$#5$};
		\node [style=none] (24) at (4.25, 0) {};
		\node [style=none] (25) at (-3, 0) {};
	\end{pgfonlayer}
	\begin{pgfonlayer}{edgelayer}
		\draw [semithick, in=0, out=-90] (4.center) to (10.center);
		\draw [semithick, in=-90, out=180] (5.center) to (3.center);
		\draw [semithick, in=180, out=90] (0.center) to (2.center);
		\draw [semithick, in=90, out=0] (11.center) to (7.center);
		\draw [semithick] (2.center) to (11.center);
		\draw [semithick] (7.center) to (4.center);
		\draw [semithick] (10.center) to (5.center);
		\draw [semithick] (3.center) to (0.center);
		\draw (15.center) to (14.center);
		\draw (6.center) to (22.center);
		\draw (1.center) to (20.center);
		\draw [->, bend right=90, looseness=2.25] (20.center) to (14.center);
		\draw [->, bend right=90, looseness=2.25] (15.center) to (22.center);
		\draw (8.center) to (24.center);
		\draw [->] (9.center) to (8.center);
		\draw (17.center) to (18.center);
		\draw [->] (25.center) to (17.center);
	\end{pgfonlayer}
\end{tikzpicture}
\tikzset{x=1em, y=1.5ex}}

\newcommand{\myeq}[1]{\mathrel{\overset{\makebox[0pt]{\mbox{\normalfont\tiny\sffamily (#1)}}}{=}}}
\newcommand{\myleq}[1]{\mathrel{\overset{\makebox[0pt]{\mbox{\normalfont\tiny\sffamily (#1)}}}{\leq}}}
\newcommand{\mygeq}[1]{\mathrel{\overset{\makebox[0pt]{\mbox{\normalfont\tiny\sffamily (#1)}}}{\geq}}}

\newcommand{\N}{\mathbb{N}}

\newcommand{\from}{\mathrel{:}\,}

\newcommand{\Aut}{\mathsf{Aut_{\scriptscriptstyle{\Sigma}}}}

\newcommand{\poi}{\,;\,}

\newcommand{\adjto}{\,\lower1pt\hbox{$\dashv$}\,}

\newcommand{\Rel}{\mathsf{Rel}}

\newcommand{\BProf}{\mathsf{Prof}_\mathbb{B}}

\newcommand{\eqKa}{=_{\scriptscriptstyle KDA}}
\newcommand{\leqKa}{\leq_{\scriptscriptstyle KDA}}
\newcommand{\Syn}{\mathsf{Syn}}
\newcommand{\Sem}{\mathsf{Sem}}
\newcommand{\eqE}[1]{=_{#1}}

\makeatletter
\def\moverlay{\mathpalette\mov@rlay}
\def\mov@rlay#1#2{\leavevmode\vtop{%
\baselineskip\z@skip \lineskiplimit-\maxdimen
\ialign{\hfil$#1##$\hfil\cr#2\crcr}}}
\makeatother

\newcommand\twarr[2]{%
\mathrel{\mathop{\moverlay{\scriptstyle\xrightarrow{\,#1\,}\cr{\lower.2em\hbox{$\scriptstyle{}_{#2}$}}}}}}
\newcommand\twarrw[2]{%
\mathrel{\mathop{\moverlay{\scriptstyle\Longrightarrow\cr{\lower-.6em\hbox{$\scriptstyle{}_{#1}$}}
\cr{\lower.3em\hbox{$\scriptstyle{}_{#2}$}}}}}}

\newcommand{\dtransw}[2]{\raise1pt\hbox{$\;\twarrw{#1}{#2}\;$}}

\newcommand{\diagregexp}[1]{
\begin{tikzpicture}
	\begin{pgfonlayer}{nodelayer}
		\node [style=none] (0) at (1.5, 0) {};
		\node [style=rcoreg] (1) at (0, 0) {{\color{red} $e$}};
	\end{pgfonlayer}
	\begin{pgfonlayer}{edgelayer}
		\draw [red] (1) to (0.center);
	\end{pgfonlayer}
\end{tikzpicture}}

\newcommand{\smalldiag}[1]{\begin{tikzpicture}
	\begin{pgfonlayer}{nodelayer}
		\node [style=none] (0) at (-0.5, 0.5) {};
		\node [style=none] (1) at (-0.25, 0.75) {};
		\node [style=none] (2) at (-0.5, -0.5) {};
		\node [style=none] (3) at (0.5, -0.5) {};
		\node [style=none] (4) at (-0.25, -0.75) {};
		\node [style=none] (5) at (0.5, 0.5) {};
		\node [style=none] (6) at (1.5, 0) {};
		\node [style=none] (7) at (0.5, 0) {};
		\node [style=none] (8) at (0.25, -0.75) {};
		\node [style=none] (9) at (0.25, 0.75) {};
		\node [style=none] (10) at (0, 0) {$#1$};
		\node [style=none] (11) at (-1.5, 0) {};
		\node [style=none] (12) at (-0.5, 0) {};
	\end{pgfonlayer}
	\begin{pgfonlayer}{edgelayer}
		\draw [in=180, out=0, looseness=1.25] (7.center) to (6.center);
		\draw [semithick, in=0, out=-90] (3.center) to (8.center);
		\draw [semithick, in=-90, out=180] (4.center) to (2.center);
		\draw [semithick, in=180, out=90] (0.center) to (1.center);
		\draw [semithick, in=90, out=0] (9.center) to (5.center);
		\draw [semithick] (1.center) to (9.center);
		\draw [semithick] (5.center) to (3.center);
		\draw [semithick] (8.center) to (4.center);
		\draw [semithick] (2.center) to (0.center);
		\draw [in=180, out=0, looseness=1.25] (11.center) to (12.center);
	\end{pgfonlayer}
\end{tikzpicture}
}

\begin{document}

\title[A Finite Axiomatisation of Finite-State Automata Using String Diagrams]{A Finite Axiomatisation of Finite-State Automata \texorpdfstring{\\}{} Using String Diagrams}

\author[R.~Piedeleu]{Robin Piedeleu\lmcsorcid{0000-0002-3945-2704}}[a]
\address{\lsuper{a}University College London, United Kingdom}	
\email{\{r.piedeleu,f.zanasi\}@ucl.ac.uk}  

\author[F.~Zanasi]{Fabio Zanasi\lmcsorcid{0000-0001-6457-1345}}[a,b]	
\address{\lsuper{b}University of Bologna, Italy}	

\begin{abstract}
	\noindent We develop a fully diagrammatic approach to finite-state automata, based on reinterpreting their usual state-transition graphical representation as a two-dimensional syntax of string diagrams. In this setting, we are able to provide a complete equational theory for language equivalence, with two notable features. First, the proposed axiomatisation is finite. Second, the Kleene star is a derived concept, as it can be decomposed into more primitive algebraic blocks.
\end{abstract}

\maketitle

\section{Introduction}\label{S:one}

Finite-state automata are one of the most studied structures in theoretical computer science, with an illustrious history and roots reaching far beyond, in the work of biologists, psychologists, engineers and mathematicians. Kleene~\cite{kleene1951representation} introduced regular expressions to give finite-state automata an algebraic presentation, motivated by the study of (biological) neural networks~\cite{mcculloch1943logical}. They are the terms freely generated by the following grammar:
\begin{equation*}
e, f ::= e + f \mid ef \mid e^* \mid 0 \mid 1 \mid a\in A
\end{equation*}
Equational properties of regular expressions were studied by Conway~\cite{conway2012regular} who introduced the term \emph{Kleene algebra}: this is an idempotent semiring with an operation $(-)^*$ for iteration, called the (Kleene) star. The equational theory of Kleene algebra is now well-understood, and multiple complete axiomatisations, both for language and relational models, have been given. Crucially, Kleene algebra is not finitely-based: no finite equational theory can appropriately capture the behaviour of the star~\cite{redko1964defining}. Instead, there are purely  equational infinitary axiomatisations~\cite{krob1991complete,bloom1993equational} and finitary implicational theories, like that of Kozen~\cite{kozen1994completeness}.

Since then, much research has been devoted to extending Kleene algebra with operations capturing richer patterns of behaviour, useful in program verification. Examples include conditional branching (Kleene algebra with tests~\cite{kozen1997kleene}, and its recent guarded version~\cite{smolka2019guarded}),  concurrent computation (CKA~\cite{hoare2009concurrent,KappeB0Z18}), and specification of message-passing behaviour in networks (NetKAT~\cite{anderson2014netkat}).

The meta-theory of the formalisms above relies on a three step methodology: (1) given an operational model (e.g., finite-state automata), (2) devise a syntax (regular expressions) that is sufficiently expressive to capture the class of behaviours of the operational model (regular languages), and (3) find a complete axiomatisation (Kleene algebra) for the given semantics.

In this paper, we open up a direct path from (1) to (3). Instead of thinking of automata as a combinatorial model, we formalise them as a bona-fide (two-dimensional) syntax, using the well-established mathematical theory of \emph{string diagrams} for monoidal categories~\cite{Selinger2009}. This approach lets us axiomatise the behaviour of automata directly, freeing us from the necessity of compressing them down to a one-dimensional notation like regular expressions.

This perspective not only sheds new light on a venerable topic, but has significant consequences. First, as our most important contribution, we are able to provide a \emph{finite and purely equational} axiomatisation of finite-state automata, up to language equivalence. This does not contradict the impossibility of finding a finite basis for Kleene algebra, as the algebraic setting is different: our result gives a finite presentation as a symmetric monoidal category, while the impossibility result prevents any such presentation to exist as an algebraic theory (in the standard sense). In other words, there is no finite axiomatisation based on terms (\emph{tree}-like structures), but we demonstrate that there is one based on string diagrams (\emph{graph}-like structures).

Secondly, embracing the two-dimensional nature of automata guarantees a strong form of compositionality that the one-dimensional syntax of regular expressions does not have. In the string diagrammatic setting, automata may have multiple inputs and outputs and, as a result, can be decomposed into subcomponents that retain a meaningful interpretation. For example, if we split the automata below left, the resulting components are still valid string diagrams within our syntax, below right:
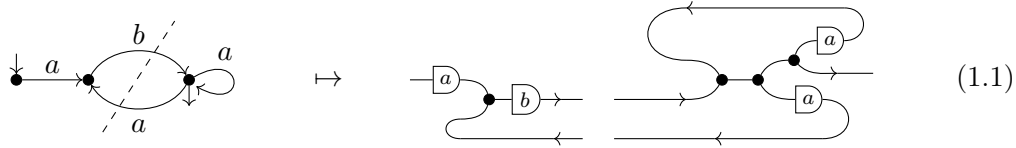
\begin{equation}\label{ex:decompose-automaton}

\tikzset{x=1em, y=2.1ex}
\InputIfFileExists{ex-automaton-graph-split.tikz}{}{\input{./tikz/ex-automaton-graph-split.tikz}}
\tikzset{x=1em, y=1.5ex}
 \qquad\mapsto\qquad 
\tikzset{x=1em, y=2.1ex}
\InputIfFileExists{ex-automaton-diagram-multiple.tikz}{}{\input{./tikz/ex-automaton-diagram-multiple.tikz}}
\tikzset{x=1em, y=1.5ex}
\quad 
\tikzset{x=1em, y=2.1ex}
\InputIfFileExists{ex-automaton-diagram-multiple-1.tikz}{}{\input{./tikz/ex-automaton-diagram-multiple-1.tikz}}
\tikzset{x=1em, y=1.5ex}

\end{equation}
In line with the compositional approach, it is significant that the Kleene star can be decomposed into more elementary building blocks (which come together to form a feedback loop):
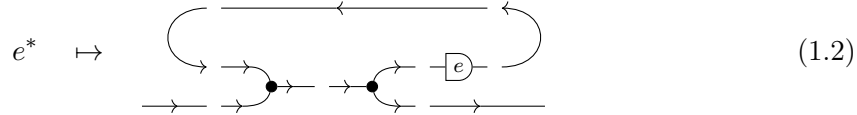
\begin{equation}\label{eq:star-decomposed}
e^* \quad \mapsto\quad 
\tikzset{x=1em, y=2.1ex}
\InputIfFileExists{star-decomposed.tikz}{}{\input{./tikz/star-decomposed.tikz}}
\tikzset{x=1em, y=1.5ex}

\end{equation}

\noindent This opens up for interesting possibilities when studying extensions of Kleene algebra within the same approach---we elaborate on this in Section~\ref{sec:conclusion}.

Finally, we believe our proof of completeness is of independent interest, as it relies on fully diagrammatic reformulation of Brzozowski's minimisation algorithm~\cite{brzozowski1962canonical}. In the string diagrammatic setting, the symmetries of the equational theory give this procedure a particularly elegant and simple form. Because all of the axioms involved in the determinisation procedure come with a dual, a co-determinisation procedure can be defined immediately by simply reversing the former. This reduces the proof of completeness to a proof that determinisation can be performed diagrammatically. Moreover, note that our completeness proof goes through a richer language, with additional algebraic operations that are adjoint (in a sense that we explain in Section~\ref{sec:axioms}) to those that allow us to express standard automata. Within this extended language, we are able to derive a completeness result for (diagrams corresponding to) automata, but leave open the completeness of the full language. We will come back to this point in Section~\ref{sec:conclusion}.

\medskip

This is not the first time that automata and regular languages are recast into a categorical mould. The \emph{iteration theories}~\cite{bloom1993iteration} of Bloom and {\'E}sik, \emph{sharing graphs}~\cite{Hasegawa97recursionfrom} of Hasegawa or \emph{network algebras}~\cite{stefanescu2000network} of Stefanescu are all categorical frameworks designed to reason about iteration or recursion, that have found fruitful applications in this domain. They are based on a notion of parameterised fixed-point  which defines a categorical \emph{trace} in the sense of~\cite{Joyal_tracedcategories}. While our proposal bears resemblance to (and is inspired by) this prior work, it goes beyond in one fundamental aspect: it is the first to give a \emph{finite} complete axiomatisation of automata up to language equivalence.

A second difference is methodological: our syntax does not feature any primitive for iteration or recursion. In particular, the star is a derived concept, in the sense that it is decomposable into more elementary operations~\eqref{eq:star-decomposed}. Categorically, our starting point is a compact-closed rather than  traced category.

We elaborate further on the relation between ours and existing work in Section~\ref{sec:conclusion}.

\paragraph{Conference version.} This work is based on the conference paper~\cite{piedeleu2021string}. It amends a mistaken completeness claim made in that paper and proposes a new approach to the same question, based on a different syntax. We detail the relationship between the two papers in Section~\ref{sec:conclusion} and where clarifications are necessary in the main text.

\section{Syntax and semantics}\label{sec:syntax-semantics}

Following a standard methodology (recalled in Appendix~\ref{app:method}), we will define two symmetric monoidal categories (SMCs), one serving as syntax, the other as semantics. Moreover, to guarantee a compositional interpretation, we will define a symmetric monoidal functor between the two.

\paragraph{Syntax.}
We fix an alphabet $\Sigma$. We call $\Syn$ the strict SMC freely generated by the following objects and morphisms:
\begin{itemize}
\item two generating objects $\objr$ (`right') and $\objl$ (`left') with their identity morphisms depicted respectively as $\arrowright$ and $\arrowleft$.
\item the following generating morphisms, depicted as \emph{string diagrams}~\cite{Selinger2009}:
\begin{equation}%
\label{eq:syn1}

\tikzset{x=1em, y=2.1ex}
\InputIfFileExists{lr-copy.tikz}{}{\input{./tikz/lr-copy.tikz}}
\tikzset{x=1em, y=1.5ex}
\quad
\tikzset{x=1em, y=2.1ex}
\begin{tikzpicture}
	\begin{pgfonlayer}{nodelayer}
		\node [style=black] (37) at (0.75, 0) {};
		\node [style=none] (43) at (0.25, 0) {};
		\node [style=none] (44) at (-0.5, 0) {};
	\end{pgfonlayer}
	\begin{pgfonlayer}{edgelayer}
		\draw (43.center) to (37);
		\draw [->] (44.center) to (43.center);
	\end{pgfonlayer}
\end{tikzpicture}
}
\tikzset{x=1em, y=1.5ex}
\quad
\tikzset{x=1em, y=2.1ex}
\InputIfFileExists{lr-merge.tikz}{}{\input{./tikz/lr-merge.tikz}}
\tikzset{x=1em, y=1.5ex}
\quad
\tikzset{x=1em, y=2.1ex}
\begin{tikzpicture}
	\begin{pgfonlayer}{nodelayer}
		\node [style=black] (37) at (-0.5, 0) {};
		\node [style=none] (43) at (0.25, 0) {};
		\node [style=none] (44) at (0.75, 0) {};
	\end{pgfonlayer}
	\begin{pgfonlayer}{edgelayer}
		\draw [->] (37) to (43.center);
		\draw (44.center) to (43.center);
	\end{pgfonlayer}
\end{tikzpicture}
}
\tikzset{x=1em, y=1.5ex}
 \quad
\tikzset{x=1em, y=2.1ex}
\InputIfFileExists{cap-down.tikz}{}{\input{./tikz/cap-down.tikz}}
\tikzset{x=1em, y=1.5ex}
 \quad
\tikzset{x=1em, y=2.1ex}
\InputIfFileExists{cup-down.tikz}{}{\input{./tikz/cup-down.tikz}}
\tikzset{x=1em, y=1.5ex}
\quad \scalar{a} \quad (a\in \Sigma)
\end{equation}
\begin{equation}%
\label{eq:white-gen}
\Wmult \quad \Wunit\quad  \Wcomult \quad \Wcounit
\end{equation}
\end{itemize}

\paragraph{Semantics.} We first define the semantics for string diagrams simply as a mapping from the set of generators to relations between tuples of languages over $\Sigma$, and then discuss how to extend it to a functor from $\Syn$ to a category that we will define below. Let $\Lang := 2^{\Sigma^\star}$. A diagram with $m$ ports on the left and $n$ on the right, is interpreted as a subset of $\Lang^m\times\Lang^n$.
\begin{gather}%
\label{def:lr-copy} \nonumber
\sem{
\tikzset{x=1em, y=2.1ex}
\InputIfFileExists{lr-copy.tikz}{}{\input{./tikz/lr-copy.tikz}}
\tikzset{x=1em, y=1.5ex}
} = \left\{\big(L, (K_1,K_2)\big)\mid L\subseteq K_i,\, i=1,2 \text{ and } L,K_1, K_2 \in\Lang \right\}
\\%
\label{def:lr-merge} \nonumber
\sem{
\tikzset{x=1em, y=2.1ex}
\InputIfFileExists{lr-merge.tikz}{}{\input{./tikz/lr-merge.tikz}}
\tikzset{x=1em, y=1.5ex}
} = \left\{\big((L_1,L_2), K\big)\mid L_i\subseteq K,\, i=1,2 \text{ and } L_1, L_2, K \in\Lang \right\}
\\%
\label{def:cup} \nonumber
\quad \sem{
\tikzset{x=1em, y=2.1ex}
\begin{tikzpicture}
	\begin{pgfonlayer}{nodelayer}
		\node [style=black] (37) at (0.75, 0) {};
		\node [style=none] (43) at (0.25, 0) {};
		\node [style=none] (44) at (-0.5, 0) {};
	\end{pgfonlayer}
	\begin{pgfonlayer}{edgelayer}
		\draw (43.center) to (37);
		\draw [->] (44.center) to (43.center);
	\end{pgfonlayer}
\end{tikzpicture}
}
\tikzset{x=1em, y=1.5ex}
\;} = \left\{(L, \bullet)\mid L \in\Lang \right\} \qquad
\sem{
\tikzset{x=1em, y=2.1ex}
\InputIfFileExists{cup-down.tikz}{}{\input{./tikz/cup-down.tikz}}
\tikzset{x=1em, y=1.5ex}
} = \left\{(\bullet, (L,K)) \mid L\subseteq K \mid L,K \in\Lang \right\}
\\%
\label{def:cap} \nonumber
\quad \sem{
\tikzset{x=1em, y=2.1ex}
\begin{tikzpicture}
	\begin{pgfonlayer}{nodelayer}
		\node [style=black] (37) at (-0.5, 0) {};
		\node [style=none] (43) at (0.25, 0) {};
		\node [style=none] (44) at (0.75, 0) {};
	\end{pgfonlayer}
	\begin{pgfonlayer}{edgelayer}
		\draw [->] (37) to (43.center);
		\draw (44.center) to (43.center);
	\end{pgfonlayer}
\end{tikzpicture}
}
\tikzset{x=1em, y=1.5ex}
\;} = \left\{( \bullet, K)\mid K \in\Lang \right\} \qquad
\sem{
\tikzset{x=1em, y=2.1ex}
\InputIfFileExists{cap-down.tikz}{}{\input{./tikz/cap-down.tikz}}
\tikzset{x=1em, y=1.5ex}
} = \left\{((L,K), \bullet) \mid K\subseteq L ,\; L,K \in\Lang \right\}
\\
\nonumber
\sem{\scalar{a}} = \left\{(L, K)\mid La\subseteq K,\; L,K \in\Lang \right\} \qquad (a\in\Sigma)\\%
\label{def:blackid} \nonumber
\sem{\arrowright} = \left\{(L, K)\mid L\subseteq K,\; L,K \in\Lang \right\}\\ \nonumber \sem{\arrowleft} = \left\{(L, K)\mid  K\subseteq L,\; L,K \in\Lang \right\}
\\%
\label{eq:intersection}
\sem{\Wmult} = \{((L_1,L_2), K) \mid L_1\cap L_2\subseteq K,\; L_1,L_2,K\in\Lang \} \qquad \sem{\Wunit} = \left\{\left(\bullet, \Sigma^\star\right)\right\}
\\%
\label{eq:union}
\sem{\Wcomult} = \{(L,(K_1,K_2)) \mid L \subseteq K_1\cup K_2,\; L,K_1,K_2\in\Lang\} \qquad \sem{\Wcounit} = \left\{\left(\varnothing, \bullet\right)\right\}
\end{gather}
In a nutshell, the generating diagrams denote operations that relate tuples of languages: $
\tikzset{x=1em, y=2.1ex}
\InputIfFileExists{lr-copy.tikz}{}{\input{./tikz/lr-copy.tikz}}
\tikzset{x=1em, y=1.5ex}
$ represents copying, $
\tikzset{x=1em, y=2.1ex}
\begin{tikzpicture}
	\begin{pgfonlayer}{nodelayer}
		\node [style=black] (37) at (0.75, 0) {};
		\node [style=none] (43) at (0.25, 0) {};
		\node [style=none] (44) at (-0.5, 0) {};
	\end{pgfonlayer}
	\begin{pgfonlayer}{edgelayer}
		\draw (43.center) to (37);
		\draw [->] (44.center) to (43.center);
	\end{pgfonlayer}
\end{tikzpicture}
}
\tikzset{x=1em, y=1.5ex}
$ discarding, $
\tikzset{x=1em, y=2.1ex}
\InputIfFileExists{cup-down.tikz}{}{\input{./tikz/cup-down.tikz}}
\tikzset{x=1em, y=1.5ex}
$ and $
\tikzset{x=1em, y=2.1ex}
\InputIfFileExists{cap-down.tikz}{}{\input{./tikz/cap-down.tikz}}
\tikzset{x=1em, y=1.5ex}
$ feed back outputs into inputs and vice-versa, and $\scalar{a}$ represents the action of each letter of $\Sigma$ on the set of languages, by concatenation on the right. These are the generators that allow us to encode automata, as we will see in Section~\ref{sec:encoding}. The other generators, $\Wmult,\Wunit, \Wcomult, \Wcounit$, represent operations that are adjoint to their black counterparts, in a sense that we will explain in Section~\ref{sec:axioms}. The directed syntax highlights the dual roles played by the two generating objects, representing inclusion and reverse inclusion of languages respectively.

\begin{rem}
A word of warning: in the conference paper~\cite{piedeleu2021string} on which this work is based, we use a different three sorted syntax, with an additional red wire denoting the set of regular expressions. These acted on the set of languages via a white node whose appearance is reminiscent of the white nodes of the syntax above. However, the reader should not confuse them---their semantics are totally unrelated and, as we will see, so are their equational properties.
\end{rem}

In order for the mapping $\sem{\cdot}$ to be functorial, we now introduce a suitable target SMC for the semantics. Interestingly, this will not be the category $\Rel$ of sets and relations: indeed, the identity morphisms $\arrowright$ and $\arrowleft$ are not interpreted as identities of $\Rel$ (since they denote the order on the set of languages). Instead, the semantic domain will be the category $\BProf$ of \emph{Boolean(-enriched) profunctors}~\cite{fong2018seven} (also variously called relational profunctors~\cite{hyland2003glueing} or weakening relations~\cite{moshier2015coherence} in the literature).

\begin{defi}\label{def:bool-prof}
Given two preorders $(X, \leq_X)$ and $(Y, \leq_Y)$, a \emph{Boolean profunctor} $R\from X \to Y$ is a relation $R\subseteq X \times Y$ such that if $(x,y)\in R \text{ and } x'\leq_X x,\; y\leq_Y y' \text{ then } (x',y')\in R$.
\end{defi}
Preorders and Boolean profunctors form a SMC $\BProf$ with composition given by relational composition. The identity for an object $(X, \leq_X)$ is the order relation $\leq_X$ itself. The monoidal product is the usual product of preorders, where $(x,y)\leq (x',y')$ iff $x\leq_X x'$ and $y\leq_Y y'$. For more details on Boolean profunctors, including applications to engineering design, we refer the reader to~\cite[Chapter 4]{fong2018seven}.
Since relations can be ordered by inclusion in a way that is compatible with composition, they form a \emph{bicategory}, and so does $\BProf$. The bicategorical structure is rather simple as, for any two morphisms, the set of 2-cells between them forms a partial order. This also means that $\BProf$ is an \emph{order-enriched} (1-)category. In fact, it is a \emph{Cartesian bicategory}~\cite{Carboni1987}.
\begin{rem}\label{rmk:monotone-maps}
Note that every monotone map $f\from (X,\leq_X)\to (Y,\leq_Y)$ can be turned into a monotone relations in two different ways: the relation $Yf := \{(x,y)\mid f(x)\leq y\}$ with type $Yf\from (X,\leq_X)\to (Y,\leq_Y)$ and $Y^{op}f := \{(y,x)\mid y\leq f(x)\}$ with type $(Y,\leq_Y)\to (X,\leq_X)$. This implies that $\BProf$ contains (two different copies of) the SMC of monotone maps as a monoidal sub-category.
\end{rem}
The features of our diagrammatic syntax reflect the rich structure of the profunctor semantics. Indeed, the order relation is built into the wires $\arrowright$ and $\arrowleft$. The two possible directions represent the identities on the set of languages ordered by inclusion, and on the same set equipped with the reverse order, respectively.
\begin{prop}
$\sem{\cdot}$ defines a symmetric monoidal functor of type $\Syn \to \BProf$.
\end{prop}
\begin{proof}
It suffices to check that the interpretation of all generators define Boolean profunctors. It is clear that all generators satisfy the condition of Definition~\ref{def:bool-prof}. For example, the action generator $\scalar{a}$ is a Boolean profunctor: if $(L,K)$ are such that $La\subseteq K$ and, moreover we have $L'\subseteq L$ and $K\subseteq K'$, then $L'a\subseteq La\subseteq K\subseteq K'$ by monotony of concatenation of languages.
\end{proof}
In particular, because $\Syn$ is free, we can unambiguously assign meaning to any composite diagram from the semantics of its components using composition and the monoidal product in $\BProf$:
\begin{align*}
\sem{c\poi d} =\sem{
\tikzset{x=1em, y=2.1ex}
\InputIfFileExists{seq-compose.tikz}{}{\input{./tikz/seq-compose.tikz}}
\tikzset{x=1em, y=1.5ex}
} &= \left\{(L,K)\mid \exists M \, (L,M)\in\sem{\smalldiag{c}}, (M,K)\in\sem{\smalldiag{{\scriptstyle d}}}\right\}\\\sem{c_1\oplus c_2} = \sem{\,
\tikzset{x=1em, y=2.1ex}
\InputIfFileExists{par-compose.tikz}{}{\input{./tikz/par-compose.tikz}}
\tikzset{x=1em, y=1.5ex}
\,} &= \left\{\big((L_1,L_2),(K_1, K_2)\big)\mid (L_i,K_i)\in\sem{\smalldiag{{\scriptstyle c_i}}},\, i=1,2\right\}
\end{align*}
Single wires labelled by a list $X$ of generating objects (here, $\objr$ and $\objl$) represent $|X|$ parallel ordered wires, labelled from top to bottom with the elements of $X$.

\begin{exa}\label{ex:star-semantics}
We include here a worked out example to show how to compute the behaviour of a composite diagram which, as we will see, represents (the action by concatenation of) the regular language $a^* = \{\epsilon, a, aa, \dots\}$. To reason about complex diagrams, it is easier to assign variable names to each wire: let us call $N$ to the top wire of the feedback loop, and $M$ to the middle wire joining $
\tikzset{x=1em, y=2.1ex}
\InputIfFileExists{lr-merge.tikz}{}{\input{./tikz/lr-merge.tikz}}
\tikzset{x=1em, y=1.5ex}
$ to $
\tikzset{x=1em, y=2.1ex}
\InputIfFileExists{lr-copy.tikz}{}{\input{./tikz/lr-copy.tikz}}
\tikzset{x=1em, y=1.5ex}
$. After simplifying a few redundant constraints, we get
\begin{equation*}
\sem{
\tikzset{x=1em, y=2.1ex}
\InputIfFileExists{star-ex.tikz}{}{\input{./tikz/star-ex.tikz}}
\tikzset{x=1em, y=1.5ex}
} \begin{array}{l}
=  \{(L,K)\mid \exists M.N.\; \, L, N\subseteq M,\; Ma \subseteq N\,, M\subseteq K\}\\
=  \{(L,K)\mid \exists M.\; \, L \subseteq M,\; Ma \subseteq M\, M\subseteq K\}\\
=  \{(L,K)\mid \exists M.\; \,  L\cup Ma\subseteq M\; M\subseteq K\}.
\end{array}
\end{equation*}
Call this diagram $d$.
By Arden's lemma~\cite{arden1961delayed}, $La^*$ is the \emph{smallest} solution of the language inequality $L\cup Ma\subseteq M$; thus $\exists M.\; \,  L\cup Ma\subseteq M \Leftrightarrow \exists M.\; \,  La^* \subseteq M$ and
\[ \sem{d} =  \{(L,K)\mid \exists M.\; \,  La^* \subseteq M,\; M\subseteq K\} =  \{(L,K)\mid\,  La^* \subseteq K\}\text{.}\]
\end{exa}

\section{Inequational theory}\label{sec:axioms}

\begin{figure}
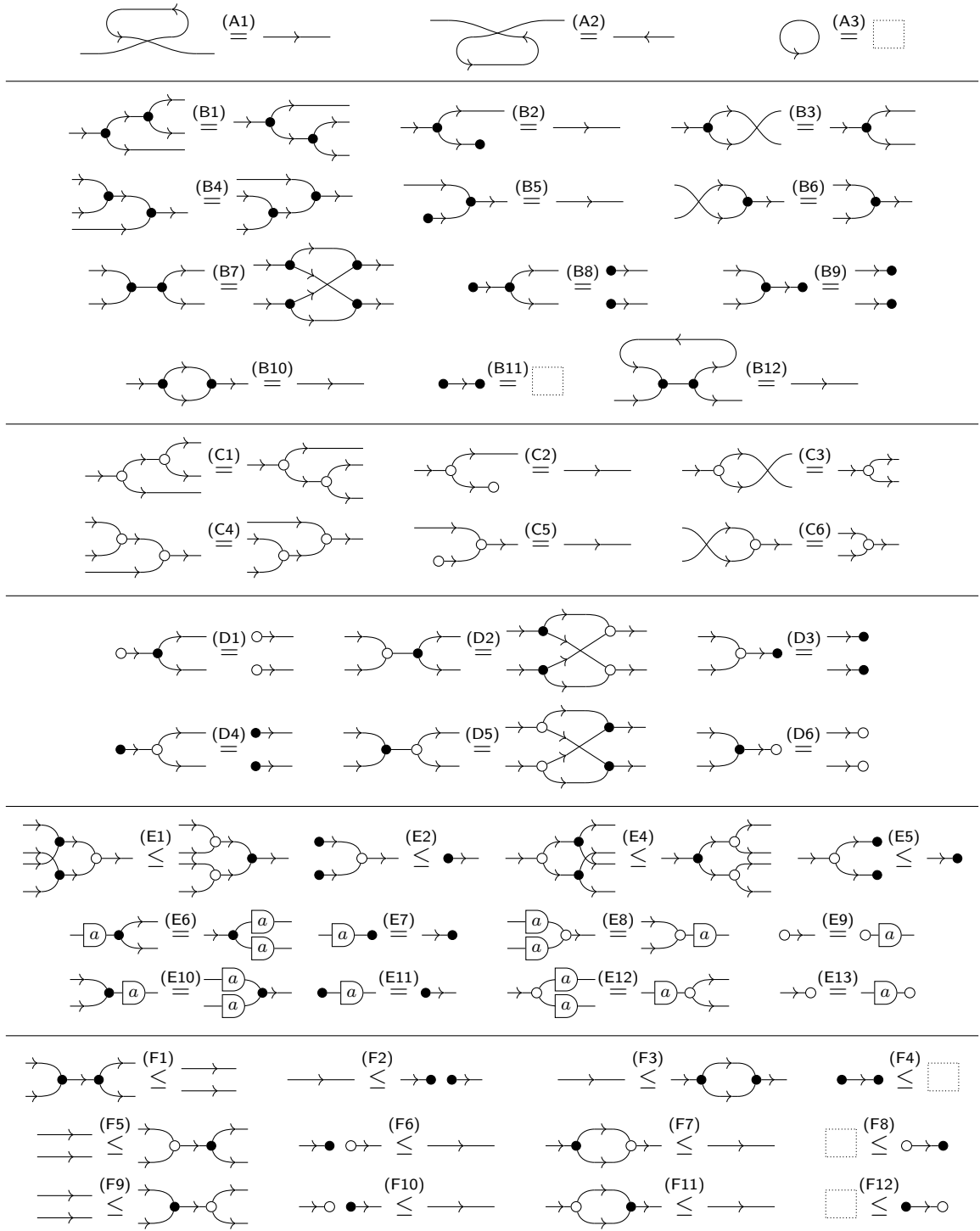

\vspace{-1cm}
\begin{equation*}\label{ax:compact-closed}

\tikzset{x=1em, y=2.1ex}
\InputIfFileExists{yanking-bis.tikz}{}{\input{./tikz/yanking-bis.tikz}}
\tikzset{x=1em, y=1.5ex}
\;\myeq{A1}\;\arrowright\quad \qquad \quad
\tikzset{x=1em, y=2.1ex}
\InputIfFileExists{yanking.tikz}{}{\input{./tikz/yanking.tikz}}
\tikzset{x=1em, y=1.5ex}
\;\myeq{A2}\;\arrowleft\quad \qquad \quad 
\tikzset{x=1em, y=2.1ex}
\InputIfFileExists{loop.tikz}{}{\input{./tikz/loop.tikz}}
\tikzset{x=1em, y=1.5ex}
\;\myeq{A3}\; 
\tikzset{x=1em, y=2.1ex}
\InputIfFileExists{empty-diag.tikz}{}{\input{./tikz/empty-diag.tikz}}
\tikzset{x=1em, y=1.5ex}

\end{equation*}

\hrule

\begin{equation*}\label{ax:rel}

\tikzset{x=1em, y=2.1ex}
\InputIfFileExists{co-associativity.tikz}{}{\input{./tikz/co-associativity.tikz}}
\tikzset{x=1em, y=1.5ex}
 \; \myeq{B1} \; 
\tikzset{x=1em, y=2.1ex}
\InputIfFileExists{co-associativity-1.tikz}{}{\input{./tikz/co-associativity-1.tikz}}
\tikzset{x=1em, y=1.5ex}
 \qquad 
\tikzset{x=1em, y=2.1ex}
\InputIfFileExists{right-co-unitality.tikz}{}{\input{./tikz/right-co-unitality.tikz}}
\tikzset{x=1em, y=1.5ex}
 \; \myeq{B2} \; \arrowright  \qquad 
\tikzset{x=1em, y=2.1ex}
\InputIfFileExists{co-commutativity.tikz}{}{\input{./tikz/co-commutativity.tikz}}
\tikzset{x=1em, y=1.5ex}
\; \myeq{B3}\; 
\tikzset{x=1em, y=2.1ex}
\InputIfFileExists{large-copy.tikz}{}{\input{./tikz/large-copy.tikz}}
\tikzset{x=1em, y=1.5ex}

\end{equation*}
\begin{equation*}

\tikzset{x=1em, y=2.1ex}
\InputIfFileExists{associativity.tikz}{}{\input{./tikz/associativity.tikz}}
\tikzset{x=1em, y=1.5ex}
 \; \myeq{B4} \; 
\tikzset{x=1em, y=2.1ex}
\InputIfFileExists{associativity-1.tikz}{}{\input{./tikz/associativity-1.tikz}}
\tikzset{x=1em, y=1.5ex}
 \qquad  
\tikzset{x=1em, y=2.1ex}
\InputIfFileExists{right-unitality.tikz}{}{\input{./tikz/right-unitality.tikz}}
\tikzset{x=1em, y=1.5ex}
\; \myeq{B5} \; \arrowright \qquad 
\tikzset{x=1em, y=2.1ex}
\InputIfFileExists{commutativity.tikz}{}{\input{./tikz/commutativity.tikz}}
\tikzset{x=1em, y=1.5ex}
\; \myeq{B6}\; 
\tikzset{x=1em, y=2.1ex}
\InputIfFileExists{large-merge.tikz}{}{\input{./tikz/large-merge.tikz}}
\tikzset{x=1em, y=1.5ex}

\end{equation*}
\begin{equation*}

\tikzset{x=1em, y=2.1ex}
\InputIfFileExists{bimonoid.tikz}{}{\input{./tikz/bimonoid.tikz}}
\tikzset{x=1em, y=1.5ex}
\; \myeq{B7} \;
\tikzset{x=1em, y=2.1ex}
\InputIfFileExists{bimonoid-1.tikz}{}{\input{./tikz/bimonoid-1.tikz}}
\tikzset{x=1em, y=1.5ex}
 \qquad \quad
\tikzset{x=1em, y=2.1ex}
\InputIfFileExists{copy-co-delete.tikz}{}{\input{./tikz/copy-co-delete.tikz}}
\tikzset{x=1em, y=1.5ex}
\; \myeq{B8} \;
\tikzset{x=1em, y=2.1ex}
\InputIfFileExists{copy-co-delete-1.tikz}{}{\input{./tikz/copy-co-delete-1.tikz}}
\tikzset{x=1em, y=1.5ex}
\qquad\quad
\tikzset{x=1em, y=2.1ex}
\InputIfFileExists{merge-delete.tikz}{}{\input{./tikz/merge-delete.tikz}}
\tikzset{x=1em, y=1.5ex}
\; \myeq{B9} \;
\tikzset{x=1em, y=2.1ex}
\InputIfFileExists{merge-delete-1.tikz}{}{\input{./tikz/merge-delete-1.tikz}}
\tikzset{x=1em, y=1.5ex}

\end{equation*}
\begin{equation*}
 
\tikzset{x=1em, y=2.1ex}
\InputIfFileExists{idempotence.tikz}{}{\input{./tikz/idempotence.tikz}}
\tikzset{x=1em, y=1.5ex}
\; \myeq{B10} \;\arrowright \qquad \quad 
\tikzset{x=1em, y=2.1ex}
\InputIfFileExists{bone.tikz}{}{\input{./tikz/bone.tikz}}
\tikzset{x=1em, y=1.5ex}
\; \myeq{B11} \; 
\tikzset{x=1em, y=2.1ex}
\InputIfFileExists{empty-diag.tikz}{}{\input{./tikz/empty-diag.tikz}}
\tikzset{x=1em, y=1.5ex}
 \qquad 
\tikzset{x=1em, y=2.1ex}
\InputIfFileExists{feedback.tikz}{}{\input{./tikz/feedback.tikz}}
\tikzset{x=1em, y=1.5ex}
\; \myeq{B12} \; \arrowright
\end{equation*}

\vspace{.5em}
\hrule
\begin{equation*}

\tikzset{x=1em, y=2.1ex}
\InputIfFileExists{w-coassociative.tikz}{}{\input{./tikz/w-coassociative.tikz}}
\tikzset{x=1em, y=1.5ex}
\:\myeq{C1}\: 
\tikzset{x=1em, y=2.1ex}
\InputIfFileExists{w-coassociative-1.tikz}{}{\input{./tikz/w-coassociative-1.tikz}}
\tikzset{x=1em, y=1.5ex}
\qquad
 
\tikzset{x=1em, y=2.1ex}
\InputIfFileExists{w-right-counital.tikz}{}{\input{./tikz/w-right-counital.tikz}}
\tikzset{x=1em, y=1.5ex}
 \:\myeq{C2}\: 
\tikzset{x=1em, y=2.1ex}
\InputIfFileExists{lr-order.tikz}{}{\input{./tikz/lr-order.tikz}}
\tikzset{x=1em, y=1.5ex}

\qquad 
\tikzset{x=1em, y=2.1ex}
\InputIfFileExists{w-cocommutative.tikz}{}{\input{./tikz/w-cocommutative.tikz}}
\tikzset{x=1em, y=1.5ex}
\:\myeq{C3}\: 
\tikzset{x=1em, y=2.1ex}
\InputIfFileExists{wcomult.tikz}{}{\input{./tikz/wcomult.tikz}}
\tikzset{x=1em, y=1.5ex}

\end{equation*}
\begin{equation*}

\tikzset{x=1em, y=2.1ex}
\InputIfFileExists{w-associative.tikz}{}{\input{./tikz/w-associative.tikz}}
\tikzset{x=1em, y=1.5ex}
\:\myeq{C4}\: 
\tikzset{x=1em, y=2.1ex}
\InputIfFileExists{w-associative-1.tikz}{}{\input{./tikz/w-associative-1.tikz}}
\tikzset{x=1em, y=1.5ex}
\qquad

\tikzset{x=1em, y=2.1ex}
\InputIfFileExists{w-right-unital.tikz}{}{\input{./tikz/w-right-unital.tikz}}
\tikzset{x=1em, y=1.5ex}
\:\myeq{C5}\: 
\tikzset{x=1em, y=2.1ex}
\InputIfFileExists{lr-order.tikz}{}{\input{./tikz/lr-order.tikz}}
\tikzset{x=1em, y=1.5ex}
\qquad 
\tikzset{x=1em, y=2.1ex}
\InputIfFileExists{w-commutative.tikz}{}{\input{./tikz/w-commutative.tikz}}
\tikzset{x=1em, y=1.5ex}
\:\myeq{C6}\: 
\tikzset{x=1em, y=2.1ex}
\InputIfFileExists{wmult.tikz}{}{\input{./tikz/wmult.tikz}}
\tikzset{x=1em, y=1.5ex}

\end{equation*}

\vspace{.5em}
\hrule

\begin{equation*}

\tikzset{x=1em, y=2.1ex}
\InputIfFileExists{wb-copy-top.tikz}{}{\input{./tikz/wb-copy-top.tikz}}
\tikzset{x=1em, y=1.5ex}
 \:\myeq{D1}\: 
\tikzset{x=1em, y=2.1ex}
\InputIfFileExists{wunit-2.tikz}{}{\input{./tikz/wunit-2.tikz}}
\tikzset{x=1em, y=1.5ex}
 \qquad 
\tikzset{x=1em, y=2.1ex}
\InputIfFileExists{wb-bimonoid.tikz}{}{\input{./tikz/wb-bimonoid.tikz}}
\tikzset{x=1em, y=1.5ex}
\: \myeq{D2}\: 
\tikzset{x=1em, y=2.1ex}
\InputIfFileExists{wb-bimonoid-1.tikz}{}{\input{./tikz/wb-bimonoid-1.tikz}}
\tikzset{x=1em, y=1.5ex}
\qquad 
\tikzset{x=1em, y=2.1ex}
\InputIfFileExists{wb-intersect-counit.tikz}{}{\input{./tikz/wb-intersect-counit.tikz}}
\tikzset{x=1em, y=1.5ex}
 \:\myeq{D3}\: 
\tikzset{x=1em, y=2.1ex}
\InputIfFileExists{bcounit-2.tikz}{}{\input{./tikz/bcounit-2.tikz}}
\tikzset{x=1em, y=1.5ex}

\end{equation*}
\begin{equation*}

\tikzset{x=1em, y=2.1ex}
\InputIfFileExists{bw-counit-comult.tikz}{}{\input{./tikz/bw-counit-comult.tikz}}
\tikzset{x=1em, y=1.5ex}
\: \myeq{D4}\: 
\tikzset{x=1em, y=2.1ex}
\InputIfFileExists{bunit-2.tikz}{}{\input{./tikz/bunit-2.tikz}}
\tikzset{x=1em, y=1.5ex}
\qquad 
\tikzset{x=1em, y=2.1ex}
\InputIfFileExists{bw-bimonoid.tikz}{}{\input{./tikz/bw-bimonoid.tikz}}
\tikzset{x=1em, y=1.5ex}
\: \myeq{D5}\: 
\tikzset{x=1em, y=2.1ex}
\InputIfFileExists{bw-bimonoid-1.tikz}{}{\input{./tikz/bw-bimonoid-1.tikz}}
\tikzset{x=1em, y=1.5ex}
\qquad 
\tikzset{x=1em, y=2.1ex}
\InputIfFileExists{bw-mult-counit.tikz}{}{\input{./tikz/bw-mult-counit.tikz}}
\tikzset{x=1em, y=1.5ex}
 \:\myeq{D6}\: 
\tikzset{x=1em, y=2.1ex}
\InputIfFileExists{wcounit-2.tikz}{}{\input{./tikz/wcounit-2.tikz}}
\tikzset{x=1em, y=1.5ex}

\end{equation*}

\vspace{.5em}
\hrule

\begin{equation*}\label{eq:cartesian-bicat}

\tikzset{x=1em, y=2.1ex}
\InputIfFileExists{bmult-wmult.tikz}{}{\input{./tikz/bmult-wmult.tikz}}
\tikzset{x=1em, y=1.5ex}
\: \myleq{E1} \:
\tikzset{x=1em, y=2.1ex}
\InputIfFileExists{wmult-bmult.tikz}{}{\input{./tikz/wmult-bmult.tikz}}
\tikzset{x=1em, y=1.5ex}
 \quad 
\tikzset{x=1em, y=2.1ex}
\InputIfFileExists{bunit-2-wmult.tikz}{}{\input{./tikz/bunit-2-wmult.tikz}}
\tikzset{x=1em, y=1.5ex}
\:\myleq{E2}\:\Bunit \quad 
\tikzset{x=1em, y=2.1ex}
\InputIfFileExists{wcomult-bcomult.tikz}{}{\input{./tikz/wcomult-bcomult.tikz}}
\tikzset{x=1em, y=1.5ex}
\: \myleq{E4} \:
\tikzset{x=1em, y=2.1ex}
\InputIfFileExists{bcomult-wcomult.tikz}{}{\input{./tikz/bcomult-wcomult.tikz}}
\tikzset{x=1em, y=1.5ex}
 \quad 
\tikzset{x=1em, y=2.1ex}
\InputIfFileExists{wcomult-bcounit.tikz}{}{\input{./tikz/wcomult-bcounit.tikz}}
\tikzset{x=1em, y=1.5ex}
\:\myleq{E5}\:\Bcounit
\end{equation*}
\begin{equation*}\label{eq:scalar-bcomult-wmult} 
\tikzset{x=1em, y=2.1ex}
\InputIfFileExists{copy-atom.tikz}{}{\input{./tikz/copy-atom.tikz}}
\tikzset{x=1em, y=1.5ex}
 \:\myeq{E6}\: 
\tikzset{x=1em, y=2.1ex}
\InputIfFileExists{copy-atom-1.tikz}{}{\input{./tikz/copy-atom-1.tikz}}
\tikzset{x=1em, y=1.5ex}
 \quad 
\tikzset{x=1em, y=2.1ex}
\InputIfFileExists{delete-atom.tikz}{}{\input{./tikz/delete-atom.tikz}}
\tikzset{x=1em, y=1.5ex}
 \:\myeq{E7}\: \Bcounit \qquad 
\tikzset{x=1em, y=2.1ex}
\InputIfFileExists{atoms-wmult.tikz}{}{\input{./tikz/atoms-wmult.tikz}}
\tikzset{x=1em, y=1.5ex}
 \:\myeq{E8}\: 
\tikzset{x=1em, y=2.1ex}
\InputIfFileExists{wmult-atom.tikz}{}{\input{./tikz/wmult-atom.tikz}}
\tikzset{x=1em, y=1.5ex}
 \qquad \Wunit  \:\myeq{E9}\: 
\tikzset{x=1em, y=2.1ex}
\InputIfFileExists{wunit-atom.tikz}{}{\input{./tikz/wunit-atom.tikz}}
\tikzset{x=1em, y=1.5ex}

\end{equation*}
\begin{equation*}\label{eq:scalar-bmult-wcomult} 
\tikzset{x=1em, y=2.1ex}
\InputIfFileExists{co-copy-atom.tikz}{}{\input{./tikz/co-copy-atom.tikz}}
\tikzset{x=1em, y=1.5ex}
 \:\myeq{E10}\: 
\tikzset{x=1em, y=2.1ex}
\InputIfFileExists{co-copy-atom-1.tikz}{}{\input{./tikz/co-copy-atom-1.tikz}}
\tikzset{x=1em, y=1.5ex}
 \quad 
\tikzset{x=1em, y=2.1ex}
\InputIfFileExists{co-delete-atom.tikz}{}{\input{./tikz/co-delete-atom.tikz}}
\tikzset{x=1em, y=1.5ex}
 \:\myeq{E11}\: \Bunit \qquad  
\tikzset{x=1em, y=2.1ex}
\InputIfFileExists{wcomult-atoms.tikz}{}{\input{./tikz/wcomult-atoms.tikz}}
\tikzset{x=1em, y=1.5ex}
\:\myeq{E12}\: 
\tikzset{x=1em, y=2.1ex}
\InputIfFileExists{atom-wcomult.tikz}{}{\input{./tikz/atom-wcomult.tikz}}
\tikzset{x=1em, y=1.5ex}
 \qquad \Wcounit  \:\myeq{E13}\: 
\tikzset{x=1em, y=2.1ex}
\InputIfFileExists{atom-wcounit.tikz}{}{\input{./tikz/atom-wcounit.tikz}}
\tikzset{x=1em, y=1.5ex}

\end{equation*}

\vspace{.5em}
\hrule

\begin{equation*}

\tikzset{x=1em, y=2.1ex}
\InputIfFileExists{bmult-bcomult.tikz}{}{\input{./tikz/bmult-bcomult.tikz}}
\tikzset{x=1em, y=1.5ex}
 \:\myleq{F1}\: 
\tikzset{x=1em, y=2.1ex}
\InputIfFileExists{id-2.tikz}{}{\input{./tikz/id-2.tikz}}
\tikzset{x=1em, y=1.5ex}
\qquad  \idone \:\myleq{F2}\:\Bcounit\:\:\Bunit\qquad\quad \idone \:\myleq{F3}\: 
\tikzset{x=1em, y=2.1ex}
\InputIfFileExists{bcomult-bmult.tikz}{}{\input{./tikz/bcomult-bmult.tikz}}
\tikzset{x=1em, y=1.5ex}
 \qquad  
\tikzset{x=1em, y=2.1ex}
\InputIfFileExists{bunit-bcounit.tikz}{}{\input{./tikz/bunit-bcounit.tikz}}
\tikzset{x=1em, y=1.5ex}
 \:\myleq{F4}\: 
\tikzset{x=1em, y=2.1ex}
\InputIfFileExists{empty-diag.tikz}{}{\input{./tikz/empty-diag.tikz}}
\tikzset{x=1em, y=1.5ex}

\end{equation*}
\begin{equation*}

\tikzset{x=1em, y=2.1ex}
\InputIfFileExists{id-2.tikz}{}{\input{./tikz/id-2.tikz}}
\tikzset{x=1em, y=1.5ex}
\:\myleq{F5} \:
\tikzset{x=1em, y=2.1ex}
\InputIfFileExists{wmult-bcomult.tikz}{}{\input{./tikz/wmult-bcomult.tikz}}
\tikzset{x=1em, y=1.5ex}
 \qquad \Bcounit\:\:\Wunit \:\myleq{F6}\: \idone\qquad
\tikzset{x=1em, y=2.1ex}
\InputIfFileExists{bcomult-wmult.tikz}{}{\input{./tikz/bcomult-wmult.tikz}}
\tikzset{x=1em, y=1.5ex}
\:\myleq{F7}\: \idone \qquad  
\tikzset{x=1em, y=2.1ex}
\InputIfFileExists{empty-diag.tikz}{}{\input{./tikz/empty-diag.tikz}}
\tikzset{x=1em, y=1.5ex}
 \:\myleq{F8}\:
\tikzset{x=1em, y=2.1ex}
\InputIfFileExists{wunit-bcounit.tikz}{}{\input{./tikz/wunit-bcounit.tikz}}
\tikzset{x=1em, y=1.5ex}

\end{equation*}
\begin{equation*}

\tikzset{x=1em, y=2.1ex}
\InputIfFileExists{id-2.tikz}{}{\input{./tikz/id-2.tikz}}
\tikzset{x=1em, y=1.5ex}
\:\myleq{F9}\:
\tikzset{x=1em, y=2.1ex}
\InputIfFileExists{bmult-wcomult.tikz}{}{\input{./tikz/bmult-wcomult.tikz}}
\tikzset{x=1em, y=1.5ex}
 \qquad  \Wcounit\:\:\Bunit \:\myleq{F10}\: \idone\qquad
\tikzset{x=1em, y=2.1ex}
\InputIfFileExists{wcomult-bmult.tikz}{}{\input{./tikz/wcomult-bmult.tikz}}
\tikzset{x=1em, y=1.5ex}
 \:\myleq{F11}\: \idone \qquad  
\tikzset{x=1em, y=2.1ex}
\InputIfFileExists{empty-diag.tikz}{}{\input{./tikz/empty-diag.tikz}}
\tikzset{x=1em, y=1.5ex}
 \:\myleq{F12}\:
\tikzset{x=1em, y=2.1ex}
\InputIfFileExists{bunit-wcounit.tikz}{}{\input{./tikz/bunit-wcounit.tikz}}
\tikzset{x=1em, y=1.5ex}

\end{equation*}



\caption{Theory of Kleene Diagram Algebra (KDA).}\label{fig:axioms}
\vspace{-0.5cm}
\end{figure}

In Figure~\ref{fig:axioms} we introduce KDA, the  theory of \emph{Kleene Diagram Algebra}, on $\Syn$. Once we have shown how to encode automata into it, we will show that it is \emph{complete} for equivalence of automata-diagrams (Definition~\ref{def:automata-diagram} below). We explain some salient features of KDA below. As explained in Appendix~\ref{app:method}, we use equality as a shorthand for two inequalities.

\begin{itemize}
\item (A1)-(A2) relate $
\tikzset{x=1em, y=2.1ex}
\InputIfFileExists{cap-down.tikz}{}{\input{./tikz/cap-down.tikz}}
\tikzset{x=1em, y=1.5ex}
$ and $
\tikzset{x=1em, y=2.1ex}
\InputIfFileExists{cup-down.tikz}{}{\input{./tikz/cup-down.tikz}}
\tikzset{x=1em, y=1.5ex}
$, allowing us to bend and straighten wires at will. This makes $\Syn$ modulo (A1)-(A2), a \emph{compact closed category}~\cite{kelly1980compactclosed}. (A3) allows us to eliminate isolated loops.
\item The B block states that $
\tikzset{x=1em, y=2.1ex}
\InputIfFileExists{lr-copy.tikz}{}{\input{./tikz/lr-copy.tikz}}
\tikzset{x=1em, y=1.5ex}
,
\tikzset{x=1em, y=2.1ex}
\begin{tikzpicture}
	\begin{pgfonlayer}{nodelayer}
		\node [style=black] (37) at (0.75, 0) {};
		\node [style=none] (43) at (0.25, 0) {};
		\node [style=none] (44) at (-0.5, 0) {};
	\end{pgfonlayer}
	\begin{pgfonlayer}{edgelayer}
		\draw (43.center) to (37);
		\draw [->] (44.center) to (43.center);
	\end{pgfonlayer}
\end{tikzpicture}
}
\tikzset{x=1em, y=1.5ex}
$ forms a cocommutative comonoid (B1)-(B3), while $
\tikzset{x=1em, y=2.1ex}
\InputIfFileExists{lr-merge.tikz}{}{\input{./tikz/lr-merge.tikz}}
\tikzset{x=1em, y=1.5ex}
,
\tikzset{x=1em, y=2.1ex}
\begin{tikzpicture}
	\begin{pgfonlayer}{nodelayer}
		\node [style=black] (37) at (-0.5, 0) {};
		\node [style=none] (43) at (0.25, 0) {};
		\node [style=none] (44) at (0.75, 0) {};
	\end{pgfonlayer}
	\begin{pgfonlayer}{edgelayer}
		\draw [->] (37) to (43.center);
		\draw (44.center) to (43.center);
	\end{pgfonlayer}
\end{tikzpicture}
}
\tikzset{x=1em, y=1.5ex}
$ form a commutative monoid (B4)-(B6). By (co)commutativity, the (co)unitality axiom also holds with the (co)unit plugged into the other wire. More generally, as is usual in standard algebra reason tacitly modulo (co)commutativity, (co)associativity and (co)unitality axioms whenever convenient.  Moreover, $
\tikzset{x=1em, y=2.1ex}
\InputIfFileExists{lr-copy.tikz}{}{\input{./tikz/lr-copy.tikz}}
\tikzset{x=1em, y=1.5ex}
,
\tikzset{x=1em, y=2.1ex}
\begin{tikzpicture}
	\begin{pgfonlayer}{nodelayer}
		\node [style=black] (37) at (0.75, 0) {};
		\node [style=none] (43) at (0.25, 0) {};
		\node [style=none] (44) at (-0.5, 0) {};
	\end{pgfonlayer}
	\begin{pgfonlayer}{edgelayer}
		\draw (43.center) to (37);
		\draw [->] (44.center) to (43.center);
	\end{pgfonlayer}
\end{tikzpicture}
}
\tikzset{x=1em, y=1.5ex}
,
\tikzset{x=1em, y=2.1ex}
\InputIfFileExists{lr-merge.tikz}{}{\input{./tikz/lr-merge.tikz}}
\tikzset{x=1em, y=1.5ex}
,
\tikzset{x=1em, y=2.1ex}
\begin{tikzpicture}
	\begin{pgfonlayer}{nodelayer}
		\node [style=black] (37) at (-0.5, 0) {};
		\node [style=none] (43) at (0.25, 0) {};
		\node [style=none] (44) at (0.75, 0) {};
	\end{pgfonlayer}
	\begin{pgfonlayer}{edgelayer}
		\draw [->] (37) to (43.center);
		\draw (44.center) to (43.center);
	\end{pgfonlayer}
\end{tikzpicture}
}
\tikzset{x=1em, y=1.5ex}
$ together form an idempotent bimonoid (B7)-(B11).  Incidentally, note that (B1)-(B11) axiomatise the SMC of finite sets and relations, with monoidal product given by the disjoint sum. (B12)  allows us to eliminate trivial feedback loops, extending the previous axiomatisation to the traced SMC of sets and relations.
\item The C block makes $(\Wmult, \Wunit)$ into a commutative monoid and $(\Wcomult, \Wcounit)$ into a cocommutative comonoid. As for the black generators, we will reason tacitly modulo (co)commutativity, (co)associativity and (co)unitality axioms whenever convenient.
\item The D block states that both $(\Bcomult, \Bcounit,\Wmult, \Wunit)$ and $(\Wcomult, \Wcounit,\Bmult, \Bunit)$ form two bimonoids.
\item The E block encodes fundamental axioms of Cartesian bicategories~\cite{Carboni1987}. They are lax versions of distributive laws, of $\Wmult$ over $\Bmult$ and $\Wcomult$ over $\Bcomult$ (as well as their units and counits). We also have equalities that force the $\scalar{a}$ to distribute over the other operations. Semantically, this corresponds to the fact that the $\scalar{a}$ are lattice homomorphisms.

\item The F block state a number of adjunctions in the 2-categorical sense\footnote{In this setting, the 2-cells are simply inclusions, so the reader can also think about these adjunctions simply as Galois connections.}: two morphisms $f\from X \to Y$ and $g\from Y \to X$ are adjoint if $id_X \leq f\poi g$ and $g\poi f \leq id_Y$. We write $f\vdash g$ and say that $f$ is left ajdoint to $g$. The situation for KDA is summarised by the following six adjunctions:
\begin{equation}\label{eq:adjunctions-ternary}
\Wmult \,\dashv\, \Bcomult \,\dashv\, \Bmult \,\dashv\, \Wcomult
\end{equation}
\begin{equation}\label{eq:adjunctions-unary}
\Wunit \,\dashv\, \Bcounit \,\dashv\,\Bunit \,\dashv\, \Wcounit
\end{equation}
The central adjunctions involving only $\Bcomult, \Bcounit,\Bmult, \Bunit$ are the key defining adjunctions of Cartesian bicategories. The remaining adjunctions hold whenever the supporting poset is a lattice (has binary meets and joins), which is the case for the set of (regular) languages over a given alphabet, as it is closed under union and intersection. To better understand where these adjunctions come from, it is helpful to adopt a semantic point of view. For example, recall that  $\sem{\Wcomult} = \{(L,(K_1,K_2)) \mid L \subseteq K_1\cup K_2\}$ and $\sem{\Bcomult} = \left\{\big(L, (K_1,K_2)\big)\mid L\subseteq K_i,\, i=1,2\right\} = \left\{\big(L, (K_1,K_2)\big)\mid L\subseteq K_1\cap K_2\right\}$. Thus, one can see the adjunction $\Wmult \dashv \Bcomult$ as arising from the duality between the two different ways of turning intersection---a monotone map---into a monotone relation, \emph{cf.} Remark~\ref{rmk:monotone-maps}. These two embeddings of the same monotone map always give rise to an adjunction of this form.
Note that we can strengthen some of the inequalities in this block to equalities: the equalities for (F6),(F8),(F10), and (F12) can all be derived.
\item The equational theory contains a number of important symmetries: many equations also hold when taking the horizontal reflection of the diagrams involved, e.g., (B1) and (B4) or (D2) and (D5). This will play an important role in our proof of completeness in Section~\ref{sec:completeness}.
\item Finally, the proposed axiomatisation is \emph{not} minimal. For example (F3) and (F4) are obviously subsumed by (B10) and (B11).
Similarly, the equations of the E block could be weakened to inequalities. As we will frequently make use of several symmetries of the equational theory, for the reader's convenience we have preferred to add these redundant axioms to Fig.~\ref{fig:axioms}, instead of scattering them in different subsequent lemmas. They also serve to highlight common algebraic structures that occur in related theories (\emph{e.g.} bimonoids).
\end{itemize}

\noindent
From the data of the set of generators and the relations of Fig.~\ref{fig:axioms}, we can construct a partial order on each homset of $\Syn$ as explained in Appendix~\ref{sec:smit}. First we build a preorder on each homset by closing KDA under $\oplus$ and taking the reflexive and transitive closure of the resulting relation. Then we obtain a partial order by quotienting the resulting pre-order to impose anti-symmetry. Below, we will call $\leqKa$ the resulting order on each homset.

We are interested in the properties of those diagrams that correspond to automata. As we will see in the next section, there is a close relationship between diagrams composed exclusively of the generators in~\eqref{eq:syn1} and standard automata, justifying the following definition.
\begin{defi}\label{def:automata-diagram}
An \emph{automaton-diagram} is a morphism of $\Syn$ built from the generators of~\eqref{eq:syn1}, namely \[\Bcomult, \Bcounit, \Bmult, \Bunit, \scalar{a}, 
\tikzset{x=1em, y=2.1ex}
\InputIfFileExists{cap-down.tikz}{}{\input{./tikz/cap-down.tikz}}
\tikzset{x=1em, y=1.5ex}
, 
\tikzset{x=1em, y=2.1ex}
\InputIfFileExists{cup-down.tikz}{}{\input{./tikz/cup-down.tikz}}
\tikzset{x=1em, y=1.5ex}
.\]
We call $\Aut$ the corresponding monoidal subcategory.
\end{defi}
\begin{rem}
Note that, as anticipated in the introduction, constructing automata does not require the white generators of $\Syn$, as defined in~\eqref{eq:white-gen}. However, together with the corresponding adjoint structure in the inequational theory, they are essential to the completeness proof given below.
\end{rem}

We can now state our soundness and completeness result for automata-diagrams.
\begin{thm}[Soundness and Completeness]\label{thm:completeness}
For any two automata-diagrams $d$ and $d'$,
\[\sem{d}\subseteq \sem{d'} \text{ if and only if } d \leqKa d'\text{.}\]
\end{thm}
The \emph{soundness} of $\leqKa$ for the chosen interpretation is not difficult to show and involves a routine verification that all the axioms in Fig.~\ref{fig:axioms} hold in the semantics. We show (D2) here as an example.  We have
\begin{align*}
\sem{\,
\tikzset{x=1em, y=2.1ex}
\InputIfFileExists{wb-bimonoid.tikz}{}{\input{./tikz/wb-bimonoid.tikz}}
\tikzset{x=1em, y=1.5ex}
\,} &= \{((L_1, L_2), (K_1, K_2)) \mid \exists M.\; L_1 \cap L_2 \subseteq M \subseteq K_1\cap K_2\} \\
&= \{((L_1, L_2), (K_1, K_2)) \mid L_1 \cap L_2 \subseteq K_1\cap K_2\}\\
& \subseteq \left\{((L_1, L_2), (K_1, K_2)) \mid \;\begin{array}{l}
L_1 \subseteq L_1 \cup (K_1\cap K_2),\,
L_2\subseteq L_2 \cup (K_1\cap K_2),\\
L_1\cap L_2 \subseteq K_1,\, (K_1\cap K_2)\cap (K_1\cap K_2) \subseteq K_2,
\end{array} \right\}\\
& \subseteq \left\{((L_1, L_2), (K_1, K_2)) \mid \exists M_1.M_2.M_3.M_4.\;\begin{array}{l}
L_1 \subseteq M_1\cup M_3, \\
L_2\subseteq M_2\cup M_4,\\
 M_1\cap M_2 \subseteq K_1,\\
 M_3\cap M_4 \subseteq K_2
\end{array} \right\}\\
&= \sem{\,
\tikzset{x=1em, y=2.1ex}
\InputIfFileExists{wb-bimonoid-1.tikz}{}{\input{./tikz/wb-bimonoid-1.tikz}}
\tikzset{x=1em, y=1.5ex}
\,}
\end{align*}
and, conversely
\[
\sem{\,
\tikzset{x=1em, y=2.1ex}
\InputIfFileExists{wb-bimonoid-1.tikz}{}{\input{./tikz/wb-bimonoid-1.tikz}}
\tikzset{x=1em, y=1.5ex}
\,} = \left\{((L_1, L_2), (K_1, K_2)) \mid \exists M_1.M_2.M_3.M_4.\;\begin{array}{l}
L_1 \subseteq M_1\cup M_3, \\
L_2\subseteq M_2\cup M_4,\\
 M_1\cap M_2 \subseteq K_1,\\
 M_3\cap M_4 \subseteq K_2
\end{array} \right\}
\]
\begin{align*}
\qquad & \subseteq \left\{((L_1, L_2), (K_1, K_2)) \mid \exists M_1.M_2.M_3.M_4.\; \begin{array}{l} L_1 \cap L_2 \\ \subseteq (M_1\cup M_3)\cap (M_2\cup M_4)\\ \subseteq (M_1\cap M_2)\cup (M_3\cap M_4)\\\subseteq K_1\cap K_2\end{array}\right\}\\
& \subseteq \{((L_1, L_2), (K_1, K_2)) \mid \exists M.\; L_1 \cap L_2 \subseteq M \subseteq K_1\cap K_2\}\\
& = \sem{\,
\tikzset{x=1em, y=2.1ex}
\InputIfFileExists{wb-bimonoid.tikz}{}{\input{./tikz/wb-bimonoid.tikz}}
\tikzset{x=1em, y=1.5ex}
\,}
\end{align*}
\begin{rem}
The two black generating objects are not \emph{discrete} in the terminology of Carboni and Walters~\cite{Carboni1987}: this means that $
\tikzset{x=1em, y=2.1ex}
\InputIfFileExists{lr-copy.tikz}{}{\input{./tikz/lr-copy.tikz}}
\tikzset{x=1em, y=1.5ex}
, 
\tikzset{x=1em, y=2.1ex}
\begin{tikzpicture}
	\begin{pgfonlayer}{nodelayer}
		\node [style=black] (37) at (0.75, 0) {};
		\node [style=none] (43) at (0.25, 0) {};
		\node [style=none] (44) at (-0.5, 0) {};
	\end{pgfonlayer}
	\begin{pgfonlayer}{edgelayer}
		\draw (43.center) to (37);
		\draw [->] (44.center) to (43.center);
	\end{pgfonlayer}
\end{tikzpicture}
}
\tikzset{x=1em, y=1.5ex}
$ and $
\tikzset{x=1em, y=2.1ex}
\InputIfFileExists{lr-merge.tikz}{}{\input{./tikz/lr-merge.tikz}}
\tikzset{x=1em, y=1.5ex}
, 
\tikzset{x=1em, y=2.1ex}
\begin{tikzpicture}
	\begin{pgfonlayer}{nodelayer}
		\node [style=black] (37) at (-0.5, 0) {};
		\node [style=none] (43) at (0.25, 0) {};
		\node [style=none] (44) at (0.75, 0) {};
	\end{pgfonlayer}
	\begin{pgfonlayer}{edgelayer}
		\draw [->] (37) to (43.center);
		\draw (44.center) to (43.center);
	\end{pgfonlayer}
\end{tikzpicture}
}
\tikzset{x=1em, y=1.5ex}
$ do not satisfy the Frobenius law. In fact, because they already form a bimonoid, satisfying the Frobenius law would trivialise the equational theory, making any two diagrams of the same type equal.
\end{rem}
\begin{rem}\label{rmk:different-monoid}
The atomic actions $\scalar{a} \,(a\in\Sigma)$ compose freely. This is because, as we study automata, we are interested in the \emph{free} monoid $\Sigma^*$ over $\Sigma$. However, nothing would prevent us from modelling other structures. Free commutative monoids (powers of $\N$), whose rational subsets correspond to semilinear sets~\cite[Chapter 11]{conway2012regular} would be of particular interest.
\end{rem}
\begin{rem}
Semantically, the generators of automata-diagrams allow us to specify systems of linear language inequalities of the form $La\subseteq K$. The addition of the white nodes extends the expressiveness of our calculus, giving us the ability to specify systems involving intersection and union on both sides of an inequality. While it is clear that this is a strictly more expressive calculus---for example, the relations that interpret any of the white generators cannot be expressed using only the black ones---we leave the precise characterisation of the image of $\sem{\cdot}$ for future work.
\end{rem}
\begin{rem}
We have already explained that the adjunctions between the white and black nodes hold whenever the underlying poset is a lattice. In fact, the calculus we give in this paper could be interpreted over an arbitrary lattice, with the semantics of each $\scalar{a}$ ($a\in\Sigma$) given as a lattice endomorphism.

This last claim also shows that KDA is incomplete for the given interpretation. Indeed, the lattice of languages (and that of regular languages) is a Boolean algebra. This additional structure can be captured equationally, by making $(\Wcomult, \Wcounit,\Wmult, \Wunit)$ a Frobenius algebra. However, the defining equations of Frobenius algebras cannot be derived from KDA\@. One way to show this is to devise a counter-model by interpreting the calculus over a different lattice, which is not complemented.

Note that this is a feature, not a bug. We are hoping that the methods of this paper can be translated to other settings, \emph{e.g.} to automata-like structures or modal logics that do not require the underlying semantics to be Boolean.

\end{rem}

We will now write $\leqKa$ (resp. $\eqKa$) simply as $\leq$ (resp. $=$) to simplify notation, and say that diagrams $c$ and $d$ of the same type are \emph{equal} when $c\eqKa d$.

\section{Encoding regular expressions and automata}\label{sec:encoding}

A major appeal of our approach is that both regular expressions and automata can be uniformly represented in the graphical language of string diagrams, and the translation of one into the other becomes a (in)equational derivation in KDA\@. In fact, we will see there is a close resemblance between automata and the shape of the string diagrams interpreting them---the main difference being that string diagrams can be composed.

In this section we describe how regular expressions (resp.\ automata) can be encoded as string diagrams, such that their semantics corresponds in a precise way to the languages that they describe (resp.\ recognise).

\subsection{From regular expressions to string diagrams}

We can define an encoding $\transreg{-}$ of regular expressions into string diagrams of $\Aut$ inductively as follows:
\begin{align}
\transreg{e+f}&\; =\;  
\tikzset{x=1em, y=2.1ex}
\InputIfFileExists{sum.tikz}{}{\input{./tikz/sum.tikz}}
\tikzset{x=1em, y=1.5ex}
\quad &\transreg{0} \;=\; 
\tikzset{x=1em, y=2.1ex}
\begin{tikzpicture}
	\begin{pgfonlayer}{nodelayer}
		\node [style=black] (0) at (0.5, 0) {};
		\node [style=black] (1) at (-0.5, 0) {};
		\node [style=none] (2) at (-1.75, 0) {};
		\node [style=none] (3) at (1.75, 0) {};
	\end{pgfonlayer}
	\begin{pgfonlayer}{edgelayer}
		\draw (1) to (2.center);
		\draw (3.center) to (0);
	\end{pgfonlayer}
\end{tikzpicture}
}
\tikzset{x=1em, y=1.5ex}
  \nonumber\\
\transreg{ef}&\;= \; 
\tikzset{x=1em, y=2.1ex}
\InputIfFileExists{product.tikz}{}{\input{./tikz/product.tikz}}
\tikzset{x=1em, y=1.5ex}
  &\transreg{1} \;=\;  
\tikzset{x=1em, y=2.1ex}
\begin{tikzpicture}
	\begin{pgfonlayer}{nodelayer}
		\node [style=none] (0) at (-1.5, 0) {};
		\node [style=none] (1) at (0, 0) {};
		\node [style=none] (2) at (1.25, 0) {};
	\end{pgfonlayer}
	\begin{pgfonlayer}{edgelayer}
		\draw [->] (0.center) to (1.center);
		\draw (1.center) to (2.center);
	\end{pgfonlayer}
\end{tikzpicture}
}
\tikzset{x=1em, y=1.5ex}
 \nonumber\\%
\label{eq:regexpstar}
\transreg{e^*}&\; =\; 
\tikzset{x=1em, y=2.1ex}
\InputIfFileExists{star.tikz}{}{\input{./tikz/star.tikz}}
\tikzset{x=1em, y=1.5ex}
 &  \transreg{a} \;=\;   \scalar{a}
\end{align}
For example,
\begin{equation}\label{ex:regexp-diagram}
\transreg{ab(a+ab)^*} \; =\; 
\tikzset{x=1em, y=2.1ex}
\InputIfFileExists{ex-regexp.tikz}{}{\input{./tikz/ex-regexp.tikz}}
\tikzset{x=1em, y=1.5ex}

\end{equation}

\noindent Let $\semreg{e} \in \Lang$ be the standard semantics of a regular expression $e$, defined inductively as follows:
\begin{align*}
\semreg{e+f} = \semreg{e} \cup \semreg{f} \quad
\semreg{ef} = \{vw \mid v\in \semreg{e}, w\in \semreg{f}\} \\
\semreg{1} = \{\varepsilon\} \qquad \semreg{0} = \varnothing
\qquad \semreg{a} =  \{a\} (a\in\alphabet)  \qquad
\semreg{e^*} = \bigcup_{n\in\N}\semreg{e^n} \quad
\end{align*}
where $e^{n+1} := ee^{n}$ and $e^0 := 1$.
As expected, the translation preserves the language interpretation of regular expressions in a sense that the following proposition makes precise.
\begin{prop}\label{thm:diagram-regexp}
For $e, f$ two regular expressions, $\semreg{e} = \semreg{f}$ iff $\sem{\transreg{e}} = \sem{\transreg{f}}$.
\end{prop}
\begin{proof}
To prove the statement, it is enough to show that
 $\sem{\transreg{e}} = \left\{(L,K)\mid L \semreg{e} \,  \subseteq K\right\}$.
We do so by induction on the structure of regular expressions. Note that we write ``$\poi$'' for relational composition, from left to right: $R\poi S = \{(x,z)\mid \exists y, (x,y)\in R, (y,z)\in S\}$.

The proposition holds by definition for the generators: $\sem{\transreg{a}} = \{(L,K)\mid La\subseteq K\}$. There are three inductive cases to consider. Assume that $e$ and $f$ satisfy the proposition.
\begin{itemize}
\item For the $ef$ case, $\sem{\transreg{ef}} = \sem{\transreg{e}}\poi \sem{\transreg{f}} = \{(L,K)\mid L\semreg{e}\subseteq K\}\poi \{(L,K)\mid L\semreg{f}\subseteq K\}$. Hence, by monotony of the product, we have $\sem{\transreg{ef}} =\{(L,K)\mid L\semreg{e}\semreg{f}\subseteq K\} = \{(L,K)\mid \semreg{ef}L\subseteq K\}$.
\item For the case of $e+f$ we have \begin{align*}
\sem{\transreg{e+f}} &= \left\{(L,K)\mid \exists K_1, K_2, L_1, L_2.\;\begin{array}{l} K_1, K_2\subseteq K, \\
L\subseteq L_1, L_2,
\\L_1\semreg{e}\subseteq K_1,
\\L_2\semreg{f}\subseteq K_2\end{array}\right\}\\
&= \left\{(L,K)\mid \exists L_1, L_2.\; \begin{array}{l}L\subseteq L_1, L_2,
\\L_1\semreg{e}\subseteq K,
\\L_2\semreg{f}\subseteq K\end{array}\right\}\\
& = \left\{(L,K)\mid \exists L_1, L_2.\;\begin{array}{l} L\subseteq L_1, L_2,
\\ L_1\semreg{e} \cup L_2\semreg{f} \subseteq K\end{array}\right\}\\
&= \{(L,K)\mid L\semreg{e} \cup L\semreg{f} \subseteq K\} \\
&= \{(L,K)\mid L(\semreg{e}\cup \semreg{f}) \subseteq K\}\\
& =  \{(L,K)\mid L\semreg{e+f} \subseteq K\}
\end{align*}
\item Finally, for $e^*$,
\begin{align*}
\sem{\transreg{e^*}} &=  \{(L,K)\mid \exists M,N.\; \, M,L\subseteq N, N\semreg{e} \subseteq M, N\subseteq K\}\\
& = \{(L,K)\mid \exists N.\;  N\semreg{e} \subseteq N, L\subseteq N\subseteq K\}\\
& = \{(L,K)\mid \exists N.\;  L\cup N\semreg{e} \subseteq N, L\subseteq N\subseteq K\}\\
& \myeq{$\star$} \{(L,K)\mid \exists N.\; L \semreg{e}^* \subseteq N, L\subseteq N\subseteq K\}\\
& = \{(L,K)\mid \exists N.\; L\semreg{e^*} \subseteq N, L\subseteq N\subseteq K\}\\
& = \{(L,K)\mid L\semreg{e^*} \subseteq K\}
\end{align*}
where the starred equation is a consequence of Arden's lemma~\cite{arden1961delayed}: $A^*B$ is the smallest solution (for $X$) of the language equation $B\cup AX\subseteq X$, where we write $A^*$ for the language $\bigcup_{n\geq 0}A^n$.
\qedhere
\end{itemize}
\end{proof}
\begin{rem}
Regular expressions can also be interpreted as binary relations over an arbitrary set: such an interpretation is given by a mapping of each letter to a binary relation on some set, and extended inductively to all regexes (with the sum interpreted as union, product as relational composition, $1$ as the identity, $0$ as the empty relation, and the star as the reflexive, transitive closure). We can write $\mathsf{Rel}\vDash e=f$ if every relational interpretation identifies $e$ and $f$. The statement and proof of Proposition~\ref{thm:diagram-regexp} follow a slight modification of the proof that, $\mathsf{Rel}\vDash e=f$ implies $\semreg{e} = \semreg{f}$\footnote{We thank the anonymous reviewer for pointing this out. The earliest reference we could find is~\cite[p.24]{pratt1980dynamic}}. The standard argument proceeds as follows:  define a map $\sigma : \Sigma \to 2^{\Sigma^\star\times\Sigma^\star}$ given by $\sigma(a) = \{(w, wa) \mid w\in \Sigma^\star\}$, which can be extended inductively to a map $\hat\sigma$ defined over all regexes. Then, one can show by induction that $\hat\sigma(e) = \{(w,wu) \mid u\in \semreg{e}\}$ so that, in particular, $\semreg{e} = \{w \mid (\epsilon,w) \in \hat\sigma(e)\}$. The version presented above modifies this idea by adding inclusion where necessary, to turn all the relevant relations into \emph{monotone} relations.
\end{rem}
From a diagrammatic perspective, regular expressions correspond to diagrams that enforce a restricted form of composition. They can be characterised in the syntax as the image of $\transreg{\cdot}$ or, equivalently, as those diagrams of type $\objr\to\objr$ built inductively from the following three operations
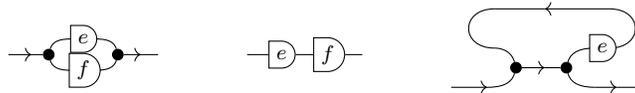
\[
 
\tikzset{x=1em, y=2.1ex}
\InputIfFileExists{sum.tikz}{}{\input{./tikz/sum.tikz}}
\tikzset{x=1em, y=1.5ex}
\qquad\quad   
\tikzset{x=1em, y=2.1ex}
\InputIfFileExists{product.tikz}{}{\input{./tikz/product.tikz}}
\tikzset{x=1em, y=1.5ex}
 \qquad\quad 
\tikzset{x=1em, y=2.1ex}
\InputIfFileExists{star.tikz}{}{\input{./tikz/star.tikz}}
\tikzset{x=1em, y=1.5ex}

\]
starting from the basic diagrams $\scalar{a}$, $
\tikzset{x=1em, y=2.1ex}
\begin{tikzpicture}
	\begin{pgfonlayer}{nodelayer}
		\node [style=black] (0) at (0.5, 0) {};
		\node [style=black] (1) at (-0.5, 0) {};
		\node [style=none] (2) at (-1.75, 0) {};
		\node [style=none] (3) at (1.75, 0) {};
	\end{pgfonlayer}
	\begin{pgfonlayer}{edgelayer}
		\draw (1) to (2.center);
		\draw (3.center) to (0);
	\end{pgfonlayer}
\end{tikzpicture}
}
\tikzset{x=1em, y=1.5ex}
$, and $
\tikzset{x=1em, y=2.1ex}
\begin{tikzpicture}
	\begin{pgfonlayer}{nodelayer}
		\node [style=none] (0) at (-1.5, 0) {};
		\node [style=none] (1) at (0, 0) {};
		\node [style=none] (2) at (1.25, 0) {};
	\end{pgfonlayer}
	\begin{pgfonlayer}{edgelayer}
		\draw [->] (0.center) to (1.center);
		\draw (1.center) to (2.center);
	\end{pgfonlayer}
\end{tikzpicture}
}
\tikzset{x=1em, y=1.5ex}
$. In what follows, we will refer to any diagram of this form as a \emph{regex-diagram}.

\subsection{From automata to string diagrams}\label{sec:automata-diag}

Example~\eqref{ex:regexp-diagram} suggests that the string diagram $\transreg{e}$ corresponding to a regular expression $e$ looks a lot like a nondeterministic finite-state automaton (NFA) for $e$. In fact, the translation $\transreg{-}$ can be seen as the diagrammatic counterpart of Thompson's construction~\cite{thompson1968programming} that builds an NFA from a regular expression.

We can generalise the encoding of regular expressions and translate NFA directly into string diagrams, in at least two ways. The first is to encode an NFA as the diagrammatic counterpart of its transition relation. The second is to translate directly its graph representation into the diagrammatic syntax.

\paragraph{Encoding the transition relation.}\label{sec:encoding-relations} This is a simple variant of the  translation of matrices over semirings that has appeared in several places in the literature~\cite{Lack2004a,ZanasiThesis}.

Let $A$ be an NFA with set of states $Q$, initial state $q_0\in Q$, accepting states $F\subseteq Q$ and transition relation $\delta\subseteq Q\times \Sigma\times Q$.
We can represent $\delta$
as a string diagram $d$ with $|Q|$ incoming wires on the left and $|Q|$ outgoing wires on the right.
The left $j$-th port of $d$ is connected to the $i$-th port on the right through an $\scalar{a}$ whenever $(q_i,a,q_j)\in\delta$. To accommodate nondeterminism, when the same two ports are connected by several different letters of $\Sigma$, we join these using 
\tikzset{x=1em, y=2.1ex}
\InputIfFileExists{lr-copy.tikz}{}{\input{./tikz/lr-copy.tikz}}
\tikzset{x=1em, y=1.5ex}
 and 
\tikzset{x=1em, y=2.1ex}
\InputIfFileExists{lr-merge.tikz}{}{\input{./tikz/lr-merge.tikz}}
\tikzset{x=1em, y=1.5ex}
. When $(q_i, \epsilon, q_j)\in\delta$, the two ports are simply connected via a plain identity wire. If there is no tuple in $\delta$ such that $(q_i, a, q_j)\in\delta$ for any $a$, the two corresponding ports are disconnected, using $\Bcounit\; \Bunit$ if necessary.

For example, the transition relation of an NFA with three states and \[\delta = \{((q_0, a, q_1), (q_1, b, q_2), (q_2, a, q_1), (q_2, a, q_2))\}\] (disregarding the initial and accepting states for the moment) is depicted below. Conversely, given such a diagram, we can recover $\delta$ by collecting $\Sigma$-weighted paths from left to right ports.
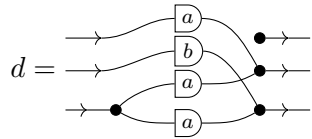
\[d = 
\tikzset{x=1em, y=2.1ex}
\InputIfFileExists{ex-matrix-nfa.tikz}{}{\input{./tikz/ex-matrix-nfa.tikz}}
\tikzset{x=1em, y=1.5ex}
\]

To deal with the initial state, we add an additional incoming wire connected to the right port corresponding to the initial state of the automaton. Similarly, for accepting states we add an additional outgoing wire, connected to the left ports corresponding to each accepting state, via 
\tikzset{x=1em, y=2.1ex}
\InputIfFileExists{lr-merge.tikz}{}{\input{./tikz/lr-merge.tikz}}
\tikzset{x=1em, y=1.5ex}
 if there is more than one.

Finally, we trace out the $|Q|$ wires of the diagrammatic transition relation to obtain the associated string diagram. In other words, for a NFA with initial state  $q_0$, set of accepting states $F$, transition relation $\delta$, we obtain the string diagram below, where  $d$ is the diagrammatic counterpart of $\delta$ as defined above, $e$ is the injection of a single wire as the first amongst $|Q|$ wires, and $f$ discards all wires that are not associated to states in $F$ with $
\tikzset{x=1em, y=2.1ex}
\begin{tikzpicture}
	\begin{pgfonlayer}{nodelayer}
		\node [style=black] (37) at (0.75, 0) {};
		\node [style=none] (43) at (0.25, 0) {};
		\node [style=none] (44) at (-0.5, 0) {};
	\end{pgfonlayer}
	\begin{pgfonlayer}{edgelayer}
		\draw (43.center) to (37);
		\draw [->] (44.center) to (43.center);
	\end{pgfonlayer}
\end{tikzpicture}
}
\tikzset{x=1em, y=1.5ex}
$, and applies $
\tikzset{x=1em, y=2.1ex}
\InputIfFileExists{lr-merge.tikz}{}{\input{./tikz/lr-merge.tikz}}
\tikzset{x=1em, y=1.5ex}
$ to merge them into a single outgoing wire.
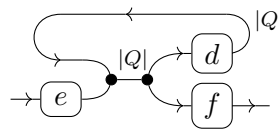
\[
\tikzset{x=1em, y=2.1ex}
\InputIfFileExists{automata-rep.tikz}{}{\input{./tikz/automata-rep.tikz}}
\tikzset{x=1em, y=1.5ex}
\]

For example, if $A$ with $\delta$ as above has initial state $q_0$ and set of accepting states $\{q_2\}$, we get the diagram below left; if instead, all states are accepting, we obtain the diagram below right:
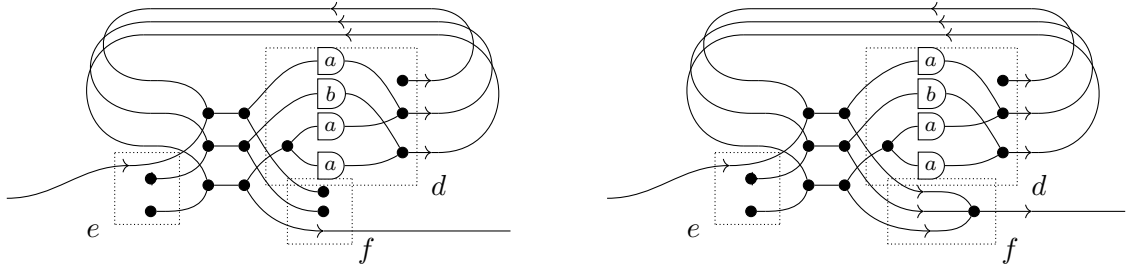
\begin{equation*}\label{eq:nfa-reps}

\tikzset{x=1em, y=2.1ex}
\InputIfFileExists{ex-rep-star.tikz}{}{\input{./tikz/ex-rep-star.tikz}}
\tikzset{x=1em, y=1.5ex}
 \qquad\quad  
\tikzset{x=1em, y=2.1ex}
\InputIfFileExists{ex-rep-all-accept.tikz}{}{\input{./tikz/ex-rep-all-accept.tikz}}
\tikzset{x=1em, y=1.5ex}

\end{equation*}
The correctness of this simple translation is justified by a semantic correspondence between the language recognised by a given NFA $A$ and the denotation of the corresponding string diagram.
\begin{prop}\label{thm:nfa-to-diag}
Given an NFA $A$ which recognises the language $L$, let $c_A$ be its associated string diagram, constructed as above. Then $\sem{c_A} = \left\{(K,K')\mid LK \subseteq K'\right\}$.
\end{prop}
\begin{proof}
This is the diagrammatic counterpart of the representation of automata as matrices of regular expressions given in~\cite[Definition 12]{kozen1994completeness}.

We write $\mathbf{K}$ for a vector of languages $(K_1,\dots, K_{|Q|})$ and $\mathbf{A}$ for a square matrix of languages; let $\mathbf{A}\mathbf{K}$ be the language vector resulting from applying $\mathbf{A}$ to $\mathbf{K}$ in the obvious way (note that we use standard matrix multiplication order, which is the opposite of the diagrammatic order). By~\cite[Theorem 11]{kozen1994completeness}, square language matrices form a Kleene algebra, with the composition as product, component-wise union as sum and the star defined as in~\cite[Lemma 10]{kozen1994completeness}. We also write write $\mathbf{K}\subseteq \mathbf{K}'$ if the inclusions all hold component-wise. Furthermore, Arden's lemma holds in this more general setting: the least solution of the language-matrix equation $\mathbf{B}\cup \mathbf{A}\mathbf{X}\subseteq \mathbf{X}$ is $\mathbf{X} = \mathbf{A}^*\mathbf{B}$. This is another consequence of the fact that matrices of languages also form a Kleene algebra~\cite[Theorem 11]{kozen1994completeness}.

Now, for a given automaton $A$ we construct the diagram below as explained above:
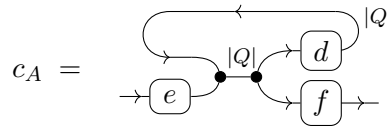
\[c_A \; = \quad 
\tikzset{x=1em, y=2.1ex}
\InputIfFileExists{automata-rep.tikz}{}{\input{./tikz/automata-rep.tikz}}
\tikzset{x=1em, y=1.5ex}
\]
with $d$ the diagram encoding the transition relation of $A$, $e_0$ the diagram encoding its initial state, and $f$ the diagram encoding its set of final states. Let $\sem{d} = \mathbf{D}$ be the language matrix obtained from $A$ by letting $\mathbf{D}_{ij} = \{a\}$ if $(q_i, a, q_j)$ is in the transition relation of $A$.
We proceed as in Example~\ref{ex:star-semantics}. First, we have
\[\sem{
\tikzset{x=1em, y=2.1ex}
\InputIfFileExists{d-star-Q-states.tikz}{}{\input{./tikz/d-star-Q-states.tikz}}
\tikzset{x=1em, y=1.5ex}
}
\begin{array}{l}
=  \{(\mathbf{K},\mathbf{K}')\mid \exists \mathbf{M},\mathbf{N}, \; \mathbf{M},\mathbf{K}\subseteq \mathbf{N},\; \mathbf{D}\mathbf{N}\subseteq \mathbf{M},\; \mathbf{N}\subseteq \mathbf{K}'\}\\
=  \{(\mathbf{K},\mathbf{K}')\mid \exists \mathbf{N}, \; \mathbf{D}\mathbf{N}\subseteq \mathbf{N},\; \mathbf{K}\subseteq \mathbf{N}\subseteq \mathbf{K}'\}\\
= \{(\mathbf{K},\mathbf{K}')\mid \exists \mathbf{N}, \; \mathbf{K}\cup \mathbf{D}\mathbf{N}\subseteq \mathbf{N}, \; \mathbf{N}\subseteq \mathbf{K}'\}\\
\myeq{$\star$} \{(\mathbf{K},\mathbf{K}')\mid \exists \mathbf{N}, \;  \mathbf{D}^*\mathbf{K}\subseteq \mathbf{N}, \mathbf{N}\subseteq \mathbf{K}'\}\\
= \{(\mathbf{K},\mathbf{K}')\mid  \mathbf{D}^*\mathbf{K}\subseteq \mathbf{K}'\}
\end{array}
\]
where the starred step holds by the matrix-form of Arden's lemma. Then, $\sem{e}$ and $\sem{f}$ pick out the component languages of $\mathbf{D}^*$ that correspond to the initial state of $A$ and each final state, and takes their union. Thus, we get
\[\sem{c_A} = \sem{
\tikzset{x=1em, y=2.1ex}
\InputIfFileExists{automata-rep.tikz}{}{\input{./tikz/automata-rep.tikz}}
\tikzset{x=1em, y=1.5ex}
}\begin{array}{l}
=\sem{e}\poi \{(\mathbf{K},\mathbf{K}')\mid  \mathbf{D}^*\mathbf{K}\subseteq \mathbf{K}'\}\poi \sem{f}\\ = \left\{(K,K')\mid LK \subseteq K'\right\}
\end{array} \]
where $L$ is the language accepted by the original automaton.
\end{proof}

\paragraph{From graphs to string diagrams.} The second way of translating automata into string diagrams mimics more directly the usual representation of automata as graphs. The idea (which should be sufficiently intuitive to not need to be made formal here) is, for each state, to use 
\tikzset{x=1em, y=2.1ex}
\InputIfFileExists{lr-merge.tikz}{}{\input{./tikz/lr-merge.tikz}}
\tikzset{x=1em, y=1.5ex}
 to represent incoming edges, and  
\tikzset{x=1em, y=2.1ex}
\InputIfFileExists{lr-copy.tikz}{}{\input{./tikz/lr-copy.tikz}}
\tikzset{x=1em, y=1.5ex}
 to represent outgoing edges. As above, labels $a \in A$ will be modelled using $\scalar{a}$. For example, the graph and the associated string diagram corresponding with the NFA above are
\begin{equation}\label{eq:ex-automaton-diagram}

\tikzset{x=1em, y=2.1ex}
\InputIfFileExists{ex-automaton-graph.tikz}{}{\input{./tikz/ex-automaton-graph.tikz}}
\tikzset{x=1em, y=1.5ex}
 \quad \mapsto \quad 
\tikzset{x=1em, y=2.1ex}
\InputIfFileExists{ex-automaton-diagram.tikz}{}{\input{./tikz/ex-automaton-diagram.tikz}}
\tikzset{x=1em, y=1.5ex}

\end{equation}
Note that the initial state (which we indicate with an arrow pointing down and into a state) of the automaton corresponds to the left interface of the string diagram, and the accepting state (which we indicate with an arrow pointing down and out of a state) to the right interface of the same diagram. As before, when there are multiple accepting states, they all connect to a single right interface, via $
\tikzset{x=1em, y=2.1ex}
\InputIfFileExists{lr-merge.tikz}{}{\input{./tikz/lr-merge.tikz}}
\tikzset{x=1em, y=1.5ex}
$. For example, if we make all states accepting in the automaton above, we get the following diagrammatic representation:
\begin{align*}

\tikzset{x=1em, y=2.1ex}
\InputIfFileExists{ex-automaton-graph-multiple-accept.tikz}{}{\input{./tikz/ex-automaton-graph-multiple-accept.tikz}}
\tikzset{x=1em, y=1.5ex}
\quad & \mapsto \quad 
\tikzset{x=1em, y=2.1ex}
\InputIfFileExists{ex-automaton-diagram-multiple-accept.tikz}{}{\input{./tikz/ex-automaton-diagram-multiple-accept.tikz}}
\tikzset{x=1em, y=1.5ex}

\end{align*}

\subsection{From string diagrams to automata}%
\label{sec:diagram-to-automata}

The previous discussion shows how NFAs can be seen as string diagrams of type $\objr\to \objr$. The converse is also true: we now show how to extract an automaton from any automaton-diagram $d \colon \objr\to \objr$, such that the language the automaton recognises matches the semantics of $d$.

In order to phrase this correspondence formally, we need to introduce some terminology. We call \emph{left-to-right} those automata-diagrams whose domain and co-domain contain only $\objr$, i.e.\ their type is of the form $\objr^{n} \to \objr^{m}$. The idea is that, in any such string diagram, the $n$ left interfaces act as \emph{inputs} of the computation, and the $m$ right interfaces act as \emph{outputs}. For instance,~\eqref{eq:ex-automaton-diagram} is a left-to-right diagram $\objr \to \objr$.
We call \emph{block} of a certain subset of generators a diagram composed only of these generators (using both $\poi$ and $\oplus$), possibly including some permutation of the wires.
\begin{defi}\label{def:matrix-diagram}
A \emph{matrix-diagram}
is a left-to-right diagram
that factors as a composition of a block of $
\tikzset{x=1em, y=2.1ex}
\InputIfFileExists{lr-copy.tikz}{}{\input{./tikz/lr-copy.tikz}}
\tikzset{x=1em, y=1.5ex}
, 
\tikzset{x=1em, y=2.1ex}
\begin{tikzpicture}
	\begin{pgfonlayer}{nodelayer}
		\node [style=black] (37) at (0.75, 0) {};
		\node [style=none] (43) at (0.25, 0) {};
		\node [style=none] (44) at (-0.5, 0) {};
	\end{pgfonlayer}
	\begin{pgfonlayer}{edgelayer}
		\draw (43.center) to (37);
		\draw [->] (44.center) to (43.center);
	\end{pgfonlayer}
\end{tikzpicture}
}
\tikzset{x=1em, y=1.5ex}
$, followed by a block of $\scalar{a}$ for $a\in \Sigma$ and finally, a block of $
\tikzset{x=1em, y=2.1ex}
\InputIfFileExists{lr-merge.tikz}{}{\input{./tikz/lr-merge.tikz}}
\tikzset{x=1em, y=1.5ex}
, 
\tikzset{x=1em, y=2.1ex}
\begin{tikzpicture}
	\begin{pgfonlayer}{nodelayer}
		\node [style=black] (37) at (-0.5, 0) {};
		\node [style=none] (43) at (0.25, 0) {};
		\node [style=none] (44) at (0.75, 0) {};
	\end{pgfonlayer}
	\begin{pgfonlayer}{edgelayer}
		\draw [->] (37) to (43.center);
		\draw (44.center) to (43.center);
	\end{pgfonlayer}
\end{tikzpicture}
}
\tikzset{x=1em, y=1.5ex}
$, such that any path from a left port to a right port passes through \emph{at most one} $\scalar{a}$.
\end{defi}
\noindent To each matrix-diagram $d$ we can associate a unique transition relation $\delta$ by gathering paths from each input to each output: $(q_i, a, q_j)\in\delta$ if there is $\scalar{a}$ joining the $i$-th input to the $j$-th output.
\noindent
A transition relation is \emph{$\epsilon$-free} if it does not contain the empty word. It is \emph{deterministic} if it is $\epsilon$-free and, for each $i$ and each $a\in\Sigma$ there is at most one $j$ such that $(q_i, a, q_j)\in\delta$. We will apply these terms to matrix-diagrams and the associated transition relation interchangeably. The example of Section~\ref{sec:automata-diag} below, with the three blocks highlighted, is a matrix-diagram.
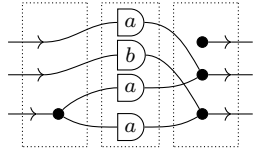
\begin{center}

\tikzset{x=1em, y=2.1ex}
\InputIfFileExists{ex-matrix-nfa-blocks.tikz}{}{\input{./tikz/ex-matrix-nfa-blocks.tikz}}
\tikzset{x=1em, y=1.5ex}

\end{center}
It is $\epsilon$-free but not deterministic since there are two $a$-labelled transitions starting from the third input.

We also call \emph{relation-diagram} a matrix-diagram that contain no $\scalar{a}$. Intuitively, in the absence of the $\scalar{a}$ generators, the corresponding theory is simply that of Boolean matrices, \emph{i.e.} relations. We now introduce representations of (automata-)diagrams, the diagrammatic counterpart of Kozen's automata in matrix form (written $(u,M,v)$, with semantics $uM^*v$ in~\cite[Definition 12]{kozen1994completeness}).

\begin{defi}[Representation]\label{def:representation}
For a diagram\footnote{Representations could also be defined for arbitrary left-to-right diagrams $\objr^m\to \objr^n$, but we will only need them to connect diagrams and automata, so it is sufficient to consider the $n=m=1$ case for our purpose.} $c\from \objr\to \objr$, a \emph{representation} is a triple $(e,d,f)$ of an $\epsilon$-free matrix-diagram $d\from \objr^{l}\to \objr^{l}$ representing the transition dynamics, and two relation-diagrams $e\from \objr\to \objr^{l}$, and $f\from \objr^{l}\to \objr$ representing the initial and the final states respectively, such that
\[\dbox{c}  \; = \; 
\tikzset{x=1em, y=2.1ex}
\InputIfFileExists{automata-rep-star.tikz}{}{\input{./tikz/automata-rep-star.tikz}}
\tikzset{x=1em, y=1.5ex}
 \quad \text{ where } \dbox{d^*} := 
\tikzset{x=1em, y=2.1ex}
\InputIfFileExists{d-star.tikz}{}{\input{./tikz/d-star.tikz}}
\tikzset{x=1em, y=1.5ex}
\]
It is a \emph{deterministic representation} if moreover $d$ is a deterministic matrix-diagram and there is only one right-port connected to the only left-port of $e$ (\emph{i.e.}, there is exactly one initial state).
\end{defi}
For example, given the string diagram below on the left, we can use the axioms of KDA to rewrite it to an equivalent diagram from which a representation can easily be read---the highlighted matrix-diagram corresponds to the same transition matrix $d$ as in the example above:
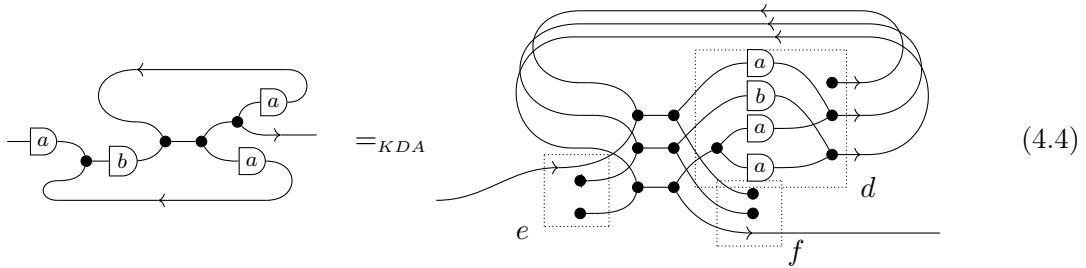
\begin{equation}

\tikzset{x=1em, y=2.1ex}
\InputIfFileExists{ex-automaton-diagram.tikz}{}{\input{./tikz/ex-automaton-diagram.tikz}}
\tikzset{x=1em, y=1.5ex}
\quad\eqKa
\tikzset{x=1em, y=2.1ex}
\InputIfFileExists{ex-rep-star.tikz}{}{\input{./tikz/ex-rep-star.tikz}}
\tikzset{x=1em, y=1.5ex}

\end{equation}
From a diagram $c : \objr\to \objr$  with representation %
$(e,d,f)$, we can construct an NFA as follows:
\begin{itemize}
\item its state set is $Q = \{q_1, \dots, q_l\}$, i.e., there is one state for each wire of $d\from \objr^l\to\objr^l$;
\item its transition relation built from $d$ as described above;
\item its initial state is the only non-zero coefficient of $e\from \objr\to \objr^l$, i.e., the only wire in the codomain of $e$ connected to the single wire in the domain;
\item its final states $F$ are those $q_j$ for which the $j$-th coefficient of $f\from \objr^l\to \objr$ is non-zero, i.e., the wires of the domain of $f$ connected to its single codomain wire.
\end{itemize}
The construction above is the inverse of that of Section~\ref{sec:automata-diag}. The link between the constructed automaton and the original string diagram $c$ is summarised in the following statement, which is a straightforward corollary of Proposition~\ref{thm:nfa-to-diag}.
\begin{prop}\label{thm:diagram-nfa}
For a diagram $c\from \objr\to \objr$ with a representation $\hat{c} = (e,d,f)$, let $A_{\hat{c}}$ be the associated automaton, constructed as above. Then $\hat{L}$ is the language recognised by $A_{\hat{c}}$ iff $\sem{c} = \sem{e;d^*;f} = \left\{(K,K')\mid \hat{L}K \subseteq K'\right\}$.
\end{prop}
The next proposition is crucial: it states that a representation can be extracted from any diagram $\objr\to\objr$.
\begin{prop}\label{thm:traceform}
Any automaton-diagram $\objr\to \objr$ has a representation.
\end{prop}
\noindent We will need to prove a few preliminary results before tackling the proof of Proposition~\ref{thm:traceform}. The following lemma will also be needed in the determinisation procedure of Section~\ref{sec:determinisation}.
The next theorem establishes completeness for a restricted fragment of our language, corresponding to matrices of finite languages. More precisely, every diagram formed only of the generators $
\tikzset{x=1em, y=2.1ex}
\InputIfFileExists{lr-copy.tikz}{}{\input{./tikz/lr-copy.tikz}}
\tikzset{x=1em, y=1.5ex}
, 
\tikzset{x=1em, y=2.1ex}
\begin{tikzpicture}
	\begin{pgfonlayer}{nodelayer}
		\node [style=black] (37) at (0.75, 0) {};
		\node [style=none] (43) at (0.25, 0) {};
		\node [style=none] (44) at (-0.5, 0) {};
	\end{pgfonlayer}
	\begin{pgfonlayer}{edgelayer}
		\draw (43.center) to (37);
		\draw [->] (44.center) to (43.center);
	\end{pgfonlayer}
\end{tikzpicture}
}
\tikzset{x=1em, y=1.5ex}
$, $
\tikzset{x=1em, y=2.1ex}
\InputIfFileExists{lr-merge.tikz}{}{\input{./tikz/lr-merge.tikz}}
\tikzset{x=1em, y=1.5ex}
, 
\tikzset{x=1em, y=2.1ex}
\begin{tikzpicture}
	\begin{pgfonlayer}{nodelayer}
		\node [style=black] (37) at (-0.5, 0) {};
		\node [style=none] (43) at (0.25, 0) {};
		\node [style=none] (44) at (0.75, 0) {};
	\end{pgfonlayer}
	\begin{pgfonlayer}{edgelayer}
		\draw [->] (37) to (43.center);
		\draw (44.center) to (43.center);
	\end{pgfonlayer}
\end{tikzpicture}
}
\tikzset{x=1em, y=1.5ex}
$, $\scalar{a}$ is interpreted via $\sem{\cdot}$ as a matrix with coefficients in $\mathbb{B}(\Sigma^*)$, the semiring of \emph{finite} sets of words over the alphabet $\Sigma$.
\begin{thm}[Matrix completeness]\label{thm:matrix-completeness}
 For any two diagrams $c,d$ formed only of the generators $
\tikzset{x=1em, y=2.1ex}
\InputIfFileExists{lr-copy.tikz}{}{\input{./tikz/lr-copy.tikz}}
\tikzset{x=1em, y=1.5ex}
, 
\tikzset{x=1em, y=2.1ex}
\begin{tikzpicture}
	\begin{pgfonlayer}{nodelayer}
		\node [style=black] (37) at (0.75, 0) {};
		\node [style=none] (43) at (0.25, 0) {};
		\node [style=none] (44) at (-0.5, 0) {};
	\end{pgfonlayer}
	\begin{pgfonlayer}{edgelayer}
		\draw (43.center) to (37);
		\draw [->] (44.center) to (43.center);
	\end{pgfonlayer}
\end{tikzpicture}
}
\tikzset{x=1em, y=1.5ex}
$, $
\tikzset{x=1em, y=2.1ex}
\InputIfFileExists{lr-merge.tikz}{}{\input{./tikz/lr-merge.tikz}}
\tikzset{x=1em, y=1.5ex}
, 
\tikzset{x=1em, y=2.1ex}
\begin{tikzpicture}
	\begin{pgfonlayer}{nodelayer}
		\node [style=black] (37) at (-0.5, 0) {};
		\node [style=none] (43) at (0.25, 0) {};
		\node [style=none] (44) at (0.75, 0) {};
	\end{pgfonlayer}
	\begin{pgfonlayer}{edgelayer}
		\draw [->] (37) to (43.center);
		\draw (44.center) to (43.center);
	\end{pgfonlayer}
\end{tikzpicture}
}
\tikzset{x=1em, y=1.5ex}
$, $\scalar{a}$, we have $\sem{c} = \sem{d}$ iff $c=d$.
\end{thm}
\begin{proof}
This result is particular case of a standard fact, that can be found for matrices over a ring in~\cite[Chapter 3]{ZanasiThesis}. However, the relevant proof of~\cite[Proposition 3.9]{ZanasiThesis} does not make use of additive inverses and generalises without any difficulty to arbitrary semirings. The required axioms are (B1)-(B11) and (E6-11).
\end{proof}

\noindent The equalities (E6-7) and (E10-11) in Fig.~\ref{fig:axioms}
can be extended to any matrix-diagram.
\begin{lem}[Matrix distributivity]\label{lem:matrix-copy}
Any matrix-diagram $d\from \objr^m\to \objr^n$ (cf. Definition~\ref{def:matrix-diagram}) satisfies
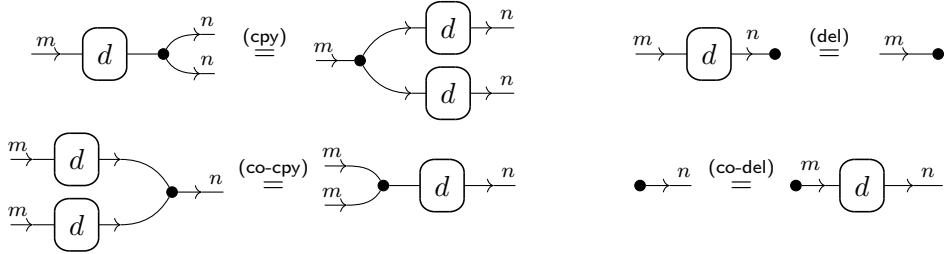
\begin{align*}

\tikzset{x=1em, y=2.1ex}
\InputIfFileExists{global-copy.tikz}{}{\input{./tikz/global-copy.tikz}}
\tikzset{x=1em, y=1.5ex}
\quad \myeq{cpy} \quad 
\tikzset{x=1em, y=2.1ex}
\InputIfFileExists{global-copy-1.tikz}{}{\input{./tikz/global-copy-1.tikz}}
\tikzset{x=1em, y=1.5ex}
 \qquad\qquad 
\tikzset{x=1em, y=2.1ex}
\InputIfFileExists{global-delete.tikz}{}{\input{./tikz/global-delete.tikz}}
\tikzset{x=1em, y=1.5ex}
\quad \myeq{del} \quad 
\tikzset{x=1em, y=2.1ex}
\InputIfFileExists{global-delete-1.tikz}{}{\input{./tikz/global-delete-1.tikz}}
\tikzset{x=1em, y=1.5ex}
\\

\tikzset{x=1em, y=2.1ex}
\InputIfFileExists{global-merge.tikz}{}{\input{./tikz/global-merge.tikz}}
\tikzset{x=1em, y=1.5ex}
\quad \myeq{co-cpy} \quad 
\tikzset{x=1em, y=2.1ex}
\InputIfFileExists{global-merge-1.tikz}{}{\input{./tikz/global-merge-1.tikz}}
\tikzset{x=1em, y=1.5ex}
\qquad\qquad  
\tikzset{x=1em, y=2.1ex}
\begin{tikzpicture}
	\begin{pgfonlayer}{nodelayer}
		\node [style=none] (16) at (-0.75, 0.5) {\scriptsize $n$};
		\node [style=none] (17) at (-1.25, 0) {};
		\node [style=black] (18) at (-2.25, 0) {};
		\node [style=none] (19) at (-0.5, 0) {};
	\end{pgfonlayer}
	\begin{pgfonlayer}{edgelayer}
		\draw (17.center) to (19.center);
		\draw [->] (18) to (17.center);
	\end{pgfonlayer}
\end{tikzpicture}
}
\tikzset{x=1em, y=1.5ex}
\quad \myeq{co-del}\quad 
\tikzset{x=1em, y=2.1ex}
\InputIfFileExists{global-co-delete-1.tikz}{}{\input{./tikz/global-co-delete-1.tikz}}
\tikzset{x=1em, y=1.5ex}

\end{align*}
\end{lem}
\begin{proof}
This lemma can be easily proved by induction, using axioms (B1)-(B11) and local distributivity (equalities (E6-7) and (E10-11)) as base case. But it is also an immediate consequence of Theorem~\ref{thm:matrix-completeness} and of the equivalent semantic statement for matrices with coefficients in $\mathbb{B}(\Sigma^*)$.
\end{proof}
\noindent Given a matrix-diagram $d\from \objr^{l+m}\to \objr^{p+n}$, we will write $d_{ij}$, with $i=l,n$ and $j=p,m$, to refer to the diagram obtained from discarding all but the left $i$-ports with $\Bunit$ and all but the right $j$-ports with $\Bcounit$. For example,
\[d_{m,p} = \quad 
\tikzset{x=1em, y=2.1ex}
\InputIfFileExists{d-projection.tikz}{}{\input{./tikz/d-projection.tikz}}
\tikzset{x=1em, y=1.5ex}
\]
\noindent The following lemma states that this operation selects the corresponding submatrices of (the matrix corresponding to) $d$.
\begin{lem}\label{lem:submatrices}
For any matrix-diagram $d\from \objr^{l+n}\to \objr^{p+m}$, with $d_{ij}$ defined as above.
, we have
\[d = \quad 
\tikzset{x=1em, y=2.1ex}
\InputIfFileExists{biproduct-decomposition.tikz}{}{\input{./tikz/biproduct-decomposition.tikz}}
\tikzset{x=1em, y=1.5ex}
\]
\end{lem}
\begin{proof}
This could be proven from Lemma~\ref{lem:matrix-copy} but we can appeal once again to the corresponding fact for matrices over a semiring and to the completeness of our theory for matrices over  $\mathbb{B}(\Sigma^*)$ (Theorem~\ref{thm:matrix-completeness}) to deduce it immediately.
\end{proof}
\noindent Note that if we discard all but one port on the left and one port on the right, we pick out a dimension-one submatrix, i.e.\ a coefficient, of the corresponding matrix. Then, Lemma~\ref{lem:submatrices} is only saying that matrix-diagrams are fully characterised by their coefficients.

The following lemma establishes a useful form for diagrams.
\begin{lem}[Trace canonical form]\label{lem:traceform}
For any automaton-diagram $c\from \objr^n\to\objr^m$, we can always find a relation-diagram $r\from \objr^{l+n}\to\objr^{l+m}$ such that
\begin{equation*}
\diagbox{c}{n}{m} \quad = \quad \traceaction{r}{n}{m}{l}{x}
\end{equation*}
where $\scalar{x}^l$ denotes a vertical composite of $l$-many $\scalar{a}$ generators.
\end{lem}
\begin{proof}

We reason by structural induction on $\Aut$. For the base case, if $c$ is $\scalar{a}$, we have
\begin{equation*}
\scalar{a} \quad \myeq{A1} \quad 
\tikzset{x=1em, y=2.1ex}
\InputIfFileExists{trace-proof-base-1.tikz}{}{\input{./tikz/trace-proof-base-1.tikz}}
\tikzset{x=1em, y=1.5ex}
\quad=\quad 
\tikzset{x=1em, y=2.1ex}
\InputIfFileExists{trace-proof-base.tikz}{}{\input{./tikz/trace-proof-base.tikz}}
\tikzset{x=1em, y=1.5ex}

\end{equation*}
and every other generator is trivially in the right form, with the trace taken over the $0$ object (the empty list of generators).

There are two inductive cases to consider:
\begin{itemize}
\item $c$ is given by the sequential composition of two morphisms of the appropriate form (using the induction hypothesis). Then
\begin{align*}

\tikzset{x=1em, y=2.1ex}
\InputIfFileExists{trace-proof-comp.tikz}{}{\input{./tikz/trace-proof-comp.tikz}}
\tikzset{x=1em, y=1.5ex}
\quad & =\quad
\tikzset{x=1em, y=2.1ex}
\InputIfFileExists{trace-proof-comp-1.tikz}{}{\input{./tikz/trace-proof-comp-1.tikz}}
\tikzset{x=1em, y=1.5ex}
 \\
& = \quad 
\tikzset{x=1em, y=2.1ex}
\InputIfFileExists{trace-proof-comp-2.tikz}{}{\input{./tikz/trace-proof-comp-2.tikz}}
\tikzset{x=1em, y=1.5ex}

\end{align*}
Here, the composite of the two relation diagrams $s$ and $t$ is also equal to a relation diagram, $r$, by completeness of KDA for matrices over $\mathbb{B}(\Sigma^*)$ (Theorem~\ref{thm:matrix-completeness}) so a fortiori for Boolean matrices (those $\mathbb{B}(\Sigma^*)$-matrices that contain only $0$ or $\epsilon$ coefficients).
\item $c$ is given as the monoidal product of two morphisms of the appropriate form. Then
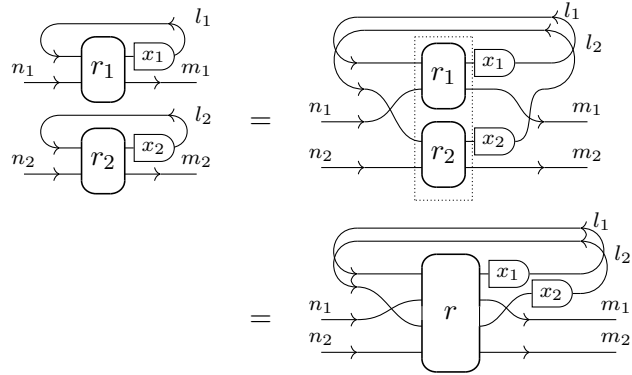
\begin{align*}

\tikzset{x=1em, y=2.1ex}
\InputIfFileExists{trace-proof-tensor.tikz}{}{\input{./tikz/trace-proof-tensor.tikz}}
\tikzset{x=1em, y=1.5ex}
\quad & =\quad
\tikzset{x=1em, y=2.1ex}
\InputIfFileExists{trace-proof-tensor-1.tikz}{}{\input{./tikz/trace-proof-tensor-1.tikz}}
\tikzset{x=1em, y=1.5ex}
\\
& = \quad 
\tikzset{x=1em, y=2.1ex}
\InputIfFileExists{trace-proof-tensor-2.tikz}{}{\input{./tikz/trace-proof-tensor-2.tikz}}
\tikzset{x=1em, y=1.5ex}

\end{align*}
where it is immediate that $r$ is a relation diagram, as product of the two relation diagrams $r_1$ and $r_2$. \qedhere
\end{itemize}
\end{proof}
\noindent We are now ready to prove that any automaton-diagram $c\from \objr \to\objr $ has a representation.
\begin{proof}[Proof of Proposition~\ref{thm:traceform}]
We first rewrite $c$ to trace canonical form (Lemma~\ref{lem:traceform}) \begin{equation}\label{eq:tracecanform}
\dbox{c}  \quad = \quad \traceaction{r}{}{}{k}{x}
\end{equation}
where the relation-diagram $r$ contains no $\scalar{a}$, and therefore factorises as a first layer of comonoid $\Bcomult, \Bcounit$ (potentially followed by some permutations) and a third layer of vertical compositions of the monoid $\Bmult, \Bunit$.

\noindent Then, we can decompose $r\from \objr^{k+1} \to \objr^{k+1}$ as in Lemma~\ref{lem:submatrices} to obtain
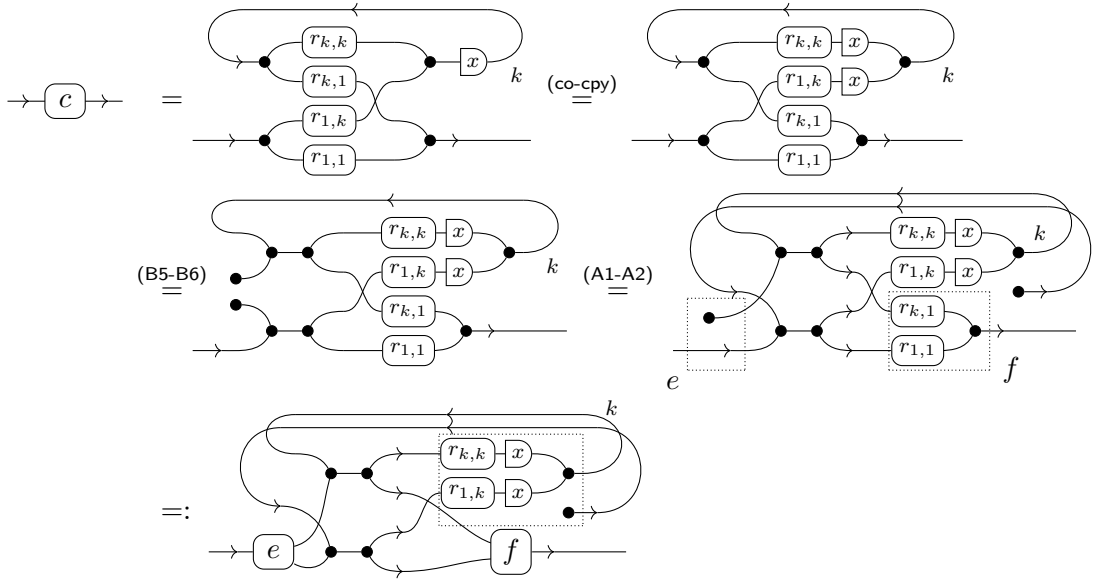
\begin{align*}
\dbox{c} \quad  &= 
\tikzset{x=1em, y=2.1ex}
\InputIfFileExists{trace-decomposition.tikz}{}{\input{./tikz/trace-decomposition.tikz}}
\tikzset{x=1em, y=1.5ex}
 \quad
\myeq{co-cpy} \quad 
\tikzset{x=1em, y=2.1ex}
\InputIfFileExists{proof-representation.tikz}{}{\input{./tikz/proof-representation.tikz}}
\tikzset{x=1em, y=1.5ex}
\\
&\myeq{B5-B6} 
\tikzset{x=1em, y=2.1ex}
\InputIfFileExists{proof-representation-2.tikz}{}{\input{./tikz/proof-representation-2.tikz}}
\tikzset{x=1em, y=1.5ex}

\quad\myeq{A1-A2}\quad 
\tikzset{x=1em, y=2.1ex}
\InputIfFileExists{proof-representation-3.tikz}{}{\input{./tikz/proof-representation-3.tikz}}
\tikzset{x=1em, y=1.5ex}
\\
&=:\; 
\tikzset{x=1em, y=2.1ex}
\InputIfFileExists{proof-representation-4.tikz}{}{\input{./tikz/proof-representation-4.tikz}}
\tikzset{x=1em, y=1.5ex}

\end{align*}
We are very close to obtaining the desired representation but the highlighted sub-diagram in the last diagram is not quite of the right form. Recall that $r_{k,k}$ and $r_{1,k}$ are relation-diagrams, which means that they factor as a block of $\Bcomult, \Bcounit$ composed sequentially with a block of $\Bmult, \Bunit$. Therefore, to obtain an $\epsilon$-free matrix-diagram we can push all the scalars in $\scalar{x}$ into $r_{k,k}$ and $r_{1,k}$ past the $\Bmult, \Bunit$ block, using
(E10)\footnote{This step implements the diagrammatic counterpart of a standard $\epsilon$-elimination procedure for NFA.}.  In doing so, we get two $\epsilon$-free matrix-diagrams---let us call them $d_{k,k}$ and $d_{1,k}$, respectively---and can write:
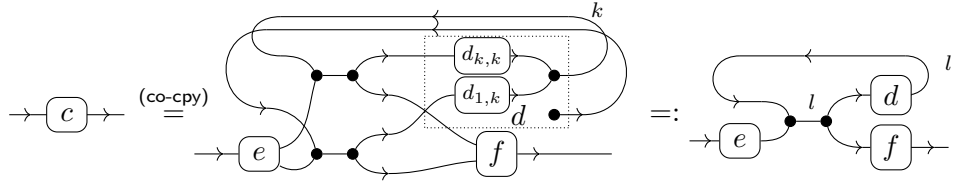
\[\dbox{c} \quad  \myeq{co-cpy} 
\tikzset{x=1em, y=2.1ex}
\InputIfFileExists{proof-representation-5.tikz}{}{\input{./tikz/proof-representation-5.tikz}}
\tikzset{x=1em, y=1.5ex}
 =: 
\tikzset{x=1em, y=2.1ex}
\InputIfFileExists{automaton-representation.tikz}{}{\input{./tikz/automaton-representation.tikz}}
\tikzset{x=1em, y=1.5ex}
 \]
for $l=k+1$. We can see directly from the form of this last expression that $(e, d, f)$ is a representation for $c$.
\end{proof}

\section{Completeness and Determinisation}\label{sec:completeness}

This section is devoted to prove our completeness result, Theorem~\ref{thm:completeness}. We use a normal form argument: more specifically we mimic automata-theoretic results to rewrite every string diagram to a normal form corresponding to a \emph{minimal} deterministic finite automaton (DFA). It is a standard result that, for a given regular language $L$, there is a minimal (in the number of states) DFA which recognises $L$ and that this DFA is unique up to renaming of the states. For a review of this fundamental result, we refer the reader to~\cite[\S 13-16]{kozen2012automata} There are several ways to obtain a minimal DFA that is language-equivalent to a given NFA\@. We will use Brzozowski's algorithm~\cite{brzozowski1962canonical}, which we implement in KDA itself as a sequence of diagrammatic (in)equalities. The proof proceeds in four distinct steps.
\begin{itemize}
\item We first show (Section~\ref{sec:simplify}) how the problem of completeness for all of $\Aut$ can be reduced to that of \emph{equality of $\objr\to\objr$ diagrams}.
\item We then give (Section~\ref{sec:determinisation}) a procedure to \emph{determinise} (the representation of) a diagram: this step consists in eliminating all subdiagrams that correspond to nondeterministic transitions in the associated automaton. For this, we build on the results of Section~\ref{sec:subset-construction}, in which we show that the standard subset construction can be carried out diagrammatically.
\item We use the previous step to implement a \emph{minimisation} procedure (Section~\ref{sec:minimisation}) from which we obtain a minimal representation for a given diagram: this is a representation whose associated automaton is minimal---with the fewest number of states---amongst DFAs that recognise the same language. To do this, we show how the four steps of Brzozowski's minimisation algorithm (reverse; determinise; reverse; determinise) translate into diagrammatic equational reasoning. Note that the first three steps taken together simply amount to applying in reverse the determinisation procedure we have already devised. That this is possible will be a consequence of the symmetry of~$\leqKa$.
\item Finally, from the uniqueness of minimal DFAs, any two diagrams that have the same denotation are both equal to the same minimal representation and we can derive completeness of~$\leqKa$ (Theorem~\ref{thm:completeness}).
\end{itemize}

\begin{rem}
At this point, it is helpful to explain the relationship between the completeness result of this paper and the erroneous claim of~\cite{piedeleu2021string}. In the present paper,  the white nodes play a double role: (1)~they allow us to reduce completeness to automata-diagrams of type $\objr \to \objr$, and (2)~to translate the use of non-equational axioms (in particular the induction axiom of Kozen's axiomatisation) in the proof below, using purely local and equational reasoning steps.

In~\cite{piedeleu2021string}, we used a different syntax: one with another generator (also represented by a white node) interpreted as the action of regular expressions on languages.  This additional generator allowed us to achieve~(1), but was not sufficient to guarantee~(2). Indeed, the proof of~(2) is based on an incorrect claim: the rewriting procedure that is supposed to implement determinisation, as explained in the proof of~\cite[Lemma 4]{piedeleu2021string}, makes unfounded assumptions on the shape of diagrams and is not guaranteed to return the desired determinisation. A counter-example is provided by the diagrammatic representation of, \emph{e.g.}, $(aa)^*(1+a)$. The corresponding diagram should be equal to that representing $a^*$, but this cannot be proven in the equational theory of that paper---we explain this further in Example~\ref{ex:(aa)*(1+a)} below.
\end{rem}

\subsection{Useful preliminaries and simplifying assumptions}\label{sec:simplify}

In this section, we use symmetries of the theory to make simplifying assumptions about the diagrams to consider in the completeness proof.

First, note that we need only consider equalities for completeness, since inequalities can be recovered from the semi-lattice structure of the binary operation defined by $\Bcomult$ and $\Bmult$, both semantically and syntactically as shown by the following two propositions.
\begin{prop}\label{thm:leq-from-eq-sem}
For any two diagrams $c,d$, $\sem{c}\subseteq \sem{d}$ if and only if $\sem{
\tikzset{x=1em, y=2.1ex}
\InputIfFileExists{conv-bcomult-bmult.tikz}{}{\input{./tikz/conv-bcomult-bmult.tikz}}
\tikzset{x=1em, y=1.5ex}
} = \sem{c}$.
\end{prop}
\begin{proof}
A routine calculation shows that $\sem{
\tikzset{x=1em, y=2.1ex}
\InputIfFileExists{conv-bcomult-bmult.tikz}{}{\input{./tikz/conv-bcomult-bmult.tikz}}
\tikzset{x=1em, y=1.5ex}
} = \sem{c}\cap \sem{d}$. So the result follows from $\sem{c}\cap \sem{d} = \sem{c} \Leftrightarrow \sem{c}\subseteq \sem{d}$.
\end{proof}
\begin{prop}\label{thm:leq-from-eq-syn}
For any two automata-diagrams $c,d$ we have $c\leq d$ iff $
\tikzset{x=1em, y=2.1ex}
\InputIfFileExists{conv-bcomult-bmult.tikz}{}{\input{./tikz/conv-bcomult-bmult.tikz}}
\tikzset{x=1em, y=1.5ex}
 = c$.
\end{prop}
\begin{proof}
If $c\leq d$, then $c \myleq{F4} 
\tikzset{x=1em, y=2.1ex}
\InputIfFileExists{conv-bcomult-bmult-c-1.tikz}{}{\input{./tikz/conv-bcomult-bmult-c-1.tikz}}
\tikzset{x=1em, y=1.5ex}
 \myleq{cpy} 
\tikzset{x=1em, y=2.1ex}
\InputIfFileExists{conv-bcomult-bmult-c.tikz}{}{\input{./tikz/conv-bcomult-bmult-c.tikz}}
\tikzset{x=1em, y=1.5ex}
 \leq  
\tikzset{x=1em, y=2.1ex}
\InputIfFileExists{conv-bcomult-bmult.tikz}{}{\input{./tikz/conv-bcomult-bmult.tikz}}
\tikzset{x=1em, y=1.5ex}
= c$. We also have $
\tikzset{x=1em, y=2.1ex}
\InputIfFileExists{conv-bcomult-bmult.tikz}{}{\input{./tikz/conv-bcomult-bmult.tikz}}
\tikzset{x=1em, y=1.5ex}
 \myleq{F2} 
\tikzset{x=1em, y=2.1ex}
\InputIfFileExists{conv-disconnect-d.tikz}{}{\input{./tikz/conv-disconnect-d.tikz}}
\tikzset{x=1em, y=1.5ex}
\myleq{del}  
\tikzset{x=1em, y=2.1ex}
\InputIfFileExists{conv-disconnect-d-1.tikz}{}{\input{./tikz/conv-disconnect-d-1.tikz}}
\tikzset{x=1em, y=1.5ex}
$
Note that we use (cpy) and (del) from Theorem~\ref{thm:copy-merge} below (but in fact, these inequalities and the statement of the proposition, holds for \emph{any} diagram. This can easily be proven by induction---the inductive cases are trivial, so we just have to check that the relevant inequalities holds for all generators. However, we will not need this more general fact here, as we only care about completeness for automata-diagrams).

Conversely if $
\tikzset{x=1em, y=2.1ex}
\InputIfFileExists{conv-bcomult-bmult.tikz}{}{\input{./tikz/conv-bcomult-bmult.tikz}}
\tikzset{x=1em, y=1.5ex}
 = c$, then we can reason as before:
$c = 
\tikzset{x=1em, y=2.1ex}
\InputIfFileExists{conv-bcomult-bmult.tikz}{}{\input{./tikz/conv-bcomult-bmult.tikz}}
\tikzset{x=1em, y=1.5ex}
\myleq{F2}
\tikzset{x=1em, y=2.1ex}
\InputIfFileExists{conv-bcomult-bmult-1.tikz}{}{\input{./tikz/conv-bcomult-bmult-1.tikz}}
\tikzset{x=1em, y=1.5ex}
\myleq{del} 
\tikzset{x=1em, y=2.1ex}
\InputIfFileExists{conv-bcomult-bmult-2.tikz}{}{\input{./tikz/conv-bcomult-bmult-2.tikz}}
\tikzset{x=1em, y=1.5ex}
 = d$.
\end{proof}

Then, we show that, without loss of generality, we can restrict our attention to diagrams of type $\objr\to\objr$. We proceed in two steps: (1)~from all $\Aut$ diagrams to left-to-right diagrams only, and from left-to-right diagrams to those of type $\objr\to\objr$.

\paragraph{From diagrams of $\Aut$ to left-to-right diagrams.} First, the following proposition implies that, without loss of generality, we need only consider to left-to-right diagrams (Section~\ref{sec:automata-diag}).
\begin{prop}\label{thm:left-to-right}
There are natural bijections between sets of string diagrams of the form
\begin{equation*}

\tikzset{x=1em, y=2.1ex}
\InputIfFileExists{wrong-way-left.tikz}{}{\input{./tikz/wrong-way-left.tikz}}
\tikzset{x=1em, y=1.5ex}
\quad\leftrightarrow\quad 
\tikzset{x=1em, y=2.1ex}
\InputIfFileExists{right-way-right.tikz}{}{\input{./tikz/right-way-right.tikz}}
\tikzset{x=1em, y=1.5ex}
 \quad\text{ and }\quad 
\tikzset{x=1em, y=2.1ex}
\InputIfFileExists{wrong-way-right.tikz}{}{\input{./tikz/wrong-way-right.tikz}}
\tikzset{x=1em, y=1.5ex}
\quad\leftrightarrow\quad 
\tikzset{x=1em, y=2.1ex}
\InputIfFileExists{right-way-left.tikz}{}{\input{./tikz/right-way-left.tikz}}
\tikzset{x=1em, y=1.5ex}

\end{equation*}
where $A,B, A_i, B_i$ represent lists of $\objr$ and $\objl$.
\end{prop}
\begin{proof}
This proposition holds in any compact-closed category and relies on the ability to bend wires using $
\tikzset{x=1em, y=2.1ex}
\InputIfFileExists{cap-down.tikz}{}{\input{./tikz/cap-down.tikz}}
\tikzset{x=1em, y=1.5ex}
$ and $
\tikzset{x=1em, y=2.1ex}
\InputIfFileExists{cup-down.tikz}{}{\input{./tikz/cup-down.tikz}}
\tikzset{x=1em, y=1.5ex}
$. Explicitly, given a diagram of the first form, we can obtain one of the second form as follows:
\begin{equation}

\tikzset{x=1em, y=2.1ex}
\InputIfFileExists{wrong-way-left.tikz}{}{\input{./tikz/wrong-way-left.tikz}}
\tikzset{x=1em, y=1.5ex}
\quad \mapsto\quad 
\tikzset{x=1em, y=2.1ex}
\InputIfFileExists{bent-wires.tikz}{}{\input{./tikz/bent-wires.tikz}}
\tikzset{x=1em, y=1.5ex}

\end{equation}
The inverse mapping is given by the same wiring with the opposite direction. That they are inverse transformations follows immediately from the defining equations of compact closed categories (A1)-(A2). The other bijection is constructed analogously.
\end{proof}
\noindent Intuitively, Proposition~\ref{thm:left-to-right} tells us that we can always bend incoming wires to the left and outgoing wires to the right before applying some equations, and recover the original orientation of the wires by bending them into their original place later.

\paragraph{From left-to-right to $\objr\to\objr$.} As we will now show, we can further restrict our attention to diagrams $\objr\to\objr$. For this we prove that any left-to-right diagram $\objr^m\to \objr^n$ is fully characterised by $n\times m$ diagrams $\objr\to\objr$ much like linear maps can be described by their coefficients in a given basis. Showing this amounts to proving that Lemma~\ref{lem:matrix-copy} extends to all left-to-right diagrams (so that the monoidal product is also a biproduct for the subcategory of left-to-right diagrams).

\begin{thm}[Global distributivity]\label{thm:copy-merge}
For any automaton-diagram $d\from \objr^m\to \objr^n$, we have
\begin{align*}

\tikzset{x=1em, y=2.1ex}
\InputIfFileExists{global-copy.tikz}{}{\input{./tikz/global-copy.tikz}}
\tikzset{x=1em, y=1.5ex}
\quad \myeq{cpy} \quad 
\tikzset{x=1em, y=2.1ex}
\InputIfFileExists{global-copy-1.tikz}{}{\input{./tikz/global-copy-1.tikz}}
\tikzset{x=1em, y=1.5ex}
 \qquad\qquad 
\tikzset{x=1em, y=2.1ex}
\InputIfFileExists{global-delete.tikz}{}{\input{./tikz/global-delete.tikz}}
\tikzset{x=1em, y=1.5ex}
\quad \myeq{del} \quad 
\tikzset{x=1em, y=2.1ex}
\InputIfFileExists{global-delete-1.tikz}{}{\input{./tikz/global-delete-1.tikz}}
\tikzset{x=1em, y=1.5ex}
\\

\tikzset{x=1em, y=2.1ex}
\InputIfFileExists{global-merge.tikz}{}{\input{./tikz/global-merge.tikz}}
\tikzset{x=1em, y=1.5ex}
\quad \myeq{co-cpy} \quad 
\tikzset{x=1em, y=2.1ex}
\InputIfFileExists{global-merge-1.tikz}{}{\input{./tikz/global-merge-1.tikz}}
\tikzset{x=1em, y=1.5ex}
\qquad\qquad  
\tikzset{x=1em, y=2.1ex}
\begin{tikzpicture}
	\begin{pgfonlayer}{nodelayer}
		\node [style=none] (16) at (-0.75, 0.5) {\scriptsize $n$};
		\node [style=none] (17) at (-1.25, 0) {};
		\node [style=black] (18) at (-2.25, 0) {};
		\node [style=none] (19) at (-0.5, 0) {};
	\end{pgfonlayer}
	\begin{pgfonlayer}{edgelayer}
		\draw (17.center) to (19.center);
		\draw [->] (18) to (17.center);
	\end{pgfonlayer}
\end{tikzpicture}
}
\tikzset{x=1em, y=1.5ex}
\quad \myeq{co-del} \quad 
\tikzset{x=1em, y=2.1ex}
\InputIfFileExists{global-co-delete-1.tikz}{}{\input{./tikz/global-co-delete-1.tikz}}
\tikzset{x=1em, y=1.5ex}

\end{align*}
\end{thm}
\begin{proof}
According to Lemma~\ref{lem:traceform}, given $d$ as in the statement of the theorem, we can find a relation-diagram $r$ such that
\begin{equation}
\diagbox{d}{n}{m} \quad = \quad \traceaction{r}{n}{m}{l}{x}
\end{equation}
 Note first that, by Lemma~\ref{lem:matrix-copy}, any relation-diagram satisfies (cpy) and (del) so we will use these two equations for $r$ below.

\medskip

\noindent First, we prove both inequalities of (cpy).
\begin{itemize}
\item The first inequality requires the introduction of new black nodes, via the two axioms $
\tikzset{x=1em, y=2.1ex}
\InputIfFileExists{bmult-bcomult.tikz}{}{\input{./tikz/bmult-bcomult.tikz}}
\tikzset{x=1em, y=1.5ex}
 \:\myleq{F1}\: 
\tikzset{x=1em, y=2.1ex}
\InputIfFileExists{id-2.tikz}{}{\input{./tikz/id-2.tikz}}
\tikzset{x=1em, y=1.5ex}
$ and $\idone \:\myleq{F3}\: 
\tikzset{x=1em, y=2.1ex}
\InputIfFileExists{bcomult-bmult.tikz}{}{\input{./tikz/bcomult-bmult.tikz}}
\tikzset{x=1em, y=1.5ex}
$:
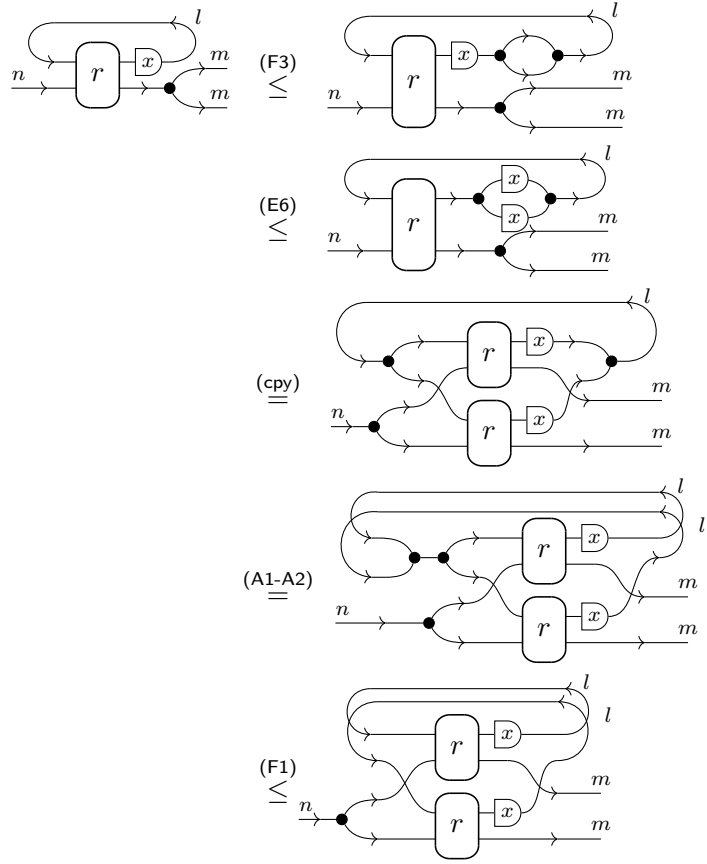
\begin{align*}

\tikzset{x=1em, y=2.1ex}
\InputIfFileExists{traceform-copy.tikz}{}{\input{./tikz/traceform-copy.tikz}}
\tikzset{x=1em, y=1.5ex}
\quad &\myleq{F3}\quad
\tikzset{x=1em, y=2.1ex}
\InputIfFileExists{traceform-copy-1.tikz}{}{\input{./tikz/traceform-copy-1.tikz}}
\tikzset{x=1em, y=1.5ex}
\\
&\myleq{E6}\quad
\tikzset{x=1em, y=2.1ex}
\InputIfFileExists{traceform-copy-2.tikz}{}{\input{./tikz/traceform-copy-2.tikz}}
\tikzset{x=1em, y=1.5ex}
\\
&\myeq{cpy}\quad
\tikzset{x=1em, y=2.1ex}
\InputIfFileExists{traceform-copy-3.tikz}{}{\input{./tikz/traceform-copy-3.tikz}}
\tikzset{x=1em, y=1.5ex}
\\
&\myeq{A1-A2}\quad
\tikzset{x=1em, y=2.1ex}
\InputIfFileExists{traceform-copy-4.tikz}{}{\input{./tikz/traceform-copy-4.tikz}}
\tikzset{x=1em, y=1.5ex}
\\
& \myleq{F1}
\tikzset{x=1em, y=2.1ex}
\InputIfFileExists{copy-traceform.tikz}{}{\input{./tikz/copy-traceform.tikz}}
\tikzset{x=1em, y=1.5ex}

\end{align*}
\item The reverse inequality requires the introduction and elimination of $\Wmult$, via the two axioms $
\tikzset{x=1em, y=2.1ex}
\InputIfFileExists{id-2.tikz}{}{\input{./tikz/id-2.tikz}}
\tikzset{x=1em, y=1.5ex}
\:\myleq{F5} \:
\tikzset{x=1em, y=2.1ex}
\InputIfFileExists{wmult-bcomult.tikz}{}{\input{./tikz/wmult-bcomult.tikz}}
\tikzset{x=1em, y=1.5ex}
$ and $d
\tikzset{x=1em, y=2.1ex}
\InputIfFileExists{bcomult-wmult.tikz}{}{\input{./tikz/bcomult-wmult.tikz}}
\tikzset{x=1em, y=1.5ex}
\:\myleq{F7}\: \idone$:
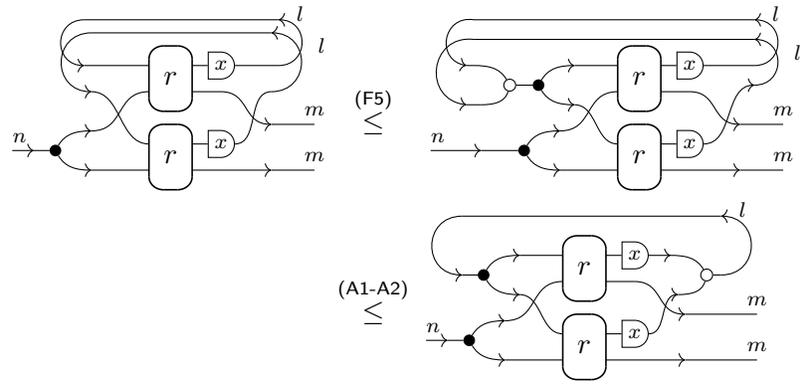
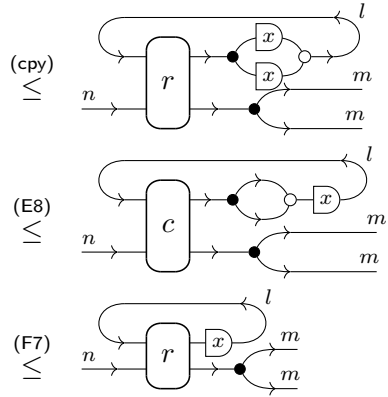
\begin{align*}

\tikzset{x=1em, y=2.1ex}
\InputIfFileExists{copy-traceform.tikz}{}{\input{./tikz/copy-traceform.tikz}}
\tikzset{x=1em, y=1.5ex}
\quad &\myleq{F5}\quad
\tikzset{x=1em, y=2.1ex}
\InputIfFileExists{copy-traceform-1.tikz}{}{\input{./tikz/copy-traceform-1.tikz}}
\tikzset{x=1em, y=1.5ex}
\\
&\myleq{A1-A2}\quad
\tikzset{x=1em, y=2.1ex}
\InputIfFileExists{copy-traceform-2.tikz}{}{\input{./tikz/copy-traceform-2.tikz}}
\tikzset{x=1em, y=1.5ex}
\\
&\myleq{cpy}\quad
\tikzset{x=1em, y=2.1ex}
\InputIfFileExists{copy-traceform-3.tikz}{}{\input{./tikz/copy-traceform-3.tikz}}
\tikzset{x=1em, y=1.5ex}
\\
&\myleq{E8}\quad
\tikzset{x=1em, y=2.1ex}
\InputIfFileExists{copy-traceform-4.tikz}{}{\input{./tikz/copy-traceform-4.tikz}}
\tikzset{x=1em, y=1.5ex}
\\
&\myleq{F7}\quad 
\tikzset{x=1em, y=2.1ex}
\InputIfFileExists{traceform-copy.tikz}{}{\input{./tikz/traceform-copy.tikz}}
\tikzset{x=1em, y=1.5ex}

\end{align*}
\end{itemize}
We can prove (del) in a similar way, as follows.
\begin{itemize}
\item The first inequality is the unary version of its (cpy) counterpart, using axioms $\idone \:\myleq{F2}\:\Bcounit\:\:\Bunit$ and $
\tikzset{x=1em, y=2.1ex}
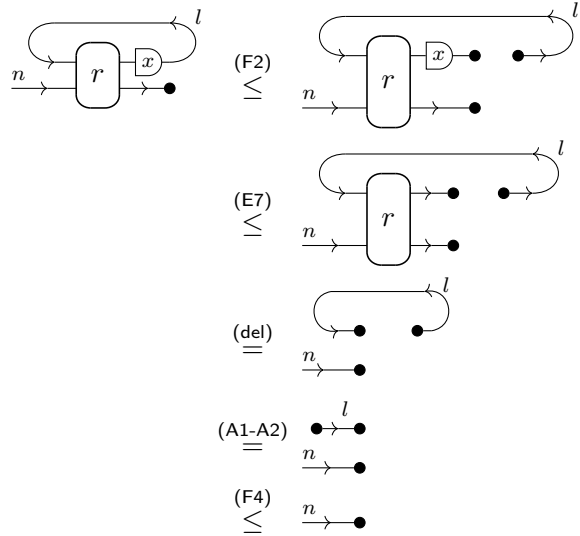
\begin{tikzpicture}
	\begin{pgfonlayer}{nodelayer}
		\node [style=black] (8) at (0, 0) {};
		\node [style=black] (17) at (-1.5, 0) {};
		\node [style=none] (18) at (-0.5, 0) {};
	\end{pgfonlayer}
	\begin{pgfonlayer}{edgelayer}
		\draw [->] (17) to (18.center);
		\draw (18.center) to (8);
	\end{pgfonlayer}
\end{tikzpicture}
}
\tikzset{x=1em, y=1.5ex}
 \:\myleq{F4}\: 
\tikzset{x=1em, y=2.1ex}
\InputIfFileExists{empty-diag.tikz}{}{\input{./tikz/empty-diag.tikz}}
\tikzset{x=1em, y=1.5ex}
$:
\begin{align*}

\tikzset{x=1em, y=2.1ex}
\InputIfFileExists{traceform-del.tikz}{}{\input{./tikz/traceform-del.tikz}}
\tikzset{x=1em, y=1.5ex}
\quad &\myleq{F2}\quad
\tikzset{x=1em, y=2.1ex}
\InputIfFileExists{traceform-del-1.tikz}{}{\input{./tikz/traceform-del-1.tikz}}
\tikzset{x=1em, y=1.5ex}
\\
&\myleq{E7}\quad
\tikzset{x=1em, y=2.1ex}
\InputIfFileExists{traceform-del-2.tikz}{}{\input{./tikz/traceform-del-2.tikz}}
\tikzset{x=1em, y=1.5ex}
\\
&\myeq{del}\quad
\tikzset{x=1em, y=2.1ex}
\InputIfFileExists{traceform-del-3.tikz}{}{\input{./tikz/traceform-del-3.tikz}}
\tikzset{x=1em, y=1.5ex}
\\
&\myeq{A1-A2}\quad
\tikzset{x=1em, y=2.1ex}
\InputIfFileExists{traceform-del-4.tikz}{}{\input{./tikz/traceform-del-4.tikz}}
\tikzset{x=1em, y=1.5ex}
\\
& \myleq{F4}\quad 
\tikzset{x=1em, y=2.1ex}
\begin{tikzpicture}
	\begin{pgfonlayer}{nodelayer}
		\node [style=black] (49) at (0.75, 0) {};
		\node [style=none] (56) at (-1, 0.5) {\scriptsize  $n$};
		\node [style=none] (57) at (-1.25, 0) {};
		\node [style=none] (58) at (-0.25, 0) {};
	\end{pgfonlayer}
	\begin{pgfonlayer}{edgelayer}
		\draw (58.center) to (49);
		\draw [->] (57.center) to (58.center);
	\end{pgfonlayer}
\end{tikzpicture}
}
\tikzset{x=1em, y=1.5ex}

\end{align*}
\item The reverse inequality requires the introduction and elimination of $\Wunit$, using axioms $\Bcounit\:\:\Wunit \:\myleq{F6}\: \idone$ and $
\tikzset{x=1em, y=2.1ex}
\InputIfFileExists{empty-diag.tikz}{}{\input{./tikz/empty-diag.tikz}}
\tikzset{x=1em, y=1.5ex}
 \:\myleq{F8}\:
\tikzset{x=1em, y=2.1ex}
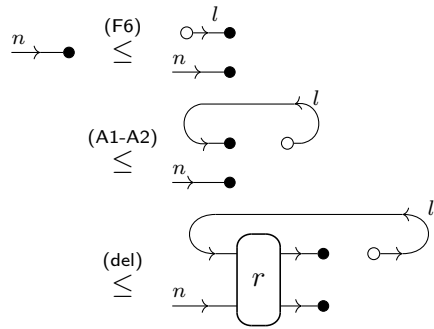
\begin{tikzpicture}
	\begin{pgfonlayer}{nodelayer}
		\node [style=black] (8) at (0, 0) {};
		\node [style=white-dot] (17) at (-1.5, 0) {};
		\node [style=none] (18) at (-0.5, 0) {};
	\end{pgfonlayer}
	\begin{pgfonlayer}{edgelayer}
		\draw [->] (17) to (18.center);
		\draw (18.center) to (8);
	\end{pgfonlayer}
\end{tikzpicture}
}
\tikzset{x=1em, y=1.5ex}
$:
\begin{align*}

\tikzset{x=1em, y=2.1ex}
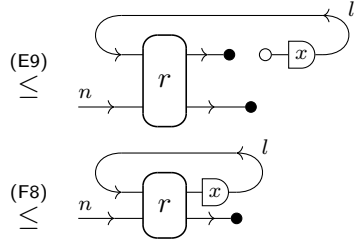
\begin{tikzpicture}
	\begin{pgfonlayer}{nodelayer}
		\node [style=black] (49) at (0.75, 0) {};
		\node [style=none] (56) at (-1, 0.5) {\scriptsize  $n$};
		\node [style=none] (57) at (-1.25, 0) {};
		\node [style=none] (58) at (-0.25, 0) {};
	\end{pgfonlayer}
	\begin{pgfonlayer}{edgelayer}
		\draw (58.center) to (49);
		\draw [->] (57.center) to (58.center);
	\end{pgfonlayer}
\end{tikzpicture}
}
\tikzset{x=1em, y=1.5ex}
\quad &\myleq{F6}\quad
\tikzset{x=1em, y=2.1ex}
\InputIfFileExists{del-traceform-1.tikz}{}{\input{./tikz/del-traceform-1.tikz}}
\tikzset{x=1em, y=1.5ex}
\\
&\myleq{A1-A2}\quad
\tikzset{x=1em, y=2.1ex}
\InputIfFileExists{del-traceform-2.tikz}{}{\input{./tikz/del-traceform-2.tikz}}
\tikzset{x=1em, y=1.5ex}
\\
&\myleq{del}\quad
\tikzset{x=1em, y=2.1ex}
\InputIfFileExists{del-traceform-3.tikz}{}{\input{./tikz/del-traceform-3.tikz}}
\tikzset{x=1em, y=1.5ex}
\\
&\myleq{E9}\quad
\tikzset{x=1em, y=2.1ex}
\InputIfFileExists{del-traceform-4.tikz}{}{\input{./tikz/del-traceform-4.tikz}}
\tikzset{x=1em, y=1.5ex}
\\
&\myleq{F8}\quad 
\tikzset{x=1em, y=2.1ex}
\InputIfFileExists{traceform-del.tikz}{}{\input{./tikz/traceform-del.tikz}}
\tikzset{x=1em, y=1.5ex}

\end{align*}
\end{itemize}
The other two equalities---(co-cpy) and (co-del)---can be proved by a symmetric argument, replacing $\Bmult$ with $\Bcomult$, $\Bunit$ with $\Bcounit$, axioms (F9) instead of (F5), (F11) instead of (F7), (F10) instead of (F6), and (F12) instead of (F8).
\end{proof}

For $d\from \objr^m\to \objr^n$, let $d_{ij}$ be the string diagram of type $\objr\to \objr$ obtained as in Lemma~\ref{lem:submatrices}, by discarding every input and output except the $i$-th input and $j$-th output, i.e., by composing every input with $
\tikzset{x=1em, y=2.1ex}
\begin{tikzpicture}
	\begin{pgfonlayer}{nodelayer}
		\node [style=black] (37) at (-0.5, 0) {};
		\node [style=none] (43) at (0.25, 0) {};
		\node [style=none] (44) at (0.75, 0) {};
	\end{pgfonlayer}
	\begin{pgfonlayer}{edgelayer}
		\draw [->] (37) to (43.center);
		\draw (44.center) to (43.center);
	\end{pgfonlayer}
\end{tikzpicture}
}
\tikzset{x=1em, y=1.5ex}
$ except the $i$-th one, and every output with $
\tikzset{x=1em, y=2.1ex}
\begin{tikzpicture}
	\begin{pgfonlayer}{nodelayer}
		\node [style=black] (37) at (0.75, 0) {};
		\node [style=none] (43) at (0.25, 0) {};
		\node [style=none] (44) at (-0.5, 0) {};
	\end{pgfonlayer}
	\begin{pgfonlayer}{edgelayer}
		\draw (43.center) to (37);
		\draw [->] (44.center) to (43.center);
	\end{pgfonlayer}
\end{tikzpicture}
}
\tikzset{x=1em, y=1.5ex}
$ except the $j$-th one. Theorem~\ref{thm:copy-merge} implies that left-to-right diagrams, like matrix-diagrams, are fully characterised by their $\objr\to\objr$ subdiagrams.
\begin{cor}\label{thm:1-to-1-restrict}
Given automata-diagrams $d,e\from \objr^m\to \objr^n$, $d\eqKa e$ iff $d_{ij} \eqKa e_{ij}$, for all $1\leq i\leq m$ and $1\leq j\leq n$.
\end{cor}
Thus, as we claimed above, we can restrict our focus further to left-to-right $\objr\to\objr$ diagrams, without loss of generality. Therefore, to prove Theorem~\ref{thm:completeness}, we only need to to prove the following result.
\begin{thm}\label{thm:1-to-1-completeness}
For any two automata-diagrams $d, d'\from\objr\to\objr$,
\[\sem{d}= \sem{d'} \text{ if and only if } d \eqKa d'\text{.}\]
\end{thm}
We will need to prove several preparatory results, including a diagrammatic form of determinisation, before the proof of Theorem~\ref{thm:1-to-1-completeness} which can be found in Section~\ref{sec:minimisation}.

\subsection{Diagrammatic subset construction}%
\label{sec:subset-construction}

In what follows we assume familiarity with the standard subset construction. The reader who wishes to refresh their memory can refer to~\cite[\S 6]{kozen2012automata}.

In diagrammatic terms, a nondeterministic transition of the automaton associated to (a representation of) a given diagram, corresponds to a subdiagram of the form $
\tikzset{x=1em, y=2.1ex}
\InputIfFileExists{non-determinism.tikz}{}{\input{./tikz/non-determinism.tikz}}
\tikzset{x=1em, y=1.5ex}
$ for some $a\in \Sigma$ in the matrix-diagram encoding its transition relation. The following example illustrates how Theorem~\ref{thm:copy-merge} can already be used to determinise some simple automata-diagrams. The following section is dedicated to giving a formal procedure that extends this idea to any automaton-diagram, by formalising a diagrammatic version of the algebraic subset construction due to Kozen~\cite{kozen1994completeness}.
\begin{exa}\label{ex:determinisation}
\begin{align*}

\tikzset{x=1em, y=2.1ex}
\InputIfFileExists{ex-graph-determinise.tikz}{}{\input{./tikz/ex-graph-determinise.tikz}}
\tikzset{x=1em, y=1.5ex}
\quad \mapsto\quad 
\tikzset{x=1em, y=2.1ex}
\InputIfFileExists{ex-diag-determinise.tikz}{}{\input{./tikz/ex-diag-determinise.tikz}}
\tikzset{x=1em, y=1.5ex}

 \myeq{cpy} \;
\tikzset{x=1em, y=2.1ex}
\InputIfFileExists{ex-diag-determinise-3.tikz}{}{\input{./tikz/ex-diag-determinise-3.tikz}}
\tikzset{x=1em, y=1.5ex}
\\
 \;\myeq{cpy}\;
\tikzset{x=1em, y=2.1ex}
\InputIfFileExists{ex-diag-determinise-4.tikz}{}{\input{./tikz/ex-diag-determinise-4.tikz}}
\tikzset{x=1em, y=1.5ex}

 := \quad
\tikzset{x=1em, y=2.1ex}
\InputIfFileExists{ex-diag-determinise-5.tikz}{}{\input{./tikz/ex-diag-determinise-5.tikz}}
\tikzset{x=1em, y=1.5ex}

\mapsfrom\quad 
\tikzset{x=1em, y=2.1ex}
\InputIfFileExists{ex-graph-deterministic.tikz}{}{\input{./tikz/ex-graph-deterministic.tikz}}
\tikzset{x=1em, y=1.5ex}

\end{align*}
where we write $\scalar{a^*}$ for the left-to-right diagram $
\tikzset{x=1em, y=2.1ex}
\InputIfFileExists{a-star.tikz}{}{\input{./tikz/a-star.tikz}}
\tikzset{x=1em, y=1.5ex}
$.
\end{exa}

The diagrammatic counterpart of the subset construction we give below makes crucial use of the adjunctions between the different generators (Block F in Fig.~\ref{fig:axioms}). Before we can give the determinisation procedure for an arbitrary automaton-diagram, we need to cover essential technical preliminaries, which will allow us to greatly generalise these adjunctions.

Recall from Section~\ref{sec:encoding-relations} that, given a bimonoid, we can encode $n\times m$ Boolean matrices---equivalently, relations between the finite sets $\{0,\dots, m-1\}$ and $\{0,\dots, n-1\}$---by a block of comultiplications and counits composed sequentially with a block of multiplications and units. The $i$-th open port on the right is connected to $j$-th one on the left iff $(i,j)$ is in the encoded relation. This time we will be working with three different bimonoids, giving three different encodings of relations: $(\Bcomult,\Bcounit, \Bmult, \Bunit)$, $(\Bcomult,\Bcounit, \Wmult, \Wunit)$, and $(\Wcomult,\Wcounit, \Bmult, \Bunit)$. We will call the corresponding matrix-diagrams $(\bullet, \bullet)$-matrices\footnote{These are the \emph{relation-diagrams} of Section~\ref{sec:diagram-to-automata}.}, $(\bullet, \circ)$-matrices and $(\circ, \bullet)$-matrices respectively.

We will now prove that the adjunctions between the white and black generators (Block F in Fig.~\ref{fig:axioms}) generalise to these matrix-encodings. We define two notions of transpose for diagrams: one which swaps the colours of the different generators and one which does not. These will assist us in generalising the adjunctions between the white and black generators to all our matrix encodings of relations.
\begin{defi}\label{def:transpose}
Given a diagram $d\from \objr^m\to \objr^n$, we define its \emph{transpose} to be the diagram $d^T\from \objr^n\to \objr^m$ obtained by flipping $d$ horizontally, except the letters $\scalar{a}$. More formally, $(\cdot)^T$ is defined inductively as follows:
\begin{gather*}
(\Bcomult)^T =\Bmult \quad (\Bcounit)^T=\Bunit \quad (\Bmult)^T =\Bcomult \quad (\Bunit)^T=\Bcounit
\\
(\Wcomult)^T =\Wmult \quad (\Wcounit)^T=\Wunit \quad (\Wmult)^T =\Wcomult \quad (\Wunit)^T=\Wcounit
\\
\left(
\tikzset{x=1em, y=2.1ex}
\InputIfFileExists{cup-down.tikz}{}{\input{./tikz/cup-down.tikz}}
\tikzset{x=1em, y=1.5ex}
\right)^T = 
\tikzset{x=1em, y=2.1ex}
\InputIfFileExists{cap-down.tikz}{}{\input{./tikz/cap-down.tikz}}
\tikzset{x=1em, y=1.5ex}
 \qquad \left(
\tikzset{x=1em, y=2.1ex}
\InputIfFileExists{cap-down.tikz}{}{\input{./tikz/cap-down.tikz}}
\tikzset{x=1em, y=1.5ex}
\right)^T = 
\tikzset{x=1em, y=2.1ex}
\InputIfFileExists{cup-down.tikz}{}{\input{./tikz/cup-down.tikz}}
\tikzset{x=1em, y=1.5ex}
 \qquad (\scalar{a})^T = \scalar{a}
\end{gather*}
\[
(\sym)^T = \sym \quad (c\poi d)^T = d^T \poi c^T \quad (c_1\oplus c_2)^T =c_1^T \oplus c_2^T
\]
\end{defi}
\begin{defi}\label{def:color-transpose}
Given a diagram $d\from \objr^m\to \objr^n$, we define its \emph{colour-transpose} to be the diagram $d^\circ\from \objr^n\to \objr^m$ obtained from $d$ by swapping all black and white nodes, and flipping the resulting diagram horizontally, except the letters $\scalar{a}$. More formally, $(\cdot)^\circ$ is defined inductively as follows:
 \begin{gather*}
(\Bcomult)^\circ =\Wmult \quad (\Bcounit)^\circ=\Wunit \quad (\Bmult)^\circ =\Wcomult \quad (\Bunit)^\circ=\Wcounit
\\
(\Wcomult)^\circ =\Bmult \quad (\Wcounit)^\circ=\Bunit \quad (\Wmult)^\circ =\Bcomult \quad (\Wunit)^\circ=\Bcounit
\\
\left(
\tikzset{x=1em, y=2.1ex}
\InputIfFileExists{cup-down.tikz}{}{\input{./tikz/cup-down.tikz}}
\tikzset{x=1em, y=1.5ex}
\right)^\circ = 
\tikzset{x=1em, y=2.1ex}
\InputIfFileExists{cap-down.tikz}{}{\input{./tikz/cap-down.tikz}}
\tikzset{x=1em, y=1.5ex}
 \qquad \left(
\tikzset{x=1em, y=2.1ex}
\InputIfFileExists{cap-down.tikz}{}{\input{./tikz/cap-down.tikz}}
\tikzset{x=1em, y=1.5ex}
\right)^\circ = 
\tikzset{x=1em, y=2.1ex}
\InputIfFileExists{cup-down.tikz}{}{\input{./tikz/cup-down.tikz}}
\tikzset{x=1em, y=1.5ex}
 \qquad (\scalar{a})^\circ = \scalar{a}
\end{gather*}
\[
(\sym)^\circ = \sym \quad (c\poi d)^\circ = d^\circ\poi c^\circ \quad (c_1\oplus c_2)^\circ =c_1^\circ \oplus c_2^\circ
\]
\end{defi}
\noindent Note that a permutation $\sigma$ is mapped to its inverse so that $\sigma^T = \sigma^\circ = \sigma^{-1}$.
Another immediate consequence of these definitions is that these mappings are involutive: $(d^\circ)^\circ = d$ and $(d^T)^T = d$ for any diagram $d$.

For the next lemmas, we use the following notation. Given some relation, we write $
\tikzset{x=1em, y=2.1ex}
\begin{tikzpicture}
	\begin{pgfonlayer}{nodelayer}
		\node [style=none] (0) at (1.25, 0) {};
		\node [style=regb] (1) at (0, 0) {};
		\node [style=none] (2) at (-1.25, 0) {};
	\end{pgfonlayer}
	\begin{pgfonlayer}{edgelayer}
		\draw (2.center) to (1);
		\draw (1) to (0.center);
	\end{pgfonlayer}
\end{tikzpicture}
}
\tikzset{x=1em, y=1.5ex}
$ for the corresponding $(\bullet,\bullet)$-matrix. We can assume it factors as $
\tikzset{x=1em, y=2.1ex}
\begin{tikzpicture}
	\begin{pgfonlayer}{nodelayer}
		\node [style=none] (0) at (1.25, 0) {};
		\node [style=regb] (1) at (0, 0) {};
		\node [style=none] (2) at (-1.25, 0) {};
	\end{pgfonlayer}
	\begin{pgfonlayer}{edgelayer}
		\draw (2.center) to (1);
		\draw (1) to (0.center);
	\end{pgfonlayer}
\end{tikzpicture}
}
\tikzset{x=1em, y=1.5ex}
 = b \poi c$, where $b$ is a diagram formed only of $\Bcomult,\Bcounit$ and $c$ is a diagram formed only of $\Bmult,\Bunit$. Then, we define $
\tikzset{x=1em, y=2.1ex}
\begin{tikzpicture}
	\begin{pgfonlayer}{nodelayer}
		\node [style=none] (0) at (1.25, 0) {};
		\node [style=coregbw] (1) at (0, 0) {};
		\node [style=none] (2) at (-1.25, 0) {};
	\end{pgfonlayer}
	\begin{pgfonlayer}{edgelayer}
		\draw (2.center) to (1);
		\draw (1) to (0.center);
	\end{pgfonlayer}
\end{tikzpicture}
}
\tikzset{x=1em, y=1.5ex}
 := c^T\poi b^\circ$ and $
\tikzset{x=1em, y=2.1ex}
\begin{tikzpicture}
	\begin{pgfonlayer}{nodelayer}
		\node [style=none] (0) at (1.25, 0) {};
		\node [style=coregwb] (1) at (0, 0) {};
		\node [style=none] (2) at (-1.25, 0) {};
	\end{pgfonlayer}
	\begin{pgfonlayer}{edgelayer}
		\draw (2.center) to (1);
		\draw (1) to (0.center);
	\end{pgfonlayer}
\end{tikzpicture}
}
\tikzset{x=1em, y=1.5ex}
 := c^\circ\poi b^T$.

It is helpful to adopt a semantic point of view in order to develop intuition about these different diagrammatic encodings of relations. If $
\tikzset{x=1em, y=2.1ex}
\begin{tikzpicture}
	\begin{pgfonlayer}{nodelayer}
		\node [style=none] (0) at (1.25, 0) {};
		\node [style=regb] (1) at (0, 0) {};
		\node [style=none] (2) at (-1.25, 0) {};
	\end{pgfonlayer}
	\begin{pgfonlayer}{edgelayer}
		\draw (2.center) to (1);
		\draw (1) to (0.center);
	\end{pgfonlayer}
\end{tikzpicture}
}
\tikzset{x=1em, y=1.5ex}
$ encodes the relation $R$, we have
\[
\sem{
\tikzset{x=1em, y=2.1ex}
\begin{tikzpicture}
	\begin{pgfonlayer}{nodelayer}
		\node [style=none] (0) at (1.25, 0) {};
		\node [style=regb] (1) at (0, 0) {};
		\node [style=none] (2) at (-1.25, 0) {};
	\end{pgfonlayer}
	\begin{pgfonlayer}{edgelayer}
		\draw (2.center) to (1);
		\draw (1) to (0.center);
	\end{pgfonlayer}
\end{tikzpicture}
}
\tikzset{x=1em, y=1.5ex}
} = \{(\mathbf{L},\mathbf{K})\mid L_i \subseteq K_j \text{ if } (i,j)\in R\}\]
\[\sem{
\tikzset{x=1em, y=2.1ex}
\begin{tikzpicture}
	\begin{pgfonlayer}{nodelayer}
		\node [style=none] (0) at (1.25, 0) {};
		\node [style=coregbw] (1) at (0, 0) {};
		\node [style=none] (2) at (-1.25, 0) {};
	\end{pgfonlayer}
	\begin{pgfonlayer}{edgelayer}
		\draw (2.center) to (1);
		\draw (1) to (0.center);
	\end{pgfonlayer}
\end{tikzpicture}
}
\tikzset{x=1em, y=1.5ex}
} = \left\{(\mathbf{L},\mathbf{K})\mid \bigcap_{(i,j)\in R} L_j \subseteq K_i\right\}\qquad
\sem{
\tikzset{x=1em, y=2.1ex}
\begin{tikzpicture}
	\begin{pgfonlayer}{nodelayer}
		\node [style=none] (0) at (1.25, 0) {};
		\node [style=coregwb] (1) at (0, 0) {};
		\node [style=none] (2) at (-1.25, 0) {};
	\end{pgfonlayer}
	\begin{pgfonlayer}{edgelayer}
		\draw (2.center) to (1);
		\draw (1) to (0.center);
	\end{pgfonlayer}
\end{tikzpicture}
}
\tikzset{x=1em, y=1.5ex}
} = \left \{(\mathbf{L},\mathbf{K})\mid  L_j \subseteq \bigcup_{(i,j)\in R} K_i\right\}
\]
(Notice the reversal of the indices $i$ and $j$  for the adjoint matrices). Furthermore, the matrix $
\tikzset{x=1em, y=2.1ex}
\begin{tikzpicture}
	\begin{pgfonlayer}{nodelayer}
		\node [style=none] (0) at (1.25, 0) {};
		\node [style=regb] (1) at (0, 0) {};
		\node [style=none] (2) at (-1.25, 0) {};
	\end{pgfonlayer}
	\begin{pgfonlayer}{edgelayer}
		\draw (2.center) to (1);
		\draw (1) to (0.center);
	\end{pgfonlayer}
\end{tikzpicture}
}
\tikzset{x=1em, y=1.5ex}
$ can be seen as the embedding of a monotone map $\Lang^m\to \Lang^{n}$ (see Remark~\ref{rmk:monotone-maps}) in monotone relations.
In fact, the corresponding map is not just monotone but a lattice homomorphism, \emph{i.e.} a map that preserves both meets/intersections and joins/unions. As a result, it admits both left ($
\tikzset{x=1em, y=2.1ex}
\begin{tikzpicture}
	\begin{pgfonlayer}{nodelayer}
		\node [style=none] (0) at (1.25, 0) {};
		\node [style=coregbw] (1) at (0, 0) {};
		\node [style=none] (2) at (-1.25, 0) {};
	\end{pgfonlayer}
	\begin{pgfonlayer}{edgelayer}
		\draw (2.center) to (1);
		\draw (1) to (0.center);
	\end{pgfonlayer}
\end{tikzpicture}
}
\tikzset{x=1em, y=1.5ex}
$) and right ($
\tikzset{x=1em, y=2.1ex}
\begin{tikzpicture}
	\begin{pgfonlayer}{nodelayer}
		\node [style=none] (0) at (1.25, 0) {};
		\node [style=coregwb] (1) at (0, 0) {};
		\node [style=none] (2) at (-1.25, 0) {};
	\end{pgfonlayer}
	\begin{pgfonlayer}{edgelayer}
		\draw (2.center) to (1);
		\draw (1) to (0.center);
	\end{pgfonlayer}
\end{tikzpicture}
}
\tikzset{x=1em, y=1.5ex}
$) adjoints.

From an equational perspective, as we will now prove, the F axioms of KDA are enough to derive $
\tikzset{x=1em, y=2.1ex}
\begin{tikzpicture}
	\begin{pgfonlayer}{nodelayer}
		\node [style=none] (0) at (1.25, 0) {};
		\node [style=coregbw] (1) at (0, 0) {};
		\node [style=none] (2) at (-1.25, 0) {};
	\end{pgfonlayer}
	\begin{pgfonlayer}{edgelayer}
		\draw (2.center) to (1);
		\draw (1) to (0.center);
	\end{pgfonlayer}
\end{tikzpicture}
}
\tikzset{x=1em, y=1.5ex}
\dashv 
\tikzset{x=1em, y=2.1ex}
\begin{tikzpicture}
	\begin{pgfonlayer}{nodelayer}
		\node [style=none] (0) at (1.25, 0) {};
		\node [style=regb] (1) at (0, 0) {};
		\node [style=none] (2) at (-1.25, 0) {};
	\end{pgfonlayer}
	\begin{pgfonlayer}{edgelayer}
		\draw (2.center) to (1);
		\draw (1) to (0.center);
	\end{pgfonlayer}
\end{tikzpicture}
}
\tikzset{x=1em, y=1.5ex}
\dashv 
\tikzset{x=1em, y=2.1ex}
\begin{tikzpicture}
	\begin{pgfonlayer}{nodelayer}
		\node [style=none] (0) at (1.25, 0) {};
		\node [style=coregwb] (1) at (0, 0) {};
		\node [style=none] (2) at (-1.25, 0) {};
	\end{pgfonlayer}
	\begin{pgfonlayer}{edgelayer}
		\draw (2.center) to (1);
		\draw (1) to (0.center);
	\end{pgfonlayer}
\end{tikzpicture}
}
\tikzset{x=1em, y=1.5ex}
$.

\begin{lem}\label{lemma:general-wb-adjunction}
If $b \from \objr^m \to \objr^n$ is a diagram made entirely from $\Bcomult$ and $\Bcounit$, we have
\[(i)\quad  \stackrel{n\qquad}{\idone}\;\leq 
\tikzset{x=1em, y=2.1ex}
\InputIfFileExists{b-color-transpose-then-b.tikz}{}{\input{./tikz/b-color-transpose-then-b.tikz}}
\tikzset{x=1em, y=1.5ex}
\qquad \text{and}\qquad (ii)\quad  
\tikzset{x=1em, y=2.1ex}
\InputIfFileExists{b-then-b-color-transpose.tikz}{}{\input{./tikz/b-then-b-color-transpose.tikz}}
\tikzset{x=1em, y=1.5ex}
\leq\; \stackrel{m\qquad}{\idone}\]
\end{lem}
\begin{proof}
Consider inequality \emph{(i)}. This can be proven by a straightforward induction on the structure of $b$. We can take care of the base cases with axioms (F5) and (F8). There are three inductive cases to consider, which can be dealt with using these two axioms again:
\begin{itemize}
	\item $b = 
\tikzset{x=1em, y=2.1ex}
\InputIfFileExists{bcomult-b.tikz}{}{\input{./tikz/bcomult-b.tikz}}
\tikzset{x=1em, y=1.5ex}
$. Then $b^\circ\poi b = 
\tikzset{x=1em, y=2.1ex}
\InputIfFileExists{b-circ-wmult-bcomult-b.tikz}{}{\input{./tikz/b-circ-wmult-bcomult-b.tikz}}
\tikzset{x=1em, y=1.5ex}
\mygeq{F5} (b')^\circ\poi b'$.
	\item $b = 
\tikzset{x=1em, y=2.1ex}
\InputIfFileExists{bcounit-over-b.tikz}{}{\input{./tikz/bcounit-over-b.tikz}}
\tikzset{x=1em, y=1.5ex}
$. Then $b^\circ \poi b = 
\tikzset{x=1em, y=2.1ex}
\InputIfFileExists{b-circ-wunit-bcounit-b.tikz}{}{\input{./tikz/b-circ-wunit-bcounit-b.tikz}}
\tikzset{x=1em, y=1.5ex}
\mygeq{F8}(b')^\circ\poi b'$
	\item $b = \sigma \poi b'$ for some permutation $\sigma$. Then $b^\circ\poi b = (b')\circ \poi \sigma^{-1}\poi\sigma \poi b' = (b')^\circ\poi b'$.
\end{itemize}
In all three cases we can conclude that $\stackrel{n\qquad}{\idone}\leq 
\tikzset{x=1em, y=2.1ex}
\InputIfFileExists{b-color-transpose-then-b.tikz}{}{\input{./tikz/b-color-transpose-then-b.tikz}}
\tikzset{x=1em, y=1.5ex}
$ using the induction hypothesis.

Inequality \emph{(ii)} can also be proven by a similar induction.
We can take care of the base cases with axioms (F6) and (F7). As before, there are three inductive cases to consider, which can be dealt with using these two axioms again:
\begin{itemize}
	\item $b = 
\tikzset{x=1em, y=2.1ex}
\InputIfFileExists{bcomult-b.tikz}{}{\input{./tikz/bcomult-b.tikz}}
\tikzset{x=1em, y=1.5ex}
$. Then $b\poi b^\circ = 
\tikzset{x=1em, y=2.1ex}
\InputIfFileExists{b-bcomult-wmult-b-circ.tikz}{}{\input{./tikz/b-bcomult-wmult-b-circ.tikz}}
\tikzset{x=1em, y=1.5ex}
\mygeq{F7} \;\stackrel{m\qquad}{\idone}$.
	\item $b = 
\tikzset{x=1em, y=2.1ex}
\InputIfFileExists{bcounit-over-b.tikz}{}{\input{./tikz/bcounit-over-b.tikz}}
\tikzset{x=1em, y=1.5ex}
$. Then $b \poi b^\circ = 
\tikzset{x=1em, y=2.1ex}
\InputIfFileExists{b-bcounit-wunit-b-circ.tikz}{}{\input{./tikz/b-bcounit-wunit-b-circ.tikz}}
\tikzset{x=1em, y=1.5ex}
\;\myleq{F6} \;\stackrel{m\qquad}{\idone}$
	\item $b = \sigma \poi b'$ for some permutation $\sigma$. Then $b^\circ\poi b = \sigma\poi b \poi b^\circ \poi\sigma^{-1} \myleq{I.H.} \sigma^{-1}\poi\sigma = \;\stackrel{m\qquad}{\idone}$.\qedhere
\end{itemize}
\end{proof}
\begin{lem}\label{lemma:general-bw-adjunction}
If $w\from \objr^m \to \objr^n$, is a diagram made entirely from $\Wcomult$ and $\Wcounit$, we have
\[(i)\quad  \stackrel{n\qquad}{\idone}\;\leq 
\tikzset{x=1em, y=2.1ex}
\InputIfFileExists{w-color-transpose-then-w.tikz}{}{\input{./tikz/w-color-transpose-then-w.tikz}}
\tikzset{x=1em, y=1.5ex}
\qquad \text{and}\qquad (ii)\quad  
\tikzset{x=1em, y=2.1ex}
\InputIfFileExists{w-then-w-color-transpose.tikz}{}{\input{./tikz/w-then-w-color-transpose.tikz}}
\tikzset{x=1em, y=1.5ex}
\leq\; \stackrel{m\qquad}{\idone}\]
\end{lem}
\begin{proof}
The proof is entirely analogous to that of the previous lemma. Inequality \emph{(i)} can be proven in the same way, replacing all uses of (F5) by (F9) and (F8) by (F12).
Similarly, we can prove inequality \emph{(ii)} by replacing all uses of (F6) by (F10) and (F7) by (F11).
\end{proof}
\begin{lem}\label{lemma:general-bb-adjunction}
If $b \from \objr^m \to \objr^n$ is a diagram made entirely from $\Bcomult$ and $\Bcounit$, we have
\[(i)\quad  \stackrel{n\qquad}{\idone}\;\leq 
\tikzset{x=1em, y=2.1ex}
\InputIfFileExists{b-then-b-transpose.tikz}{}{\input{./tikz/b-then-b-transpose.tikz}}
\tikzset{x=1em, y=1.5ex}
\qquad \text{and}\qquad (ii)\quad  
\tikzset{x=1em, y=2.1ex}
\InputIfFileExists{b-transpose-then-b.tikz}{}{\input{./tikz/b-transpose-then-b.tikz}}
\tikzset{x=1em, y=1.5ex}
\leq\; \stackrel{m\qquad}{\idone}\]
\end{lem}
\begin{proof}
The proof is once again entirely analogous to that of the Lemma~\ref{lemma:general-bw-adjunction}, with the converse inequalities (as $\Bcomult$ is left adjoint to $\Bmult$, whereas $\Bcomult$ is right adjoint to $\Wmult$).
We can prove inequality \emph{(i)} as inequality \emph{(ii)} in Lemma~\ref{lemma:general-bw-adjunction} by replacing all uses of (F6) by (F2) and (F7) by (F3).
Inequality \emph{(ii)} can be proven in the same way as inequality \emph{(i)} in Lemma~\ref{lemma:general-bw-adjunction}, replacing all uses of (F5) by (F1) and (F8) by (F4).
\end{proof}
\noindent Recall that if $
\tikzset{x=1em, y=2.1ex}
\begin{tikzpicture}
	\begin{pgfonlayer}{nodelayer}
		\node [style=none] (0) at (1.25, 0) {};
		\node [style=regb] (1) at (0, 0) {};
		\node [style=none] (2) at (-1.25, 0) {};
	\end{pgfonlayer}
	\begin{pgfonlayer}{edgelayer}
		\draw (2.center) to (1);
		\draw (1) to (0.center);
	\end{pgfonlayer}
\end{tikzpicture}
}
\tikzset{x=1em, y=1.5ex}
 = b \poi c$, $
\tikzset{x=1em, y=2.1ex}
\begin{tikzpicture}
	\begin{pgfonlayer}{nodelayer}
		\node [style=none] (0) at (1.25, 0) {};
		\node [style=coregwb] (1) at (0, 0) {};
		\node [style=none] (2) at (-1.25, 0) {};
	\end{pgfonlayer}
	\begin{pgfonlayer}{edgelayer}
		\draw (2.center) to (1);
		\draw (1) to (0.center);
	\end{pgfonlayer}
\end{tikzpicture}
}
\tikzset{x=1em, y=1.5ex}
 := c^\circ\poi b^T$ and  $
\tikzset{x=1em, y=2.1ex}
\begin{tikzpicture}
	\begin{pgfonlayer}{nodelayer}
		\node [style=none] (0) at (1.25, 0) {};
		\node [style=coregbw] (1) at (0, 0) {};
		\node [style=none] (2) at (-1.25, 0) {};
	\end{pgfonlayer}
	\begin{pgfonlayer}{edgelayer}
		\draw (2.center) to (1);
		\draw (1) to (0.center);
	\end{pgfonlayer}
\end{tikzpicture}
}
\tikzset{x=1em, y=1.5ex}
 := c^T\poi b^\circ$.
\begin{lem}\label{lemma:adjoint-bb-wb}
(i) $\stackrel{m\qquad}{\idone} \:\leq\: 
\tikzset{x=1em, y=2.1ex}
\begin{tikzpicture}
	\begin{pgfonlayer}{nodelayer}
		\node [style=none] (0) at (1.25, 0) {};
		\node [style=regb] (1) at (0, 0) {};
		\node [style=none] (2) at (-1.25, 0) {};
	\end{pgfonlayer}
	\begin{pgfonlayer}{edgelayer}
		\draw (2.center) to (1);
		\draw (1) to (0.center);
	\end{pgfonlayer}
\end{tikzpicture}
}
\tikzset{x=1em, y=1.5ex}
\!\!\!\!\! 
\tikzset{x=1em, y=2.1ex}
\begin{tikzpicture}
	\begin{pgfonlayer}{nodelayer}
		\node [style=none] (0) at (1.25, 0) {};
		\node [style=coregwb] (1) at (0, 0) {};
		\node [style=none] (2) at (-1.25, 0) {};
	\end{pgfonlayer}
	\begin{pgfonlayer}{edgelayer}
		\draw (2.center) to (1);
		\draw (1) to (0.center);
	\end{pgfonlayer}
\end{tikzpicture}
}
\tikzset{x=1em, y=1.5ex}
$ and (ii) $   
\tikzset{x=1em, y=2.1ex}
\begin{tikzpicture}
	\begin{pgfonlayer}{nodelayer}
		\node [style=none] (0) at (1.25, 0) {};
		\node [style=coregwb] (1) at (0, 0) {};
		\node [style=none] (2) at (-1.25, 0) {};
	\end{pgfonlayer}
	\begin{pgfonlayer}{edgelayer}
		\draw (2.center) to (1);
		\draw (1) to (0.center);
	\end{pgfonlayer}
\end{tikzpicture}
}
\tikzset{x=1em, y=1.5ex}
\!\!\!\!\! 
\tikzset{x=1em, y=2.1ex}
\begin{tikzpicture}
	\begin{pgfonlayer}{nodelayer}
		\node [style=none] (0) at (1.25, 0) {};
		\node [style=regb] (1) at (0, 0) {};
		\node [style=none] (2) at (-1.25, 0) {};
	\end{pgfonlayer}
	\begin{pgfonlayer}{edgelayer}
		\draw (2.center) to (1);
		\draw (1) to (0.center);
	\end{pgfonlayer}
\end{tikzpicture}
}
\tikzset{x=1em, y=1.5ex}
\:\leq\: \stackrel{n\qquad}{\idone}$
\end{lem}
\begin{proof}
Let $w= c^\circ$ so that $w$ is made entirely from $\Wcomult,\Wcounit$ (and we can apply Lemma~\ref{lemma:general-bw-adjunction} to it) and notice that $w^\circ = (c^\circ)^\circ = c$.

For \emph{(i)}, we have $\stackrel{m\qquad}{\idone}\; \leq b\poi b^T  \leq b\poi w^\circ\poi w \poi b^T := b\poi c\poi c^\circ \poi b^T    =:\; 
\tikzset{x=1em, y=2.1ex}
\begin{tikzpicture}
	\begin{pgfonlayer}{nodelayer}
		\node [style=none] (0) at (1.25, 0) {};
		\node [style=regb] (1) at (0, 0) {};
		\node [style=none] (2) at (-1.25, 0) {};
	\end{pgfonlayer}
	\begin{pgfonlayer}{edgelayer}
		\draw (2.center) to (1);
		\draw (1) to (0.center);
	\end{pgfonlayer}
\end{tikzpicture}
}
\tikzset{x=1em, y=1.5ex}
\!\!\!\!\! 
\tikzset{x=1em, y=2.1ex}
\begin{tikzpicture}
	\begin{pgfonlayer}{nodelayer}
		\node [style=none] (0) at (1.25, 0) {};
		\node [style=coregwb] (1) at (0, 0) {};
		\node [style=none] (2) at (-1.25, 0) {};
	\end{pgfonlayer}
	\begin{pgfonlayer}{edgelayer}
		\draw (2.center) to (1);
		\draw (1) to (0.center);
	\end{pgfonlayer}
\end{tikzpicture}
}
\tikzset{x=1em, y=1.5ex}
$ where the first inequality comes from Lemma~\ref{lemma:general-bb-adjunction} \emph{(i)} and the second from Lemma~\ref{lemma:general-bw-adjunction} \emph{(i)}.

For \emph{(ii)}, we have $
\tikzset{x=1em, y=2.1ex}
\begin{tikzpicture}
	\begin{pgfonlayer}{nodelayer}
		\node [style=none] (0) at (1.25, 0) {};
		\node [style=coregwb] (1) at (0, 0) {};
		\node [style=none] (2) at (-1.25, 0) {};
	\end{pgfonlayer}
	\begin{pgfonlayer}{edgelayer}
		\draw (2.center) to (1);
		\draw (1) to (0.center);
	\end{pgfonlayer}
\end{tikzpicture}
}
\tikzset{x=1em, y=1.5ex}
\!\!\!\!\! 
\tikzset{x=1em, y=2.1ex}
\begin{tikzpicture}
	\begin{pgfonlayer}{nodelayer}
		\node [style=none] (0) at (1.25, 0) {};
		\node [style=regb] (1) at (0, 0) {};
		\node [style=none] (2) at (-1.25, 0) {};
	\end{pgfonlayer}
	\begin{pgfonlayer}{edgelayer}
		\draw (2.center) to (1);
		\draw (1) to (0.center);
	\end{pgfonlayer}
\end{tikzpicture}
}
\tikzset{x=1em, y=1.5ex}
\: := c^\circ\poi b^T\poi b\poi c \leq c^\circ\poi c = w \poi w^\circ \leq \;\stackrel{n\qquad}{\idone}$ where the first inequality comes from Lemma~\ref{lemma:general-bb-adjunction} \emph{(ii)} and the second one from Lemma~\ref{lemma:general-bw-adjunction} \emph{(ii)}.
\end{proof}
\begin{lem}\label{lemma:adjoint-bw-bb}
(i) $\stackrel{n\qquad}{\idone}\;\leq 
\tikzset{x=1em, y=2.1ex}
\begin{tikzpicture}
	\begin{pgfonlayer}{nodelayer}
		\node [style=none] (0) at (1.25, 0) {};
		\node [style=coregbw] (1) at (0, 0) {};
		\node [style=none] (2) at (-1.25, 0) {};
	\end{pgfonlayer}
	\begin{pgfonlayer}{edgelayer}
		\draw (2.center) to (1);
		\draw (1) to (0.center);
	\end{pgfonlayer}
\end{tikzpicture}
}
\tikzset{x=1em, y=1.5ex}
\!\!\!\!\! 
\tikzset{x=1em, y=2.1ex}
\begin{tikzpicture}
	\begin{pgfonlayer}{nodelayer}
		\node [style=none] (0) at (1.25, 0) {};
		\node [style=regb] (1) at (0, 0) {};
		\node [style=none] (2) at (-1.25, 0) {};
	\end{pgfonlayer}
	\begin{pgfonlayer}{edgelayer}
		\draw (2.center) to (1);
		\draw (1) to (0.center);
	\end{pgfonlayer}
\end{tikzpicture}
}
\tikzset{x=1em, y=1.5ex}
$ and (ii) $
\tikzset{x=1em, y=2.1ex}
\begin{tikzpicture}
	\begin{pgfonlayer}{nodelayer}
		\node [style=none] (0) at (1.25, 0) {};
		\node [style=regb] (1) at (0, 0) {};
		\node [style=none] (2) at (-1.25, 0) {};
	\end{pgfonlayer}
	\begin{pgfonlayer}{edgelayer}
		\draw (2.center) to (1);
		\draw (1) to (0.center);
	\end{pgfonlayer}
\end{tikzpicture}
}
\tikzset{x=1em, y=1.5ex}
\!\!\!\!\! 
\tikzset{x=1em, y=2.1ex}
\begin{tikzpicture}
	\begin{pgfonlayer}{nodelayer}
		\node [style=none] (0) at (1.25, 0) {};
		\node [style=coregbw] (1) at (0, 0) {};
		\node [style=none] (2) at (-1.25, 0) {};
	\end{pgfonlayer}
	\begin{pgfonlayer}{edgelayer}
		\draw (2.center) to (1);
		\draw (1) to (0.center);
	\end{pgfonlayer}
\end{tikzpicture}
}
\tikzset{x=1em, y=1.5ex}
\leq \;\stackrel{m\qquad}{\idone}$
\end{lem}
\begin{proof}
Let $b_1=c^T$ and notice that $b_1^T = (c^T)^T = c$.

For \emph{(i)}, we have $\stackrel{n\qquad}{\idone}\; \leq b_1\poi b_1^T := c^T\poi c \leq  c^T\poi b^\circ\poi b\poi c =:\; 
\tikzset{x=1em, y=2.1ex}
\begin{tikzpicture}
	\begin{pgfonlayer}{nodelayer}
		\node [style=none] (0) at (1.25, 0) {};
		\node [style=coregbw] (1) at (0, 0) {};
		\node [style=none] (2) at (-1.25, 0) {};
	\end{pgfonlayer}
	\begin{pgfonlayer}{edgelayer}
		\draw (2.center) to (1);
		\draw (1) to (0.center);
	\end{pgfonlayer}
\end{tikzpicture}
}
\tikzset{x=1em, y=1.5ex}
\!\!\!\!\! 
\tikzset{x=1em, y=2.1ex}
\begin{tikzpicture}
	\begin{pgfonlayer}{nodelayer}
		\node [style=none] (0) at (1.25, 0) {};
		\node [style=regb] (1) at (0, 0) {};
		\node [style=none] (2) at (-1.25, 0) {};
	\end{pgfonlayer}
	\begin{pgfonlayer}{edgelayer}
		\draw (2.center) to (1);
		\draw (1) to (0.center);
	\end{pgfonlayer}
\end{tikzpicture}
}
\tikzset{x=1em, y=1.5ex}
$ where the first inequality comes from Lemma~\ref{lemma:general-bb-adjunction} \emph{(i)} and the second from Lemma~\ref{lemma:general-wb-adjunction} \emph{(i)}.

For \emph{(ii)}, we have $
\tikzset{x=1em, y=2.1ex}
\begin{tikzpicture}
	\begin{pgfonlayer}{nodelayer}
		\node [style=none] (0) at (1.25, 0) {};
		\node [style=regb] (1) at (0, 0) {};
		\node [style=none] (2) at (-1.25, 0) {};
	\end{pgfonlayer}
	\begin{pgfonlayer}{edgelayer}
		\draw (2.center) to (1);
		\draw (1) to (0.center);
	\end{pgfonlayer}
\end{tikzpicture}
}
\tikzset{x=1em, y=1.5ex}
\!\!\!\!\! 
\tikzset{x=1em, y=2.1ex}
\begin{tikzpicture}
	\begin{pgfonlayer}{nodelayer}
		\node [style=none] (0) at (1.25, 0) {};
		\node [style=coregbw] (1) at (0, 0) {};
		\node [style=none] (2) at (-1.25, 0) {};
	\end{pgfonlayer}
	\begin{pgfonlayer}{edgelayer}
		\draw (2.center) to (1);
		\draw (1) to (0.center);
	\end{pgfonlayer}
\end{tikzpicture}
}
\tikzset{x=1em, y=1.5ex}
\: := b\poi c\poi c^T\poi b^\circ =: b\poi b_1^T\poi b_1\poi b^\circ \leq b\poi b^\circ \leq \;\stackrel{n\qquad}{\idone}$ where the first inequality comes from Lemma~\ref{lemma:general-bb-adjunction} \emph{(ii)} and the second one from Lemma~\ref{lemma:general-wb-adjunction} \emph{(ii)}.
\end{proof}

Following Kozen's proof of completeness~\cite{kozen1994completeness} we can model the subset construction algebraically, now with diagrams. We can construct a $(\bullet,\bullet)$-matrix $p_s\from  \objr^{2^s} \to \objr^{s}$ encoding the membership relation of elements of $\{0,\dots,s-1\}$  (in the codomain) to subsets of $\{0,\dots,s-1\}$ (in the domain)---in diagrammatic terms, the $i$-th port on the right is connected to the $j$-th port on the left iff $i\in j\subseteq \{0,\dots,s-1\}$ (where we fix some ordering of the subsets of $\{0,\dots,s-1\}$). For example, $p_2$ is the following diagram
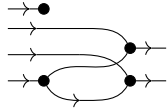
\[
\tikzset{x=1em, y=2.1ex}
\InputIfFileExists{powerset-2.tikz}{}{\input{./tikz/powerset-2.tikz}}
\tikzset{x=1em, y=1.5ex}
\]
where the ports on the right correspond to $0$ and $1$, and those on the left correspond to the subsets $\varnothing$, $\{0\}$, $\{1\}$, and $\{0,1\}$, from top to bottom.

From now on $
\tikzset{x=1em, y=2.1ex}
\begin{tikzpicture}
	\begin{pgfonlayer}{nodelayer}
		\node [style=none] (0) at (1.25, 0) {};
		\node [style=regb] (1) at (0, 0) {};
		\node [style=none] (2) at (-1.25, 0) {};
	\end{pgfonlayer}
	\begin{pgfonlayer}{edgelayer}
		\draw (2.center) to (1);
		\draw (1) to (0.center);
	\end{pgfonlayer}
\end{tikzpicture}
}
\tikzset{x=1em, y=1.5ex}
$ will refer to  $p_S$,
and $
\tikzset{x=1em, y=2.1ex}
\begin{tikzpicture}
	\begin{pgfonlayer}{nodelayer}
		\node [style=none] (0) at (1.25, 0) {};
		\node [style=coregwb] (1) at (0, 0) {};
		\node [style=none] (2) at (-1.25, 0) {};
	\end{pgfonlayer}
	\begin{pgfonlayer}{edgelayer}
		\draw (2.center) to (1);
		\draw (1) to (0.center);
	\end{pgfonlayer}
\end{tikzpicture}
}
\tikzset{x=1em, y=1.5ex}
$, $
\tikzset{x=1em, y=2.1ex}
\begin{tikzpicture}
	\begin{pgfonlayer}{nodelayer}
		\node [style=none] (0) at (1.25, 0) {};
		\node [style=coregbw] (1) at (0, 0) {};
		\node [style=none] (2) at (-1.25, 0) {};
	\end{pgfonlayer}
	\begin{pgfonlayer}{edgelayer}
		\draw (2.center) to (1);
		\draw (1) to (0.center);
	\end{pgfonlayer}
\end{tikzpicture}
}
\tikzset{x=1em, y=1.5ex}
$, to the adjoint $(\circ,\bullet)$-matrices $\objr^{s} \to \objr^{2^s}$ constructed as explained above. We will omit the explicit label $s$ when it can be easily inferred from the context.

The next lemma connects the diagrammatic representation of a given automaton to that of its determinisation. It is simply a reformulation of Kozen's construction in~\cite{kozen1994completeness}. We assume that $d\from \objr^s \to \objr^s$, $f\from \objr^s\to \objr$, and $e\from \objr \to \objr^s$ are matrix-diagrams encoding the transition relation, final and initial states of a given automaton, and $\hat{d}\from \objr^{2^s}\to \objr^{2^s}$, $\hat{f}\from \objr^{2^s}\to \objr$, and $\hat{e}\from  \objr\to \objr^{2^s}$ are matrix-diagrams encoding the transition relation, final and initial state of its determinisation, respectively.
\begin{lem}\label{lemma:powerset-construct}
\[(i) 
\tikzset{x=1em, y=2.1ex}
\InputIfFileExists{d-det-regb.tikz}{}{\input{./tikz/d-det-regb.tikz}}
\tikzset{x=1em, y=1.5ex}
 = 
\tikzset{x=1em, y=2.1ex}
\InputIfFileExists{regb-d.tikz}{}{\input{./tikz/regb-d.tikz}}
\tikzset{x=1em, y=1.5ex}
\]
\[(ii) \dbox{e} =  
\tikzset{x=1em, y=2.1ex}
\InputIfFileExists{hat-e-regb.tikz}{}{\input{./tikz/hat-e-regb.tikz}}
\tikzset{x=1em, y=1.5ex}
 \qquad (iii) \dbox{\hat{f}} =  
\tikzset{x=1em, y=2.1ex}
\InputIfFileExists{regb-f.tikz}{}{\input{./tikz/regb-f.tikz}}
\tikzset{x=1em, y=1.5ex}
\]
\end{lem}
\begin{proof}
By construction of $
\tikzset{x=1em, y=2.1ex}
\begin{tikzpicture}
	\begin{pgfonlayer}{nodelayer}
		\node [style=none] (0) at (1.25, 0) {};
		\node [style=regb] (1) at (0, 0) {};
		\node [style=none] (2) at (-1.25, 0) {};
	\end{pgfonlayer}
	\begin{pgfonlayer}{edgelayer}
		\draw (2.center) to (1);
		\draw (1) to (0.center);
	\end{pgfonlayer}
\end{tikzpicture}
}
\tikzset{x=1em, y=1.5ex}
$. The three claims for the corresponding matrix/vectors can be found in~\cite[Lemma 17]{kozen1994completeness}. The same diagrammatic facts holds because of the completeness of our theory for $(\bullet,\bullet)$-matrix-diagrams (Theorem~\ref{thm:matrix-completeness}).
\end{proof}

\subsection{Determinisation}\label{sec:determinisation}

We are now able to devise a determinisation procedure for representation of automata-diagrams.
One of the payoffs of our approach is that the proof of the following theorem can be carried out purely equationally: the adjunctions we have constructed in the previous subsection make it possible to replace Kozen's use of bisimulation laws in his completeness proof~\cite{kozen1994completeness} by  local diagrammatic rewriting steps.

First, we will need the diagrammatic counterpart of $(xy)^*x = x(yx)^*$, a well-known identity of Kleene algebra. Note that this law holds generally for arbitrary automata-diagrams (and is proved entirely analogously) but we only need it for matrix diagram to show completeness.
\begin{lem}\label{lemma:sliding}
For $x,y$ two matrix-diagrams, we have:
\[
\tikzset{x=1em, y=2.1ex}
\InputIfFileExists{xy-star-x.tikz}{}{\input{./tikz/xy-star-x.tikz}}
\tikzset{x=1em, y=1.5ex}
 \:=\: 
\tikzset{x=1em, y=2.1ex}
\InputIfFileExists{x-yx-star.tikz}{}{\input{./tikz/x-yx-star.tikz}}
\tikzset{x=1em, y=1.5ex}
\]
\end{lem}
\begin{proof}
We only need two successive applications of  matrix distributivity (Lemma~\ref{lem:matrix-copy}):
\begin{align*}

\tikzset{x=1em, y=2.1ex}
\InputIfFileExists{xy-star-x.tikz}{}{\input{./tikz/xy-star-x.tikz}}
\tikzset{x=1em, y=1.5ex}
 \: &\myeq{cpy}\: 
\tikzset{x=1em, y=2.1ex}
\InputIfFileExists{xy-star-x-1.tikz}{}{\input{./tikz/xy-star-x-1.tikz}}
\tikzset{x=1em, y=1.5ex}
 \\
&\myeq{cocpy}\: 
\tikzset{x=1em, y=2.1ex}
\InputIfFileExists{xy-star-x-2.tikz}{}{\input{./tikz/xy-star-x-2.tikz}}
\tikzset{x=1em, y=1.5ex}

\quad \myeq{compact}\: 
\tikzset{x=1em, y=2.1ex}
\InputIfFileExists{x-yx-star.tikz}{}{\input{./tikz/x-yx-star.tikz}}
\tikzset{x=1em, y=1.5ex}

\end{align*}
\end{proof}
\noindent We are now ready to prove a form the \emph{bisimulation} rule of Kozen's completeness proof.
\begin{lem}[Bisimulation]\label{lemma:bisimulation}
If $
\tikzset{x=1em, y=2.1ex}
\InputIfFileExists{d-det-regb.tikz}{}{\input{./tikz/d-det-regb.tikz}}
\tikzset{x=1em, y=1.5ex}
 = 
\tikzset{x=1em, y=2.1ex}
\InputIfFileExists{regb-d.tikz}{}{\input{./tikz/regb-d.tikz}}
\tikzset{x=1em, y=1.5ex}
$ then \[
\tikzset{x=1em, y=2.1ex}
\InputIfFileExists{d-hat-star-regb.tikz}{}{\input{./tikz/d-hat-star-regb.tikz}}
\tikzset{x=1em, y=1.5ex}
 = 
\tikzset{x=1em, y=2.1ex}
\InputIfFileExists{regb-d-star.tikz}{}{\input{./tikz/regb-d-star.tikz}}
\tikzset{x=1em, y=1.5ex}
\]
\end{lem}
\begin{proof}
We prove the forward inclusion first:
\begin{align*}

\tikzset{x=1em, y=2.1ex}
\InputIfFileExists{d-hat-star-regb.tikz}{}{\input{./tikz/d-hat-star-regb.tikz}}
\tikzset{x=1em, y=1.5ex}
  & := 
\tikzset{x=1em, y=2.1ex}
\InputIfFileExists{d-star-regb.tikz}{}{\input{./tikz/d-star-regb.tikz}}
\tikzset{x=1em, y=1.5ex}
&\qquad
\\
& \leq 
\tikzset{x=1em, y=2.1ex}
\InputIfFileExists{regb-coregwb-d-star-regb.tikz}{}{\input{./tikz/regb-coregwb-d-star-regb.tikz}}
\tikzset{x=1em, y=1.5ex}
\tag{Lemma~\ref{lemma:adjoint-bb-wb} (i)}
\\
& = 
\tikzset{x=1em, y=2.1ex}
\InputIfFileExists{regb-coregwb-d-regb-star.tikz}{}{\input{./tikz/regb-coregwb-d-regb-star.tikz}}
\tikzset{x=1em, y=1.5ex}
 \tag{Lemma~\ref{lemma:sliding}}
\\
& = 
\tikzset{x=1em, y=2.1ex}
\InputIfFileExists{regb-coregwb-regb-d-star.tikz}{}{\input{./tikz/regb-coregwb-regb-d-star.tikz}}
\tikzset{x=1em, y=1.5ex}
 \tag{Lemma~\ref{lemma:powerset-construct} (i)}
\\
& \leq 
\tikzset{x=1em, y=2.1ex}
\InputIfFileExists{regb-d-star-loop.tikz}{}{\input{./tikz/regb-d-star-loop.tikz}}
\tikzset{x=1em, y=1.5ex}
 \tag{Lemma~\ref{lemma:adjoint-bb-wb} (ii)}
\\
&=: 
\tikzset{x=1em, y=2.1ex}
\InputIfFileExists{regb-d-star.tikz}{}{\input{./tikz/regb-d-star.tikz}}
\tikzset{x=1em, y=1.5ex}
&
\end{align*}
For the reverse inclusion, we have:
\begin{align*}

\tikzset{x=1em, y=2.1ex}
\InputIfFileExists{regb-d-star.tikz}{}{\input{./tikz/regb-d-star.tikz}}
\tikzset{x=1em, y=1.5ex}
 & :=  
\tikzset{x=1em, y=2.1ex}
\InputIfFileExists{regb-d-star-loop.tikz}{}{\input{./tikz/regb-d-star-loop.tikz}}
\tikzset{x=1em, y=1.5ex}
&\qquad &
\\
& \leq
\tikzset{x=1em, y=2.1ex}
\InputIfFileExists{regb-coregbw-regb-d-star.tikz}{}{\input{./tikz/regb-coregbw-regb-d-star.tikz}}
\tikzset{x=1em, y=1.5ex}
 \tag{Lemma~\ref{lemma:adjoint-bw-bb} (i)}
\\
& =
\tikzset{x=1em, y=2.1ex}
\InputIfFileExists{regb-coregbw-d-regb-star.tikz}{}{\input{./tikz/regb-coregbw-d-regb-star.tikz}}
\tikzset{x=1em, y=1.5ex}
  \tag{Lemma~\ref{lemma:powerset-construct} (i)}
\\
& = 
\tikzset{x=1em, y=2.1ex}
\InputIfFileExists{regb-coregbw-d-star-regb.tikz}{}{\input{./tikz/regb-coregbw-d-star-regb.tikz}}
\tikzset{x=1em, y=1.5ex}
 \tag{Lemma~\ref{lemma:sliding}}
\\
& \leq 
\tikzset{x=1em, y=2.1ex}
\InputIfFileExists{d-star-regb.tikz}{}{\input{./tikz/d-star-regb.tikz}}
\tikzset{x=1em, y=1.5ex}
 \tag{Lemma~\ref{lemma:adjoint-bw-bb}  (ii)}
\\
&=: 
\tikzset{x=1em, y=2.1ex}
\InputIfFileExists{d-hat-star-regb.tikz}{}{\input{./tikz/d-hat-star-regb.tikz}}
\tikzset{x=1em, y=1.5ex}
 & &\qedhere
\end{align*}
\end{proof}
\begin{thm}
Every automaton-diagram is equal to its determinisation.
\end{thm}
\begin{proof}
Given an automata-diagram with representation $(e,d,f)$ let $(\hat e,\hat d,\hat f)$ be its determinisation as defined as above. Then
\begin{align*}

\tikzset{x=1em, y=2.1ex}
\InputIfFileExists{automata-rep-star.tikz}{}{\input{./tikz/automata-rep-star.tikz}}
\tikzset{x=1em, y=1.5ex}
& = 
\tikzset{x=1em, y=2.1ex}
\InputIfFileExists{e-hat-regb-d-star-f.tikz}{}{\input{./tikz/e-hat-regb-d-star-f.tikz}}
\tikzset{x=1em, y=1.5ex}
 \tag{Lemma~\ref{lemma:powerset-construct}(ii)}
\\
& = 
\tikzset{x=1em, y=2.1ex}
\InputIfFileExists{e-hat-d-hat-star-regb-f.tikz}{}{\input{./tikz/e-hat-d-hat-star-regb-f.tikz}}
\tikzset{x=1em, y=1.5ex}
 \tag{Bisimulation}
\\
& = 
\tikzset{x=1em, y=2.1ex}
\InputIfFileExists{determinised-rep.tikz}{}{\input{./tikz/determinised-rep.tikz}}
\tikzset{x=1em, y=1.5ex}
  \tag*{(Lemma~\ref{lemma:powerset-construct}(iii)) \qedhere}
\end{align*}
\end{proof}

Combining the theorem above with the existence of representations (Theorem~\ref{thm:traceform}) yields the result we are after.
\begin{cor}[Determinisation]\label{thm:deterministic-rep}
Any automaton-diagram $\objr\to\objr$ has a deterministic representation.
\end{cor}

\subsection{Minimisation and completeness}%
\label{sec:minimisation}

As explained above, our proof of completeness is a diagrammatic reformulation of Brzozowski's algorithm, which proceeds in four steps: determinise, reverse, determinise, reverse. We already know how to determinise a given diagram. The other three steps are simply a matter of changing our perspective on diagrams, looking at them from right to left, and noticing that all the equations that we needed to determinise them, can be performed in reverse.

We say that a matrix-diagram is \emph{co-deterministic} if the converse of its associated transition relation is deterministic.

\begin{proof}[Proof of Theorem~\ref{thm:1-to-1-completeness} (Completeness)]\label{proof:completeness}
We have a procedure to show that, if $\sem{d}=\sem{d'}$, then there exists a string diagram $c$ in normal form such that  $d=c=d'$. This normal form is the diagrammatic counterpart of the \emph{minimal} automaton associated to $d$ and $d'$. In our setting, it is the deterministic representation of $d$ and $d'$ with the smallest number of states. This is unique because we can obtain from it the corresponding minimal automaton, which is well-known to be unique. First, given any string diagram we can obtain a representation for it by Proposition~\ref{thm:traceform}. Then we obtain a minimal representation by splitting Brzozowski's algorithm in two steps.
\begin{description}
\item[1. Reverse; determinise; reverse] A close look at the determinisation procedure
shows that, at each step, the required equations all hold in reverse, read from right to left instead of left to right. For example, we can replace every instance of~(cpy) with~(co-cpy).
We can thus define, in a completely analogous manner, a co-determinisation procedure which takes care of the first three steps of Brzozowski's algorithm, and obtain a co-deterministic representation for the given diagram.
\item[2. Determinise]  By applying  Corollary~\ref{thm:deterministic-rep}, we can obtain a deterministic representation from the co-deterministic representation of the previous step. The result is the desired minimal deterministic representation and normal form. \qedhere
\end{description}
\end{proof}

\begin{exa}\label{ex:(aa)*(1+a)}
This example treats the diagrammatic equivalent of the regex $(aa)^*(1+a)$ which denotes the same language as $a^*$. This is a simple example of an equivalence that cannot be proven in Kleene algebra without the induction axiom (or some equivalent infinitary axiom scheme encoding induction). We prove the first inclusion below.
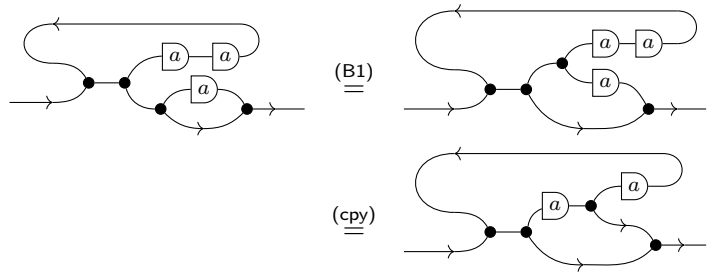
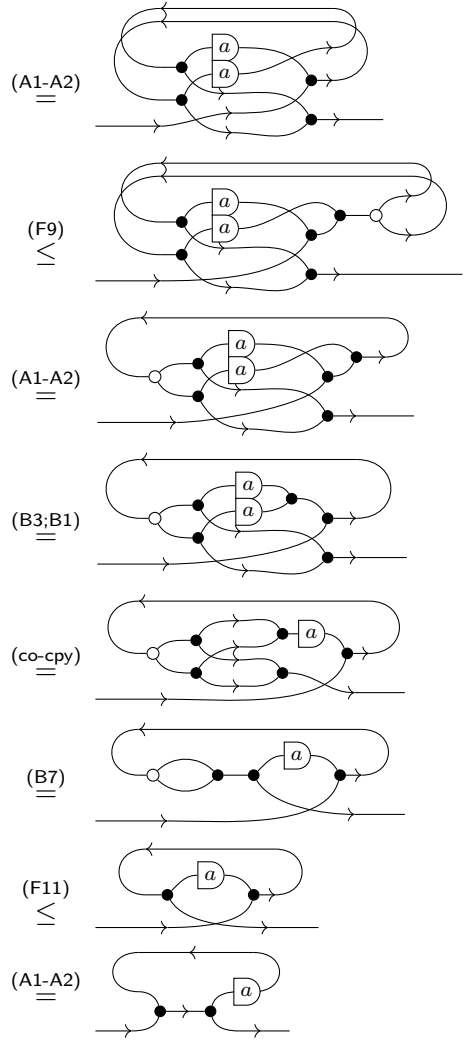
\begin{align*}

\tikzset{x=1em, y=2.1ex}
\InputIfFileExists{ex-aa-star-one+a.tikz}{}{\input{./tikz/ex-aa-star-one+a.tikz}}
\tikzset{x=1em, y=1.5ex}
 \quad &\myeq{B1} \quad 
\tikzset{x=1em, y=2.1ex}
\InputIfFileExists{ex-minimise-1.tikz}{}{\input{./tikz/ex-minimise-1.tikz}}
\tikzset{x=1em, y=1.5ex}
 \\
&\myeq{cpy} \quad 
\tikzset{x=1em, y=2.1ex}
\InputIfFileExists{ex-minimise-2.tikz}{}{\input{./tikz/ex-minimise-2.tikz}}
\tikzset{x=1em, y=1.5ex}
 \\
&\myeq{A1-A2} \quad 
\tikzset{x=1em, y=2.1ex}
\InputIfFileExists{ex-minimise-3.tikz}{}{\input{./tikz/ex-minimise-3.tikz}}
\tikzset{x=1em, y=1.5ex}
 \\
&\myleq{F9} \quad 
\tikzset{x=1em, y=2.1ex}
\InputIfFileExists{ex-minimise-4.tikz}{}{\input{./tikz/ex-minimise-4.tikz}}
\tikzset{x=1em, y=1.5ex}
 \\
&\myeq{A1-A2} \quad 
\tikzset{x=1em, y=2.1ex}
\InputIfFileExists{ex-minimise-5.tikz}{}{\input{./tikz/ex-minimise-5.tikz}}
\tikzset{x=1em, y=1.5ex}
 \\
&\myeq{B3;B1} \quad 
\tikzset{x=1em, y=2.1ex}
\InputIfFileExists{ex-minimise-6.tikz}{}{\input{./tikz/ex-minimise-6.tikz}}
\tikzset{x=1em, y=1.5ex}
 \\
&\myeq{co-cpy} \quad 
\tikzset{x=1em, y=2.1ex}
\InputIfFileExists{ex-minimise-7.tikz}{}{\input{./tikz/ex-minimise-7.tikz}}
\tikzset{x=1em, y=1.5ex}
 \\
&\myeq{B7} \quad 
\tikzset{x=1em, y=2.1ex}
\InputIfFileExists{ex-minimise-8.tikz}{}{\input{./tikz/ex-minimise-8.tikz}}
\tikzset{x=1em, y=1.5ex}
 \\
&\myleq{F11} \quad 
\tikzset{x=1em, y=2.1ex}
\InputIfFileExists{ex-minimise-9.tikz}{}{\input{./tikz/ex-minimise-9.tikz}}
\tikzset{x=1em, y=1.5ex}
 \\
&\myeq{A1-A2} \quad 
\tikzset{x=1em, y=2.1ex}
\InputIfFileExists{a-star.tikz}{}{\input{./tikz/a-star.tikz}}
\tikzset{x=1em, y=1.5ex}
 \\
\end{align*}
The reverse inclusion can be proven similarly, by introducing $\Bcomult$ instead of $\Wcomult$ (this is the easier direction, in the sense that the white generators are not needed): first, we replace use (F1) instead of (F9) to introduce $\Bmult ;\Bcomult$ in the fourth step instead of the $\Bmult ; \Wcomult$ that we introduced above; then, we use (F3) instead of (F11) to turn $\Bcomult ; \Bmult$ into an identity wire, as we did with $\Bcomult ; \Wmult$ in the penultimate step above.

This is an example of an equality that could not be proven in the equational theory of the conference paper~\cite{piedeleu2021string}: the determinisation procedure proposed in the proof of~\cite[Lemma 4]{piedeleu2021string} would fail to identify the two equivalent states (represented by the two loops in  the representation obtained on the third line of the derivation above) and get stuck.
\end{exa}

In Section~\ref{sec:automata-diag}, we explored the correspondence between $\objr\to \objr$ diagrams and regular expressions. In the light of our completeness result, we can revisit this correspondence and extend it to arbitrary left-to-right diagrams $\objr^{m} \to \objr^{n}$. More precisely, we can now prove a \emph{Kleene theorem} for left-to-right diagrams.
We can extend the notion of matrix-diagram (Definition~\ref{def:matrix-diagram}) to that of \emph{regex matrix-diagram}, which is a left to right diagram that factors as a block of $
\tikzset{x=1em, y=2.1ex}
\InputIfFileExists{lr-copy.tikz}{}{\input{./tikz/lr-copy.tikz}}
\tikzset{x=1em, y=1.5ex}
, 
\tikzset{x=1em, y=2.1ex}
\begin{tikzpicture}
	\begin{pgfonlayer}{nodelayer}
		\node [style=black] (37) at (0.75, 0) {};
		\node [style=none] (43) at (0.25, 0) {};
		\node [style=none] (44) at (-0.5, 0) {};
	\end{pgfonlayer}
	\begin{pgfonlayer}{edgelayer}
		\draw (43.center) to (37);
		\draw [->] (44.center) to (43.center);
	\end{pgfonlayer}
\end{tikzpicture}
}
\tikzset{x=1em, y=1.5ex}
$, followed by a block of regex-diagrams and finally, a block of $
\tikzset{x=1em, y=2.1ex}
\InputIfFileExists{lr-merge.tikz}{}{\input{./tikz/lr-merge.tikz}}
\tikzset{x=1em, y=1.5ex}
, 
\tikzset{x=1em, y=2.1ex}
\begin{tikzpicture}
	\begin{pgfonlayer}{nodelayer}
		\node [style=black] (37) at (-0.5, 0) {};
		\node [style=none] (43) at (0.25, 0) {};
		\node [style=none] (44) at (0.75, 0) {};
	\end{pgfonlayer}
	\begin{pgfonlayer}{edgelayer}
		\draw [->] (37) to (43.center);
		\draw (44.center) to (43.center);
	\end{pgfonlayer}
\end{tikzpicture}
}
\tikzset{x=1em, y=1.5ex}
$. This is the diagrammatic counterpart of a matrix with regex coefficients.
The completeness of KDA implies that every left-to-right diagram can be put in this form.
\begin{cor}[Kleene theorem for $\Aut$]\label{thm:matrix-regexp}
Any left-to-right diagram is equal to a regex matrix-diagram.
\end{cor}
\begin{proof}
 Let $d\from \objr^m \to \objr^n$ be a left-to-right diagram. By Corollary~\ref{thm:1-to-1-restrict}, $d$ is fully characterised by its coefficients $d_{ij}$ obtained by discarding all but one of the left ports and all but one of the right ports. According to Proposition~\ref{thm:diagram-nfa}, we can find a representation for each $d_{ij}$ and therefore a regular language $L_{ij}$ recognised by the associated automata. By the standard version of Kleene theorem, we can pick a regex $e_{ij}$ that describes $L_{ij}$. Then, by soundness $\sem{d_{ij}} = \sem{\transreg{e_{ij}}}$ and
by completeness $d_{ij} = \transreg{e_{ij}}$. This shows that $d$ is equal to a diagram that factors as a matrix of regex-diagrams, as we wanted to prove.
\end{proof}
As a result, any given $\objr^{m} \to \objr^{n}$ diagram is fully characterised by an $m\times n$ array of regular languages. Finally, by Theorem~\ref{thm:left-to-right} any given diagram (not necessarily left-to-right) with $m$ inputs and $n$ outputs is fully characterised by an array $m\times n$ array of regular languages and where each of the inputs and outputs is located (on the left or on the right interface).

\section{Discussion}\label{sec:conclusion}

In this paper, we have given a fully diagrammatic treatment of finite-state automata, with a finite equational theory that axiomatises them up to language equivalence. We have seen that this allows us to decompose the regular operations of Kleene algebra, like the star, into more primitive components, resulting in greater modularity. In this section, we compare our contributions with related work, and outline directions for future research.

Traditionally, computer scientists have used \emph{syntax diagrams} (also called \emph{railroad diagrams}) to visualise regular expressions and context-free grammars~\cite{wirth1971programming}. These diagrams resemble ours very closely but have remained mostly informal
More recently, Hinze has treated the single input-output case rigorously as a pedagogical tool to teach the correspondence between finite-state automata and regular expressions~\cite{hinze2019self}. He did not, however, study their equational properties.


Bloom and {\'E}sik's \emph{iteration theories} provide a general categorical setting in which to study the equational properties of iteration for a broad range of structures that appear in programming languages semantics~\cite{bloom1993iteration}. They are cartesian categories equipped with a parameterised fixed-point operation closely related to the feedback notion we have used to represent the Kleene star. However, the monoidal category of interest in this paper is \emph{compact-closed}, a property that is incompatible with the existence of categorical products (any compact-closed category for which the monoidal product is also the categorical product is trivial~\cite{lambek1988introduction}). Nevertheless, the subcategory of left-to-right diagrams (Section~\ref{sec:automata-diag}) is a (matrix) iteration theory~\cite{bloom1993matrix}, a structure that Bloom and {\'E}sik have used to give an (infinitary) axiomatisation of regular languages~\cite{bloom1993equational}.

Similarly, Stefanescu's work on \emph{network algebra} provides a unified algebraic treatment of various types of networks, including finite-state automata~\cite{stefanescu2000network}. In general, network algebras are traced monoidal categories where the product is not necessarily cartesian, and therefore more general than iteration theories. In both settings however, the trace is a global operation, that cannot be decomposed further into simpler components. In our work, on the other hand, the trace can be defined from the compact-closed structure, as was depicted in~\eqref{eq:star-decomposed}.

Note that the compact closed category in this paper can be recovered from the traced monoidal category of left-to-right diagrams, via the \emph{Int construction}~\cite{Joyal_tracedcategories}. Therefore, as far as mathematical expressiveness is concerned, the two approaches are equivalent. However, from a methodological point of view, taking the compact closed structure as primitive allows for improved compositionality, as example~\eqref{ex:decompose-automaton} in the introduction illustrates. Furthermore, the compact closed structure can be finitely presented relative to the theory of symmetric monoidal categories, whereas the trace operation cannot. This matters greatly in this paper, where finding a finite axiomatisation is our main concern.

In all the formalisms we have mentioned, the difficulty typically lies in capturing the behaviour of iteration---whether as the star in Kleene algebra~\cite{kozen1994completeness,bloom1993equational}, or a trace operator~\cite{bloom1993iteration} in iteration theory and network algebra~\cite{stefanescu2000network}. The axioms should be coercive enough to force it to be \emph{the least fixed-point} of the language map $L\mapsto \{\epsilon\} \cup LK$. In Kozen's axiomatisation of Kleene algebra~\cite{kozen1994completeness} for example, this is through (a) the axiom $1+ ee^*\leq e^*$ (star is a fixpoint) and (b) the Horn clause $f+ex \leq x \Rightarrow  e^*f \leq x$ (star is the least fixpoint). In our work, (a) is a consequence of the unfolding of the star into a feedback loop and can be derived from the axioms involving only the automata generators (black nodes); (b) on the other hand does require the existence of their adjoints (white nodes).

Pratt's \emph{action algebras} achieve a similar technical goal: the algebraic theory of action algebra is a finitely-based conservative extension of Kleene algebra~\cite{pratt1990action}. An action algebra is a Kleene algebra and a residuated lattice: it has two additional operations of implication that are adjoint to left/right multiplication. As in our setting, the induction axiom of action algebras can be derived from a finite number of purely equational axioms. However, contrary to the theory of Kleene algebra, equality for action algebras is undecidable---it remains to be seen whether that is also the case for the whole language of this paper, and we leave investigation of its decidability for future work.

In the conference paper~\cite{piedeleu2021string} on which this work is based, we presented a finite equational theory for $\Aut$ only (without the white nodes that we use in this work). Unfortunately, this theory turns out to not be complete. Example~\ref{ex:(aa)*(1+a)} provides a counter-example to the claim of completeness of~\cite{piedeleu2021string}. There are several ways to fix the issue. The first would be to recast the existing infinitary axiomatisations of the matricial iteration theories of regular languages~\cite{bloom1993equational} into our diagrammatic framework. This would simply involve restating the two axiom schemes characterising the behaviour of the Kleene star as equations about feedback loops, leading to an infinitary axiomatisation. While this is certainly feasible, we wanted to achieve a \emph{finitary} presentation, which has led us to the approach developed in the present paper. By extending the syntax, we have been able to exploit additional structure over the set of regular languages to obtain a finite theory.

There is an intriguing parallel between our case study and the positive fragment of relation algebra (also known as allegories~\cite{freyd1990categories}). Indeed, allegories, like Kleene algebra, do not admit a finite axiomatisation~\cite{freyd1990categories}. However, this result holds for standard algebraic theories
. It has been shown recently that a structure equivalent to allegories can be given a finite axiomatisation when formulated in terms of string diagrams in monoidal categories~\cite{bonchi2018graphical}. It seems like the greater generality of the monoidal setting---algebraic theories correspond precisely to the particular case of cartesian monoidal categories~\cite{bonchi2018deconstructing}---allows for simpler axiomatisations in some specific cases. In the future we would like to understand whether this phenomenon, of which now we have two instances, can be understood in a general context.

Lastly, various extensions of Kleene Algebra, such as Concurrent Kleene Algebra (CKA)~\cite{hoare2009concurrent,KappeB0Z18} and NetKAT~\cite{anderson2014netkat}, are increasingly relevant in current research. Enhancing our theory $\eqKa$ to encompass these extensions seems a promising research direction, for two main reasons. First, the two-dimensional nature of string diagrams  has been proven particularly suitable to reason about  concurrency (see e.g.~\cite{BonchiHPSZ19,Bruni2013}), and more generally about resource exchange between processes  (see e.g.~\cite{BonchiSZ17,pqp,jacobs2019causal,baez2015compositional,BPSZ-lics19}). Second, when trying to transfer the good meta-theoretical properties of Kleene Algebra (like completeness and decidability) to extensions such as CKA and NetKAT, the cleanest way to proceed is usually in a modular fashion. The interaction between the new operators of the extension and the Kleene star usually represents the greatest challenge to this methodology. Now, in $\eqKa$, the Kleene star is decomposable into simpler components (see~\eqref{eq:star-decomposed}). We believe this is a particularly favourable starting point to modularise a meta-theoretic study of CKA and NetKAT with string diagrams, taking advantage of the results we presented in this paper for finite-state automata.

In this work, we have left open the question of completeness for the whole language, including the white generators of~\eqref{eq:white-gen}.
In an upcoming paper~\cite{GPZSat2022} we give a partial answer for a restricted fragment: the same theory without the letters (\emph{i.e.} over an empty alphabet), and with the addition of axioms turning the white generators into a Frobenius monoid, is complete for monotone relations between Boolean algebras. Of course, we expect the case of nonempty alphabets to be more complicated. We leave this for further work, starting with the simpler case of a single-letter alphabet.

\section*{Acknowledgments}
  \noindent We wish to thank the anonymous reviewers for their many helpful comments and suggestions, and Sergey Goncharov for pointing out an issue with the previous version of this work. We also acknowledge support from \textsc{epsrc} grant EP/V002376/1.

\bibliographystyle{alphaurl}
\bibliography{refs}

\appendix

\vfill

\section{Background \& Methodology}\label{app:method}
\subsection{Props, String Diagrams, and Symmetric Monoidal Theories}\label{sec:smt}

We build on a line of research that has sought to give a formal treatment of graphical models of computation of varying expressive power within the unifying language of symmetric monoidal categories. More specifically, we rely on the notion of coloured product and permutations category (prop), a mathematical structure which generalises standard multisorted algebraic theories~\cite{bonchi2018deconstructing}. Formally, a \emph{prop} is a strict symmetric monoidal category (SMC) whose objects are lists of a finite set of objects and where the monoidal product $\oplus$ on objects is given by list concatenation. Equivalently, it is a strict SMC whose objects are all monoidal products of a finite number of generating objects. \emph{Prop morphisms} are strict symmetric monoidal functors that act as the identity on objects.

Following an established methodology, we define two props: $\Syn$ and $\Sem$, for the syntax and semantics respectively. To guarantee a compositional interpretation, we require $\sem{\cdot}\from\Syn\to\Sem$, the mapping of terms to their intended semantics, to be a prop morphism.

Typically, the syntactic prop $\Syn$ is freely generated from a \emph{monoidal signature} $\Sigma = (O,M)$: a pair of a finite set of objects $O$ and a set $M$ of arrows $g\from X \to Y$, where $X$ and $Y$ are lists of elements of $G$. In this case, we use the notation $\mathsf{P}_\mathcal{S}$ and $\Syn$ interchangeably. There are two ways of describing the arrows of the prop $\mathsf{P}_\mathcal{S}$ concretely. As terms of $(G^*,G^*)$-sorted syntax whose constants are elements of $\mathcal{S}$ and whose operations are the usual categorical composition $(-);(-)\from \Syn(X,Y)\times \Syn(Y,Z)\to \Syn(X,Z)$ and the monoidal product $(-)\oplus(-)\from \Syn(X_1,Y_1)\times \Syn(X_2,Y_2)\to \Syn(X_1 X_2,Y_1 Y_2)$, quotiented by the laws of SMCs. But this quotient is cumbersome and unintuitive to work with.

This is why, we will prefer a different representation. With their two forms of composition, monoidal categories admit a natural two-dimensional graphical notation of \emph{string diagrams}. The idea is that an arrow $c\from X\to Y$ of $\mathsf{P}_\mathcal{S}$ is better represented as a box with $|X|$ ordered wires labelled by the elements of $X$ on the left and $|Y|$ wires labelled by the elements of $Y$ on the left. We can compose these diagrams in two different ways: horizontally with $;$ by connecting the right wires of one diagram to the left wires of another when the types match, and vertically with $\oplus$ by juxtaposing two diagrams:
\[c\poi d = \quad 
\tikzset{x=1em, y=2.1ex}
\InputIfFileExists{comp-sequential-n.tikz}{}{\input{./tikz/comp-sequential-n.tikz}}
\tikzset{x=1em, y=1.5ex}
 \qquad \qquad d_1\oplus d_2 = \quad 
\tikzset{x=1em, y=2.1ex}
\InputIfFileExists{comp-parallel-n.tikz}{}{\input{./tikz/comp-parallel-n.tikz}}
\tikzset{x=1em, y=1.5ex}
\]

Thus, arrows of $\mathsf{P}_\mathcal{S}$ can be pictured as (directed acyclic) graphs whose nodes are labelled by elements of $\mathcal{S}$ and whose edges are identity $id_a\from a \to a$, denoted as a plain wire $\idone$ for generating object $a$. The symmetry $\mathcal{S}_{a,b}\from a,b \to b,a$ is drawn as a wire crossing $\sym$ which swaps the $a$-and $b$-wires, and the unit for $\oplus$, $id_0\from 0 \to 0$, as the empty diagram $\idzero$ (we use $0$ to denote the empty list). With this representation the laws of SMCs become diagrammatic tautologies.

Once we have defined $\sem{\cdot}\from\Syn\to\Sem$, it is natural to look for equations to reason about semantic equality directly on the diagrams themselves. Given a set of equations $E$, i.e., a set containing pairs of arrows of the same type, we write $\eqE{E}$ for the smallest congruence w.r.t.\ the two composition operations $;$ and $\oplus$.   We say that $\eqE{E}$ is \emph{sound} if $c\eqE{E} d$ implies $\sem{c} = \sem{d}$. It is moreover \emph{complete} when the converse implication also holds. We call a pair $(\mathcal{S}, E)$ a \emph{symmetric monoidal theory} (SMT) and we can form the prop $\mathsf{P}_{\mathcal{S},E}$ obtained by quotienting each homset of $\mathsf{P}_\mathcal{S}$ by $\eqE{E}$. There is then a prop morphism $q\from \mathsf{P}_{\mathcal{S}}\to\mathsf{P}_{\mathcal{S},E}$ witnessing this quotient.

The reader familiar with categorical logic, may find it helpful to know that the concrete description above can be described in more abstract categorical terms, in line with Lawvere's account of algebraic theories~\cite{LawvereOriginalPaper}: signatures can be organised into a category and the free prop $\mathsf{P}_{\mathcal{S}}$ given as a monad structure over this category. Furthermore, the category of props and prop morphisms is equivalent to the category of algebras for this monad. Then, by standard abstract nonsense, the prop $\mathsf{P}_{\mathcal{S},E}$ and the quotient morphism $q$ arise as a coequaliser of free props. A detailed account of this presentation can be found in~\cite[Appendix A.2]{baez2017props}.

\subsection{Ordered Props and Symmetric Monoidal Inequality Theories}\label{sec:smit}

Our semantic prop $\Sem$ often carries additional structure that we wish to lift to the syntax: relations or Boolean profunctors (which are relations satisfying an extra monotony condition) can be ordered by inclusion. The corresponding mathematical structure is that of an \emph{ordered (or order-enriched) prop}, a prop whose homsets are also posets, with composition and monoidal product are monotone maps.

In the same way that props can be presented by  SMTs, an ordered prop can be presented by \emph{symmetric monoidal inequality theory} (SMIT). Formally, the data of a SMIT is the same as that of a SMT:\@ a signature $\mathcal{S}$
and a set $I$ of pairs $c,d\from X\to Y$ of $\mathsf{P}_\mathcal{S}$-arrows  of the same type, that we now read as \emph{inequalities} $c\leq d$.

As for plain props, we can construct an ordered prop from a SMIT by building the free prop $\mathsf{P}_\mathcal{S}$ and passing to a quotient $\mathsf{P}_{\mathcal{S},I}$. First, we build a preorder on each homset by closing $I$ under $\oplus$ and taking the reflexive and transitive closure of the resulting relation. Then, we obtain the free ordered prop $\mathsf{P}_{\mathcal{S},I}$ by quotienting the resulting preorder by imposing anti-symmetry.

An aside, for the reader comfortable with categorical logic: as for props and SMTs, we can give the concrete construction of this section a more abstract formulation, in terms of enriched category theory. The free order-enriched prop could be described as a monad over an order-enriched category of signatures, and the quotient prop $\mathsf{P}_{\mathcal{S},I}$ as a weighted-colimit. We will not need this characterisation here, so leave a detailed account (which we could not find in the literature) for future work.

SMITs subsume SMTs, since every SMT can be presented as a SMIT, by splitting each equation into two inequalities. As a result, in the main text, we only consider SMITs, referring to them simply as \emph{theories}, and their defining inequalities as \emph{axioms}. When referring to a sound and complete theory, we will also use the term \emph{axiomatisation}, as is standard in the literature.

The situation for a sound and complete theory is summarised in the commutative diagram below:
\[
\begin{tikzcd}
\Syn = P_\mathcal{S} \arrow[rd, "q"', two heads] \arrow[rr, "\sem{\cdot}"] &                                 & \Sem \\
                                                                      & {P_{\mathcal{S},I}} \arrow[ru, "s"'] &
\end{tikzcd}
\]
Soundness simply means that $\sem{\cdot}$ factors as $s\circ q$ through $\mathsf{P}_{\mathcal{S},E}$ and completeness means that $s$ is a faithful prop morphism.

\end{document}